\documentclass[a4paper]{article}

\usepackage{a4wide}

\usepackage{url}
\usepackage{cite}

\usepackage{tikz,tikz-cd}\usetikzlibrary{matrix,arrows,decorations.pathmorphing}

\usepackage[breaklinks]{hyperref}

\usepackage[font=small]{caption}

\usepackage{graphics}
\usepackage{mathtools}
\usepackage{amsmath,amscd,amsthm,multirow,bm}

\newtheorem{theorem}{Theorem}[section]
\newtheorem{proposition}[theorem]{Proposition}
\newtheorem{lemma}[theorem]{Lemma}
\newtheorem*{remark}{Remark}
\newtheorem{corollary}[theorem]{Corollary}

\newtheorem{hypothesis}[]{Hypothesis}

\def\ie{{\it i.\,e.}}
\def\eg{{\it e.\,g.}}

\def\Z{\mathbb{Z}} 
\def\R{\mathbb{R}} 
\def\C{\mathbb{C}} 
\def\N{\mathbb{N}} 
\def\Q{\mathbb{Q}}
\def\es{\epsilon_\sigma}
\def\sigt{{\sigma_{\text{\tiny{\!$T$}}}\,}}
\def\qea{{\sqrt{\es}}}
\def\q{\mathfrak{q}}
\def\cK{\mathcal{C}_K}
\DeclareMathOperator{\im}{im}
\def\qd1{ \sqrt{d+1}}
\def\qe1{ \sqrt{e+1}}

\def\qD{{\sqrt{D}}}
\def\zk{{\Z_K}}
\def\zkx{{\Z^\times_K}}

\def\mm{{\mathfrak{m}}}
\def\lra{{\ \xrightarrow{\ \ \ }\ }}
\def\ve{{\varepsilon}}

\def\JJ{\mathfrak{J}}
\def\p{{P}}
\def\HK{H_K} 
\def\h{h_K} 
\def\zl{\Z_L}
\def\j2{{\JJ_2}}
\def\L2{{L_2}}
\def\qu{\sqrt{-u_K}}
\def\LL{\mathbb{L}}
\def\K2{\mathcal{K}_2}

\def\QD{\Q(\qD)}
\def\D{\mathcal{D}}

\newcommand{\cheb}[2]{T_{#1}\left({#2}\right)}
\newcommand{\chsh}[2]{T^*_{#1}\left({#2}\right)}
\newcommand{\chebh}[1]{T_{#1}}
\newcommand{\chshh}[1]{T^*_{#1}}
\newcommand{\ag}[1]{{\zk/{(#1)}}}
\newcommand{\mg}[1]{{\bigl(\ag{#1}\bigr)^\times}}
\newcommand{\ol}[1]{{\overline{{#1}}}}
\newcommand{\leg}[2]{({\scriptstyle{\frac{#1}{#2}}})}
\newcommand{\flfrac}[2]{{\scriptstyle{\left\lfloor \frac{#1}{#2} \right\rfloor }}}
\newcommand{\fl}[1]{{{\left\lfloor {#1} \right\rfloor }}}
\newcommand{\ord}[2]{{\hbox{\rm ord}}_{#1}({#2})}

\DeclareMathOperator{\coker}{coker}

\usepackage{amssymb}
\usepackage{bbm}
\usepackage{epsfig}

\usepackage{xcolor}

\usepackage{comment}

\def\openone{\mathbbm{1}}
\def\C{\mathbb{C}}
\def\Q{\mathbb{Q}}
\def\Z{\mathbb{Z}}

\def\jj{\mathfrak{j}}
\def\mm{\mathfrak{m}}
\def\pp{\mathfrak{p}}

\def\diag{\mathrm{diag}}
\def\bra#1{\langle#1|}
\def\braket#1#2{\langle#1|#2\rangle}
\def\ket#1{|#1\rangle}
\newcommand{\Gal}[2]{{\mathrm{Gal}_{#1/#2}}}
\newcommand{\Tr}[2]{{\mathrm{Tr}_{#1/#2}}}

\def\tr{\mathop{\rm tr}}

\begin{document}

\def\today{March 5, 2024}
\title{SIC-POVMs from Stark Units:\\Dimensions $n^2+3=4p$, $p$ prime}

\author{
  Ingemar Bengtsson${}^1$,
  Markus Grassl${}^{2,3}$,
  Gary McConnell${}^4$\\[2ex]
  {\footnotesize ${}^1$\,Stockholms Universitet, AlbaNova, Fysikum,
    S-106 91 Stockholm, Sweden}\\
  {\footnotesize ${}^2$\,International Centre for Theory of Quantum Technologies,
     University of Gdansk, 80-309 Gdansk, Poland}\\
  {\footnotesize ${}^3$\,Max Planck Institute for the Science of Light, 91058
      Erlangen, Germany}\\
  {\footnotesize ${}^4$\,Controlled Quantum Dynamics Theory Group, Imperial College, London, United Kingdom}
\def\footnotemark{}\thanks{E-mail addresses: \texttt{ingemar@fysik.su.se},
  \texttt{markus.grassl@ug.edu.pl},
  \texttt{g.mcconnell@imperial.ac.uk}}
}

\maketitle

\begin{abstract}
The existence problem for maximal sets of equiangular lines (or SICs)
in complex Hilbert space of dimension $d$ remains largely open. In a
previous publication we gave a conjectural algorithm for how to
construct a SIC if $d = n^2+3 = p$, a prime number. Perhaps the most
surprising number-theoretical aspect of that algorithm is the
appearance of Stark units in a key role: a single Stark unit from a
ray class field extension of a real quadratic field serves as a seed
from which the SIC is constructed. The algorithm can be modified to
apply to all dimensions $d = n^2+3$. Here we focus on the case when $d
= n^2+3 = 4p$, $p$ prime, for two reasons. First, special measures
have to be taken on the Hilbert space side of the problem when the
dimension is even. Second, the degrees of the relevant ray class
fields are `smooth' in a sense that facilitates exact calculations. As
a result the algorithm becomes easier to explain.  We give solutions
for seventeen different dimensions of this form, reaching $d =
39604$. Several improvements relative to our previous publication are
reported, but we cannot offer a proof that the algorithm works for any
dimensions where it has not been tested.
\end{abstract}

\tableofcontents

\section{Introduction}
We define a SIC, in a complex vector space of dimension $d$, as an
orbit under the Weyl--Heisenberg group that forms a maximal set of
equiangular lines. The name is short for the acronym SIC-POVM, spelt
out as symmetric informationally complete positive operator valued
measure \nocite{Zauner,Zauner_English}\cite{Zauner, Renes}.  If SICs
exist, they can be put to use in classical signal processing, quantum
information theory and quantum foundations.  Geometrically the
definition is as simple as it can be: a SIC is a regular simplex of
maximal size in complex projective space. But the existence problem
has turned out to be very hard. Perhaps the most surprising
number-theoretical aspect of the quest to solve it, thus far, is the
appearance of Stark units in a key role. To see why number theory
enters, recall that a single primitive $d$th root of unity---an
arithmetic object---gives the geometry of a regular~$d$-gon.  The
SIC existence problem seemingly has a similar flavour, with Stark
units in certain ray class field extensions of real quadratic fields
playing the role of the roots of unity.  We believe that our
construction brings something new to the original Stark conjectures,
by connecting them to a geometrical problem.

We begin with a summary of known results and conjectures about SICs. The Weyl--Heisenberg group 
\cite{Weyl} is precisely defined in equations \eqref{eq:presentation}--\eqref{eq:displ} below. 
It has generators $X$ and $Z$ obeying $X^d = Z^d = \openone$, and an additional 
generator of its centre.  
A (projective) orbit of length $d^2$ is obtained by choosing a {\it fiducial vector}  
$|\Psi_0 \rangle$, and then acting with the group to obtain the $d^2$ vectors
\begin{equation}
  |\Psi_{i,j}\rangle = X^iZ^j|\Psi_0 \rangle, \qquad 0 \leq i,j < d,
\end{equation}
where we ignore possible overall phase factors. By definition these $d^2$ vectors form 
a SIC if they define complex equiangular lines in the sense that 
\begin{equation}
  |\langle \Psi_0 |X^iZ^j|\Psi_0 \rangle |^2 =
  \begin{cases}
   \hfil 1, & \text{if $(i,j) = (0,0)$}; \\
   \frac{1}{d+1}, &  \text{if $(i,j) \neq (0,0)$.} 
  \end{cases}\label{eq:SIC_conditions}
\end{equation}
Numerical searches have found SICs for all $d \leq 193$.  Those
numerical solutions have given rise to precise conjectures about the
symmetries enjoyed by SICs \cite{Scott, Andrew}.  The unitary
automorphism group of the Weyl--Heisenberg group is known as the {\it
  Clifford group}, and every SIC found so far is invariant under an
element of order $3$ of the latter. This is known as {\it Zauner
  symmetry}. Exact solutions are known for all $d \leq 53$, and for
some higher dimensions.

It is important to notice that the Weyl--Heisenberg group admits a
canonical unitary matrix representation, so that one can meaningfully
talk about the number theoretical properties of the components of the
SIC vectors when expressed in the corresponding basis. It is also
important that the canonical representation uses only roots of unity:
more precisely the entries of the matrices are numbers from the
cyclotomic (`circle-dividing') number field generated by $(2d)$th
roots of unity. The same is true for the Clifford group. The number
field holding the SIC must therefore contain this cyclotomic field as
a subfield. A further remarkable observation was made through an
examination of exact solutions: for every case examined it was found
\cite{AYAZ} that the SIC vectors in dimension $d$ can be constructed
using some abelian extension of the real quadratic number field $\QD$,
where~$D$ is square-free and satisfies
\begin{equation}
  f^2D = (d+1)(d-3) = (d-1)^2 - 4 \label{eq:D}
\end{equation}
for some~$f \in \Z$. Since all abelian extensions of a number field are known through
class field theory, this then led to a very precise conjecture \cite{AFMY}: in every 
dimension $d$ there exists a {\it ray class SIC} that is constructed using the ray class field over the
base field $K = \QD$ with finite part~$\ol{d}$ of the {\it modulus}
defined as $d$ if $d$ is odd, and $2d$ if $d$ is
even. Typically there exist other, unitarily inequivalent SICs as
well, which live in some larger abelian extension of the base field. They will not concern us here 
(but see Ref.~\cite{KoppLagarias} for more). 

We should mention that if we fix the quadratic base field by fixing the
square-free integer $D$ we will find~\cite{AFMY} an infinite sequence of integers
$d_\ell$ that obey equation \eqref{eq:D}. We refer to them as
\emph{dimension towers}. An example for $D = 5$ is given by the sequence
\begin{equation}
  \{ d_\ell\}_{\ell = 1}^\infty = \{ 4,8,19, 48,124, 323, 844, 2208, 5779, 15128, 39604, \dots \}. 
\label{eq:tower} \end{equation}
It is conjectured that 
in every dimension $d_\ell$ there exists a SIC with unitary symmetry of order $3\ell$. A large 
symmetry makes it comparatively easy to find solutions, and for this particular sequence exact 
solutions up to $d = 323$ have been known for some time \cite{Fibonacci}. The dimension towers are of 
considerable interest in themselves, and we will devote Appendix \ref{tours} of this paper to 
outlining some of their features. 

Concerning the ray class fields we first observe that there are
powerful algorithms, implemented in computer algebra packages such as
Magma \cite{Magma}, which allow us to construct them with given $D$
and $d$, at least provided that the cyclic factors of the Galois group
over $K$ are not too large. Let us write~$\zk$ for the ring of
integers of~$K$.  The starting point for the construction is the
multiplicative group of~$\zk/\ol d \zk$, just as the multiplicative
group of~$\Z/\ol d \Z$ is the starting point for the construction of
the cyclotomic field; see Section \ref{glob} for how to continue. A
key fact is that if we have two moduli such that one divides the
other, then the ray class field whose modulus is a divisor will be a
subfield of the ray class field whose modulus it divides.

There are, however, open number theoretical questions lurking
here. For the cyclotomic number fields we are in possession of an
elegant description: a cyclotomic field with modulus $d$ is generated
by a primitive $d$th root of unity, and the $d$th roots of unity can
be obtained by evaluating the analytic function $e(x) = e^{2\pi i x}$
at rational points.  Finding an equally satisfactory description of
the ray class fields that we are interested in here is part of
Hilbert's $12$th problem \cite{Hilbert}, which has remained open for
more than a century.  However, around fifty years ago Stark proposed
that a set of algebraic units can be calculated numerically from the
value at $s=0$ of the first derivatives of an analytic $L$-function
that is associated to the number fields we are interested in
\cite{StarkIII}.  These units, whose existence is one subject of the
famous Stark conjectures, are known as {\it Stark units}. In
favourable circumstances they generate the corresponding number
fields.

How can Stark units be used to construct SICs? The first proposal was
made by Kopp \cite{Kopp}, who constructed SICs from Stark units in
dimensions $5$, $11$, $17$, and $23$. An extension to cover arbitrary
dimensions seems possible. Our proposal is quite different, and is
applicable only to dimensions of the form $d = n^2+3$ \cite{ABGHM}.
On the other hand we exploit some special features of this choice of
dimensions, or equivalently of this choice of modulus for the ray
class fields, which will enable us to reach dimensions much too large
for numerical searches to be feasible.

Let us be clear about what is achieved here: SICs are {\it
  constructed} from Stark units, but both proposals rely on a version
of the unproven Stark conjectures. There is no proof that they always
yield SICs, and at the end it has to be checked that the collections
of vectors that have been constructed do in fact solve the equations
that define a SIC. Hence SIC existence has been {\it proven} only in
those dimensions where they have been explicitly constructed (and this
will remain true at the end of this paper also).

What is special about $d = n^2+3$? Clearly $d-3 = n^2$, so equation \eqref{eq:D} 
tells us that $D$ is the square free part of $d+1$, and 
\begin{equation}
  d + 1 = f^2D
    \quad \Longrightarrow\quad
  d  = (f\sqrt{D} + 1) (f\sqrt{D}-1)
     = 4\times \frac{f\sqrt{D} + 1}{2}\times \frac{f\sqrt{D}-1}{2} .
\label{eq:split} \end{equation}
When $d = 4p$, it is easily checked that all three factors are
algebraic integers in the quadratic field. When $d = p$ we get only
two prime factors.  But in both cases the calculation shows that the
rational prime $p$ does not remain prime in the quadratic field $K =
\QD$. The key idea in Ref.~\cite{ABGHM}, which focused on the case
$d=p$, was to form the ideal $(f\sqrt{D} + 1)$ and use this as the
modulus for a ray class field over $K$. The result is a ray class field with 
degree $(p-1)/3\ell$ over the Hilbert class field $H_K$, where  
$\ell$ is the position of $d$ in the dimension tower. The ray class field with modulus $p$ is the 
compositum of that subfield with the cyclotomic field, and has degree  $(p-1)^2/3\ell$ 
over $H_K$. 

Provided that $d$ is odd this resonates with the conjectured symmetries of the SICs in 
these dimensions. They are special because the Clifford group contains operators of 
order $3\ell$ that are represented as permutation matrices. The conjectures say 
that there exist SIC fiducial vectors that are left invariant by such a permutation matrix, and 
have an anti-unitary symmetry in addition to this. Going through the details 
one finds that such a vector is formed from $(p-1)/3\ell$ distinct numbers, cyclically 
ordered by the Clifford group.  The temptation is to identify these numbers with the orbit 
of a Stark unit in the field with modulus $(f\sqrt{D} + 1)$ under its cyclic Galois group 
over $H_K$. Closer inspection shows that one has to start from the square 
root of a Stark unit. Then the construction can be made to work---at least, it works for 
the thirteen choices of $d = n^2 +3 = p$ that were tested in Ref.~\cite{ABGHM}. 

The construction generalises to all odd dimensions of the form $d =
n^2+3$, although then we have to deal with an entire lattice of ray
class fields with different moduli. But if the dimension is even there
is an immediate obstacle on the Hilbert space side of the problem. In
the standard representation of the Clifford group there is no operator
of order $3$ that is represented as a permutation matrix.  It would
therefore seem that a fiducial vector invariant under such an operator
necessarily involves cyclotomic numbers, and then the above
construction cannot work. This obstacle is completely removed in
Section \ref{sec:Clifford} below. There, a slightly non-standard
representation of the Clifford group is shown to give operators of
order $3\ell$ represented as permutation matrices.  If $d = n^2+3$ is
even then $d$ is divisible by $4$ but not by $8$, and this is one
reason why the present paper is focused on the case $d = n^2+3 =
4p$. It is the conceptually simplest case among the even dimensions.

When $d = 4p$ we have more than one ray class subfield to choose
from. We can use the ray class field whose modulus is the ideal $\pp
=\bigl((f\sqrt{D} + 1)/2\bigr)$, but we can also use the slightly
larger ray class field with modulus $2\pp$. Which of these subfields
should we use to write down the SIC fiducial vector? The answer will
turn out to be that we need both. In fact, before we are done, we will
need the modulus $4\pp$ as well. The degrees of the relevant ray class
fields over $K$ are
\begin{equation}
  \deg(K^{\pp\jj}) = \h\frac{p-1}{3\ell},\qquad
  \deg(K^{2\pp\jj}) = \h\frac{p-1}{\ell} , 
\end{equation}
where $\h$ is the class number of $K$ (the degree of the Hilbert class
field $H_K$ over $K$).  This brings us to another reason why we focus
on $d = n^2+3 = 4p$ here.  For the purpose of doing explicit
calculations in a number field it is convenient to have it expressed
as a tower of field extensions, and then the prime decomposition of
the degree matters. It helps, computationally, if the degree is {\it
  smooth}, in the sense that its prime factors are small relative to
the degree. This motivates a closer look at the factor $p-1$ in the
degrees. We find (for odd $n$) that
\begin{equation}
  d = n^2+3 = 4p
  \quad \Longrightarrow \quad
  p-1 = \frac{n^2-1}{4} = \frac{n-1}{2}\times\frac{n+1}{2}.
\end{equation}
Hence the upper bound for the largest prime factor in $p-1$ grows like
$\sqrt{p}$, while in the $d = p$ case it grows linearly with $p$. This
helpful fact has the computational consequence that we can rely on
exact arithmetic when discussing, for example, the action of the
Galois group. We hope that this will have the effect of making it
easier to follow the logic of this paper, compared to that of
Ref.~\cite{ABGHM}.

We break off this introduction here, and invite the reader to read the
rest of the paper. Section~\ref{sec:Clifford} gives an account of the
representation theory of the Clifford group on which we rely. We give
more details than usual because we handle the dimension-four factor of
the Hilbert space in an unusual way. Section \ref{sec:ansatz} gives a
first version of an Ansatz for a SIC fiducial vector in dimension $d =
n^2+3 = 4p$. Section \ref{sec:numbertheory} gives the number
theoretical results that enable us to make this Ansatz precise. What
is new in relation to~\cite{ABGHM} is that more than one ray class
field is involved, and that their moduli involve powers of
$2$. Section \ref{sec:results} contains our main result: a precise
version of our algorithm for constructing SICs in these dimensions. It
also gives some details about the dimensions where we have
successfully applied it.  Section \ref{sec:example} gives a worked
example for $d = 52$, which is small enough that we can give all the
calculations in detail. Section \ref{sec:twelve} explains why
dimension $12$ (a dimension divisible by $3$) is special, and Section
\ref{sec:overlaps} contains some useful observations concerning
overlap phases and Stark units. Finally, Section \ref{sec:conclusions}
consists of our conclusions as well as an outlook. Appendix
\ref{tours} gives some new results about dimension towers which apply
irrespective of the dimension; Appendix \ref{sec:AppendixB} places
the behaviour of the primes above~$2$ into the context of the somewhat
striking properties of the geometric scaling factor~$\xi =
\sqrt{-2-\qd1}$; Appendix \ref{sec:phase_ambiguity} gives some 
additional details concerning the representation of the groups; and 
Appendix \Ref{sec:Gik} discusses alternative strategies for exact 
verification of the SIC property. 

We have not been able to prove that the algorithm that we propose
works in all dimensions of the form $d = n^2+3$, but it does work in
every case that we have tested. This includes all $n \leq 53$ as well
as some higher dimensional cases. In this paper we will prove that the
construction works for seventeen different dimensions of the form $d
=n^2 + 3 = 4p$, including $d = 39604$.  It relies on the paradigm of
the Stark conjectures in order to give us the units which go into the
fiducial vectors, but the truth of the Stark conjectures as such is
not directly relevant to it.

\section{How to represent the Clifford group}\label{sec:Clifford}
To every finite dimensional Hilbert space $\C^d$ we can associate a
discrete Heisenberg group $H(d)$ known as a Weyl--Heisenberg group, as
well as its automorphism group with minimal centre within the unitary
group $U(d)$.  The latter is known as its Clifford group. These groups
are tied to dimension $d$ in the sense that the Weyl--Heisenberg group
admits faithful unitary irreducible representations only in dimension
$d$, and the Clifford group has a (projective) representation in dimension $d$
as well.  An interesting fact is that when the dimension is a
composite number both groups can be treated as direct products of the
corresponding groups in the factors, provided the factors are of
relatively prime dimensions. This means that we can confine our
discussion to prime power dimensions. We will restrict ourselves
further here, because the example we are interested in is $H(4p) =
H(4)\times H(p)$ where $p$ is an odd prime equal to $1$ modulo $3$.

Our goal is to construct SICs that are (projective) orbits under the
Weyl--Heisenberg group, and our focus is on the number theoretical
properties of the lines that form the SIC.  When representing a group
by unitary matrices one is forced to make a number of arbitrary
choices; notably one has to choose an orthonormal basis and make a
decision  concerning the phase factors of the vectors in that
basis.  Unfortunate choices will completely obscure the number
theoretical properties of the lines. In most of the literature the
choices originally made by Weyl are followed. We will indeed use this
{\it standard representation} when representing $H(p)$, but not when
representing $H(4)$. To explain why, we first remind the reader about
Weyl's choices \cite{Weyl}.

The group $H(d)$ can be presented using three generators $X$, $Z$, and
$\omega$. We impose the condition that $\omega$ commutes with $X$ and
$Z$, and that
\begin{equation}
  X^d = Z^d = \omega^d = \openone \qquad\text{and}\qquad ZX = \omega XZ. \label{eq:presentation}
\end{equation}
If the dimension $d$ is even, it turns out to be a good idea to extend 
the centre of the group \cite{Marcus} by defining a generator $\tau$ such that 
\begin{equation}
  \tau^2 = \omega.
\end{equation} 
For an irreducible unitary representation, Schur's lemma implies that
$\omega$ is represented by a primitive root of unity times the unit
matrix.  Our first (innocuous) choice is to set
\begin{equation}
  \omega = e^{\frac{2\pi i}{d}}\openone, \qquad
  \tau = - e^{\frac{\pi i}{d}}\openone,
\end{equation}
where throughout the following the notation~$\openone_n$ 
will denote the identity matrix operator on dimension~$n$, omitting the~$n$ where it is clear from the context. 
The sign is introduced so that we obtain the extra relation $\tau = 
\omega^{(d+1)/2}$ if $d$ is odd. (In the following 
we often use the notation $\tau = - e^{\pi i/d}$, and similarly for $\omega$. 
This should not cause confusion). If we introduce 
\begin{equation}
  \bar{d} =
  \begin{cases}
    \hfil  d, & \text{if $d$ is odd;} \\ 
    \hfil 2d, & \text{if $d$ is even,}
  \end{cases}\label{eq:dbar}
\end{equation} 
we can state that $\tau$ is represented by a primitive $\bar{d}$th 
root of unity. We remind the reader that it is usual to define the 
{\it displacement operators} 
\begin{equation}
  D_{i,j} = \tau^{ij}X^iZ^j. \label{eq:displ}
\end{equation}
Up to signs there are $d^2$ displacement operators, and they form a
unitary operator basis in $\C^d$.

It follows that the matrices representing the group necessarily 
include entries lying in the cyclotomic field $\Q(\tau)$. The 
aim is to show that no further extension of the rational numbers is needed 
in order to represent the entire group. Note that a cyclotomic number 
field generated by an $n$th root of unity $\omega_n$ necessarily contains 
the number $-\omega_n$, which is a $(2n)$th root of unity when $n$ is odd. Hence 
every cyclotomic field is of the form $\Q(\omega_{2d})$ for some $d$. 
Precisely because we decided to extend the centre of the Weyl--Heisenberg groups 
when $d$ is even, we can state that we use $\Q(\omega_{2d})$ when 
representing the Weyl--Heisenberg group in a Hilbert space of dimension $d$. 

Having chosen a primitive root of unity the next step is to choose an
orthonormal basis. The standard choice is to use the eigenbasis of the
unitary matrix representing $Z$ for this purpose. The defining
relations \eqref{eq:presentation} imply that its eigenvalues are $d$th
roots of unity, and Weyl went on to show that no repeated eigenvalues
occur. We still have to order the eigenvectors somehow, but this is an
innocuous choice. Thus we have determined that
\begin{equation}
  Z = \left( \begin{array}{cccc}
    1 & 0 & 0 & 0  \\
    0 & \omega & 0 & 0 \\
    0 & 0 & \omega^2 & 0 \\
    0 & 0 & 0 & \omega^3
  \end{array}\right).
\end{equation}
Here we assumed $d = 4$, but the generalisation to arbitrary $d$ should 
be obvious. 

There is one more choice to be made, which is non-trivial in principle. With the 
choices thus far, the defining relations \eqref{eq:presentation} imply 

\begin{equation}
  X = \left( \begin{array}{cccc}
    0 & 0 & 0 & a_1 \\
    a_2 & 0 & 0 & 0 \\ 
    0 & a_3 & 0 & 0 \\
    0 & 0 & a_4 & 0
  \end{array}\right) , 
\end{equation}
where $a_1, a_2, a_3, a_4$ are phase factors obeying $a_1a_2a_3a_4 =
1$.  Weyl chose phase factors in front of the basis vectors ensuring
that $a_1 = a_2= a_3 = a_4 = 1$. This is the standard representation
of the Weyl--Heisenberg group. It is canonical when $d = p$ is prime,
but not when $d = 4$ as we will see.

We now move on to the Clifford group, which is represented by unitary
matrices $U$ that permute the Weyl--Heisenberg group under
conjugation, \ie,
\begin{equation}
  UH(d)U^{-1} = H(d). \label{Clifforddef}
\end{equation}
We can make the restriction that the unitary matrices have determinant
equal to $\pm 1$. If the dimension is composite with relatively prime
factors, and if we ignore the matrix $-\openone$, it follows that the
Clifford group splits as a direct product. The quotient of the
Clifford group by the Weyl--Heisenberg group is isomorphic to the {\it
  symplectic group} $Sp(2, \Z_d)$.  We recall that in this two-dimensional setting, in fact
 $$Sp(2, \Z_d) \cong SL(2, \Z_d),$$
 the special linear group, and we shall use this identification, as well
 as the notation~$SL(2)$, without comment from now onwards. 
A projective representation of
this symplectic group is determined by the representation of $H(d)$ up
to overall phase factors. This means that to every $SL(2)$-matrix $F$
with entries that are integers modulo $\bar{d}$ we can associate a
unitary matrix $U_F$ in the Clifford group: here we choose $\bar{d}$
defined in eq.~\eqref{eq:dbar} instead of $d$ to keep track of phase
factors. To be precise, their action on the displacement operators is
\begin{equation}
  F = \left( \begin{array}{cc}
    \alpha & \beta \\
    \gamma & \delta
  \end{array}\right)_{\bar{d}}
  \qquad \Longrightarrow \qquad
  U_FD_{i,j}U_F^{-1} = D_{\alpha i + \beta j,\gamma i + \delta j},\label{eq:action_Clifford}
\end{equation}
where we used a subscript to indicate that the matrix elements are integers 
modulo $\bar{d}$. The unitary matrices $U_F$ are determined by this requirement, up 
to overall phase factors. When the integer $\beta$ admits an inverse modulo 
$\bar{d}$ the explicit formula for the matrix elements $(U_F)_{r,s}$ of $U_F$ is 
\begin{equation}
  F = \left( \begin{array}{cc}
    \alpha & \beta \\
    \gamma & \delta
  \end{array} \right)_{\bar{d}} \qquad \leftrightarrow \qquad
  (U_F)_{r,s} = \frac{e^{i\phi}}{\sqrt{d}} \tau^{\beta^{-1}(\delta r^2 - \gamma rs + \alpha s^2)}, 
\end{equation}
where $r,s$ run from $0$ to $d-1$.  At the other extreme, when the
symplectic matrix $F$ is diagonal, we obtain a permutation matrix with
matrix elements given by
\begin{equation}
  F = \left( \begin{array}{cc}
    \delta^{-1} & 0 \\
    0 & \delta
  \end{array}\right)_{\bar{d}}
  \qquad \leftrightarrow \qquad
  (U_F)_{r,s} = \pmb{\delta}_{\delta r, s}, \label{eq:diagonal}
\end{equation}
where the bold $\pmb{\delta}$ denotes the Kronecker delta. 
Full details are given in \cite{Marcus}. 

Here we wish to stress two points. First, the phase factor $e^{i\phi}$
can be chosen so that the entries in the representational matrices belong
to the cyclotomic field $\Q(\tau )$. Second, we represented the
Weyl--Heisenberg group by {\it monomial} matrices, that is to say by
matrices that contain only one non-zero element in each row and each
column. The matrices representing the symplectic group are not
monomial in general; but some of them are. Indeed operators that
permute the operators in the cyclic subgroup generated by $Z$ will
permute their joint eigenvectors as well, possibly up to adding phase
factors.  Hence, relative to a basis spanned by these eigenvectors,
such Clifford group elements will be given by monomial matrices, \ie,
by matrices that are permutation matrices possibly with their non-zero
elements being replaced by phase factors.

In the SIC problem we are interested in Clifford unitaries
$U_\mathcal{Z}$ such that $U_\mathcal{Z}^3 = \openone$, because it
seems that every SIC vector has a symmetry of order three
\cite{Zauner}, \cite{Scott}, \cite{Marcus}. Here $\mathcal{Z}$ is a
symplectic matrix of order three and trace $-1$. There are many such
matrices.  A choice that has become standard \cite{Marcus, Scott} is
\begin{equation}
  \mathcal{Z} = \left( \begin{array}{rr}
    0 & -1 \\
    1 & - 1
  \end{array}\right)_{\bar{d}}. \label{eq:standardzauner}
\end{equation}
Such Clifford unitaries are known as Zauner unitaries, because Zauner was 
the first to realise their importance \cite{Zauner}. Indeed,
according to Zauner's conjecture, for every dimension there is a
fiducial SIC vector $\ket{\Psi_0}$ such that with a suitable choice of the
phase factor in $U_\mathcal{Z}$ it is the case that
\begin{equation}
  U_\mathcal{Z}\ket{\Psi_0} =\ket{\Psi_0}, \qquad
  U_\mathcal{Z}^3 =\openone. \label{sicsymmetri}
 \end{equation}
This relation will be simplified considerably if we can choose a
matrix $U_\mathcal{Z}$ that is a permutation matrix of order $3$. If
so, we can hope to write down a fiducial vector using a number field
that does not contain complex roots of unity at all.

If the dimension $d$ is prime, one can show that the symplectic
group~$SL(2)$ contains a unique conjugacy class of order $3$ elements
\cite{Flammia}.  The question is whether this conjugacy class contains
a representative that is represented by a permutation matrix.  What we
need is a diagonal symplectic matrix of order $3$, as in equation
\eqref{eq:diagonal}, but now with the added requirement that $\delta^3
= 1$ modulo $d$. This has solutions if $d = p\equiv 1 \bmod 3$, but
not if $d = p\equiv 2\bmod 3$ \cite{Marcus}. There is another way to
understand this. An operator will be represented by a permutation
matrix if it permutes the vectors that form the basis. In the standard
representation this means that it must permute elements of the cyclic
subgroup generated by $Z$ among themselves. If $d = p$ the
Weyl--Heisenberg group contains $p+1$ cyclic subgroups of order $p$,
generated by $Z, X, XZ, \dots , XZ^{p-1}$ and (pairwise) having only
the unit element in common. A Clifford group element of order $3$ will
collect some of these subgroups into triplets, but if $p\equiv 1\bmod
3$ there must be a pair of subgroups ``left over'', and the order $3$
operator will indeed permute their elements among themselves. For $d =
4$, modulo the centre, there are six partially intersecting cyclic
subgroups of order $4$, and we cannot obtain a monomial
$U_\mathcal{Z}$ if we stay in the standard representation.

Fortunately an alternative representation is available whenever $d$ is
a square \cite{monomial}. We give the details for $d = 4$. When choosing a
representation it is natural to select a maximal commuting set of
operators and let their joint eigenvectors serve as the basis.  Hence
the commuting operators are represented by diagonal matrices. By
inspection of Figure \ref{fig:blomman} we see that $H(4)$ contains a
distinguished abelian subgroup consisting of order-two elements,
namely $\{ \openone, X^2,Z^2,X^2Z^2\}$. It is easy to see that the
defining relations \eqref{eq:presentation} imply that they mutually
commute. How is their joint eigenbasis related to the standard basis?
To see this we exhibit the Hilbert space as a tensor product $\C^4 =
\C^2\otimes \C^2$ by introducing a product basis
\begin{equation}
  |0_4\rangle = |0_2\rangle \otimes |0_2\rangle , \ 
  |1_4\rangle = |0_2\rangle \otimes |1_2\rangle , \
  |2_4\rangle = |1_2\rangle \otimes |0_2\rangle , \ 
  |3_4\rangle = |1_2\rangle \otimes |1_2\rangle.
\end{equation}
It is then seen that in the standard representation with respect to
this product basis: 
\begin{equation} D_{2,0} = X^2 =
  \sigma_x\otimes \openone_2, \qquad
  D_{0,2} = Z^2 = \openone_2\otimes \sigma_z, \qquad
  D_{2,2} = - X^2Z^2 = - \sigma_x\otimes \sigma_z,\label{eq:diag_ops}
\end{equation}
where $\sigma_x$, $\sigma_z$ are the Pauli matrices. Thus we wish to 
diagonalise $\sigma_x$ in the first factor without changing the already diagonal 
$\sigma_z$ in the second. For this purpose we introduce the two dimensional discrete Fourier matrix 
$F_2  = \frac{1}{\sqrt{2}}\left( \begin{array}{cr} 1 & 1 \\ 1 & -1 \end{array} \right)$ and note that 
\begin{equation}
  F_2\sigma_xF_2^{-1}
  = \frac{1}{2}\left( \begin{array}{cr} 1 & 1 \\ 1 & - 1 \end{array} \right)
    \left( \begin{array}{cc} 0 & 1 \\ 1 & 0 \end{array} \right)
     \left( \begin{array}{cr} 1 & 1 \\ 1 & - 1 \end{array} \right)
  = \left( \begin{array}{cr} 1 & 0 \\ 0 & -1 \end{array} \right)
  = \sigma_z. \label{Ftensor1}
\end{equation}
Hence we can go from the standard representation to a representation
where $X^2,Z^2, X^2Z^2$ are diagonal by applying the unitary operator
$F_2\otimes\openone_2$.  Since our point of view requires us to
carefully notice any number theoretical complication that may arise it
is comforting to note that here there are none---equation
\eqref{Ftensor1} involves only rational numbers.

\begin{figure}[hbt]
        \centerline{ \hbox{
                \epsfig{figure=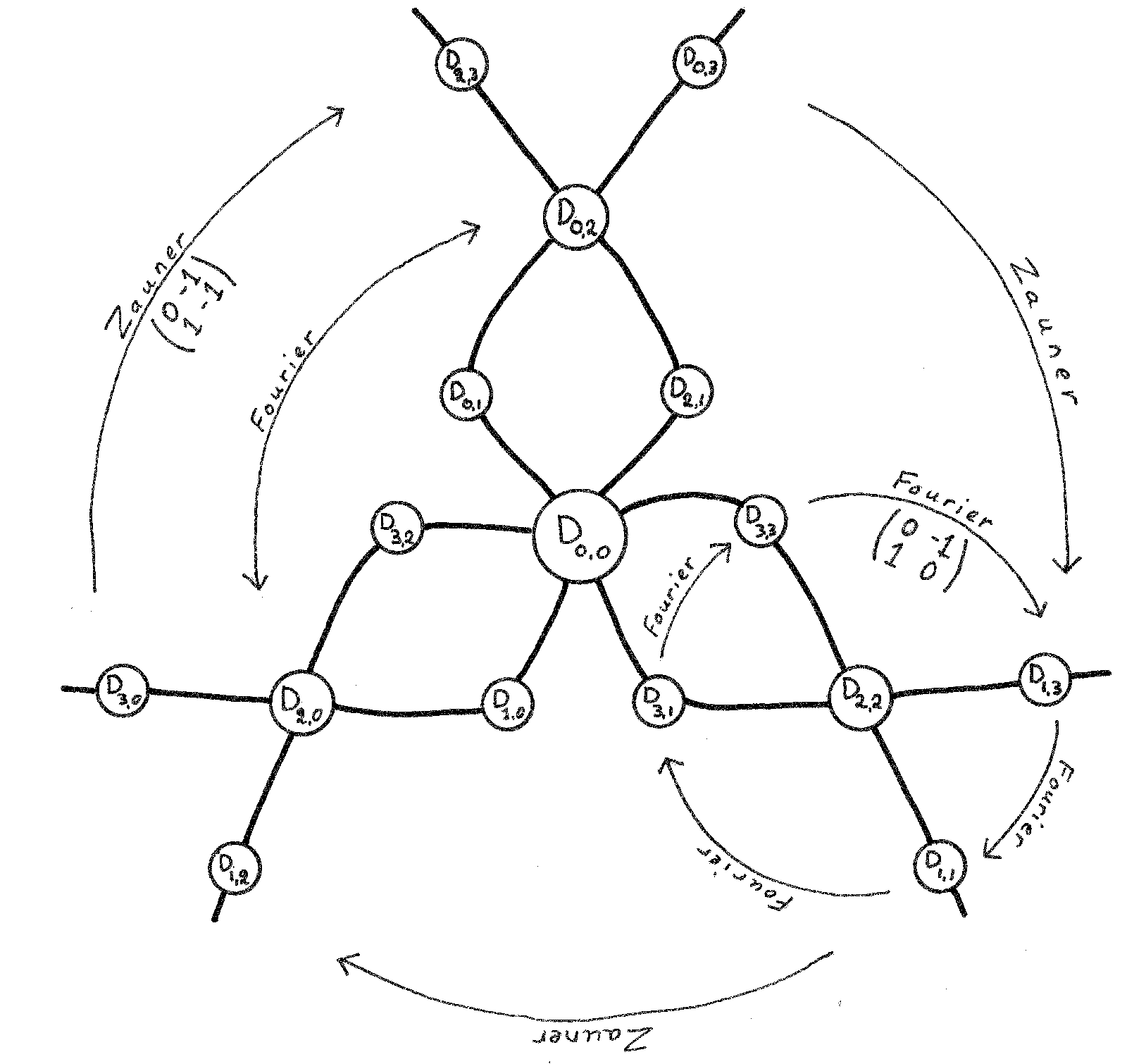,width=75mm}}}
        \caption{The Weyl--Heisenberg and Clifford groups for
          $d=4$. We show the sixteen displacement operators, the six
          cyclic subgroups generated by $D_{0,1}$, $D_{2,1}$,
          $D_{3,3}$, $D_{3,1}$, $D_{1,0}$, and $D_{3,2}$, as well as
          the action of two different symplectic unitaries, namely the
          Fourier matrix and the Zauner matrix
          $U_{\mathcal{Z}}$.\label{fig:blomman}}
 \end{figure}

The effect this basis change has on the Clifford group is dramatic.
Since the basis is defined using a maximal abelian subgroup containing
all the order-two elements in $H(4)$, and since the Clifford group
must permute the order-two elements among themselves, the effect of
the Clifford group on the new basis vectors is simply to permute them
and possibly to multiply them with phase factors. Hence the entire
Clifford group is now given by monomial matrices, which is why the
representation we have arrived at is known as the {\it monomial
  representation} \cite{monomial}.

We still have to deal with possible phase factors in front of the
basis vectors. Recall that Weyl chose them with a view to make $X$ a
real matrix. This choice is no longer natural. But the phase factors
are (almost) determined by the problem we are interested in. Our aim
is to choose the basis in such a way that a symplectic unitary
$U_\mathcal{Z}$ of order $3$, appearing in equation
\eqref{sicsymmetri}, becomes a permutation matrix. We begin by making
the standard choice given in eq.~\eqref{eq:standardzauner} with
$\bar{d} = 2d = 8$, express the symplectic unitary $U_\mathcal{Z}$ in
the standard representation \cite{Marcus}, and then perform the
transformation to the monomial basis. As a result
\begin{equation} U_\mathcal{Z} =
  \frac{1}{2}\left( \begin{array}{cccc}
    \tau^5 & \tau^5 & \tau^5 & \tau^5 \\ 
    \tau^6 & 1 & \tau^2 & \tau^4 \\
    \tau & \tau^5 & \tau & \tau^5 \\
    \tau^6 & \tau^4 & \tau^2 & 1
  \end{array} \right)
  \qquad \rightarrow \qquad
  U_\mathcal{Z} =  \left( \begin{array}{cccc}
    0 & \tau^5 & 0 & 0 \\
    0 & 0 & \tau^6 & 0 \\
    \tau^5 & 0 & 0 & 0 \\
    0 & 0 & 0 & 1
  \end{array} \right). 
\end{equation}
We will insist that $U_\mathcal{Z}$ be represented by a permutation matrix. 
We therefore change the phase factors of the basis vectors by applying a unitary 
transformation effected by the diagonal matrix
\begin{equation}
  T_1 = \diag(1,\tau^5, \tau^3, 1). \label{eq:faser}
\end{equation}
The overall transformation is
\begin{equation}
  T = T_1\cdot(F_2\otimes\openone_2)
  =\frac{1}{\sqrt{2}}
  \begin{pmatrix}
    1 & 0 & 1 & 0\\
    0 &\tau^5 & 0 &\tau^5 \\
    \tau^3 & 0 &\tau^7 & 0 \\
    0 & 1 & 0 & -1
  \end{pmatrix}. \label{eq:basis_change}
\end{equation}
The diagonal matrices representing $X^2$, $Z^2$, $X^2Z^2$ are
unaffected by this, but $U_\mathcal{Z}$ takes the form
\begin{equation}
  U_\mathcal{Z} =  \left( \begin{array}{cccc}
    0 & 1 & 0 & 0 \\
    0 & 0 & 1 & 0 \\
    1 & 0 & 0 & 0 \\
    0 & 0 & 0 & 1
  \end{array} \right).\label{eq:UZauner}
\end{equation}
In this basis, the generators of the Weyl--Heisenberg group are
represented as
\begin{alignat}{5}
  X=D^{(4)}_{1,0}&{} = \left(\begin{array}{cccc}
    0 & \tau^3 & 0 & 0 \\
    \tau^5 & 0 & 0 & 0 \\
    0 & 0 & 0 & \tau^3\\
    0 & 0 & \tau & 0
  \end{array}\right)
 &&{}= \tau\left(\begin{array}{cccc}
    0 & i & 0 & 0 \\
    -1 & 0 & 0 & 0 \\
    0 & 0 & 0 & i\\
    0 & 0 & 1 & 0
  \end{array}\right)\label{eq:X_adapted}\\
\noalign{and}
  Z=D^{(4)}_{0,1}&{} = \left(\begin{array}{cccc}
    0 & 0 & \tau^5 &0 \\
    0 & 0 & 0 & \tau^7 \\
    \tau^3 & 0 & 0 & 0 \\
    0 & \tau^5 & 0 & 0 \\
  \end{array}\right)
  &&{}=\tau\left(\begin{array}{cc}
    0 & -1\\
    i & 0
  \end{array}\right)\otimes
  \left(\begin{array}{cc}
    1 & 0\\
    0 & i
  \end{array}\right),\label{eq:Z_adapted}
\end{alignat}
where $i=\tau^2$, $i^2=-1$. Note that the eighth root of unity $\tau = \frac{1+i}{\sqrt{2}}$
is a common scalar factor; so up to this phase factor, the
matrices can be written using only fourth roots of unity.

We give the representation of two
other extended Clifford group elements that will figure in our
constructions, namely those corresponding to
\begin{equation}
  P = \left( \begin{array}{rr}
    - 1 & 0 \\
    0 & -1
  \end{array}\right)_8
  \qquad\text{and}\qquad
  J = \left( \begin{array}{rr}
    1 & 0 \\
    0 & -1
  \end{array}\right)_8.
\end{equation}
Recall that the subscript~$8$ denotes arithmetic modulo~$8$. 
The matrix $J$ is anti-symplectic (has determinant $-1$) and is represented 
by an anti-unitary operator \cite{Marcus}. In the standard representation $J$ is 
represented by pure complex conjugation on each entry of the target vector, 
with respect to a fixed embedding of our coefficient field. 
In the version of 
the monomial representation that we use here this is going to be slightly more complicated 
due to the phase factors we introduced. We find 
\begin{equation}
  U_P = \left( \begin{array}{cccc}
    1 & 0 & 0 & 0 \\
    0 & 1 & 0 & 0 \\
    0 & 0 & 1 & 0 \\
    0 & 0 & 0 & -1
  \end{array} \right)
  \qquad \text{and} \qquad
  A_J\ket{\Psi} = \left( \begin{array}{cccc}
    1 & 0 & 0 & 0 \\
    0 & i & 0 & 0 \\
    0 & 0 & -i & 0 \\
    0 & 0 & 0 & 1
  \end{array} \right)\ket{\Psi}^*\label{UP4}
\end{equation} 
The unitary operator $U_P$ is known as the parity operator. For the
anti-unitary operator $A_J$ \cite{Wigner} we give its action on the
fiducial vector $\ket{\Psi}$, where $\ket{\Psi}^*$ denotes the vector
obtained by complex conjugation of its components.

We have now specified the representation of the $d = 4p$ Clifford
group that we will use in this paper.  We first apply the Chinese
remainder theorem to rewrite the Weyl--Heisenberg group as $H(4)\times
H(p)$. In Hilbert space this transformation involves only a
permutation matrix. In the dimension-$p$ factor we use the standard
representation, and choose a representative $U_\mathcal{Z} \in U(p)$ that is a
permutation matrix. In the dimension-four factor we use the monomial
representation with its basis vectors enphased according to the
above. The full Clifford group is represented accordingly.

Unfortunately the requirement that we have an order-three permutation
matrix in the Clifford group does not determine the enphasing of the
basis vectors uniquely. 
This is related to the fact that the $d =
4$ Clifford group actually contains two isomorphic copies of the
Weyl--Heisenberg group \cite{Huangjun}. 
This also affects which of the two order-three  
permutation matrices in dimension $4p$ we construct in 
equations \eqref{eq:Zauner1} and \eqref{eq:Zauner2} below (see
Appendix \ref{sec:phase_ambiguity} for more details).

\section{An Ansatz for a SIC fiducial vector}\label{sec:ansatz}

Having decided on a representation of the Weyl--Heisenberg group, we
are in a position to write down an Ansatz (working presupposition) for a SIC fiducial vector that has the
symmetries that we expect from the numerical evidence in low
dimensions \cite{Andrew}.  Naturally we will make use of the freedom to
choose suitable representatives from the conjugacy classes of Clifford
unitaries. Moreover we will regard the Hilbert space $\C^4\otimes
\C^p$ as a direct sum of four copies of $\C^p$.

Let the dimension be $d_\ell = n^2 + 3 = 4p$, where $\ell$ is the position in the dimension 
tower. We will write the fiducial vector as 
\begin{equation}
  \ket{\Psi_0} = N\left( \begin{array}{c}
    {\bf v}_1 \\ \hline
    {\bf v}_2 \\ \hline
    {\bf v}_3 \\ \hline
    {\bf v}_4
  \end{array} \right)
  = N\sum_{i=1}^4\ket{i}\ket{{\bf v}_i},\label{eq:fid_ansatz}
\end{equation}
where ${\bf v}_1, \dots, {\bf v}_4$ are vectors in $\C^p$ and $N$ is a
normalisation factor at our disposal. This vector should have a
unitary symmetry of order $3\ell$, and we want to represent this by a
permutation matrix.  For this purpose we introduce a generator
$\theta$ for the multiplicative group of the integers modulo~$p$ (so that 
$\theta^{p-1}\equiv 1 \bmod p$ and this is the least such power), and set $\delta =
\theta^{(p-1)/3\ell}$ in equation \eqref{eq:diagonal}.  This results
in a permutation matrix $U_F$ such that $U_F^{3\ell} =\openone$.  
Consequently there are two possibilities for the Zauner symmetry in dimension $4p$,
namely $U_{\mathcal{Z}}\otimes U_F$ or $U_{\mathcal{Z}}\otimes
U_F^{-1}$, where $U_{\mathcal Z}$ is the permutation matrix from
eq.~\eqref{eq:UZauner}, and $U_F$ is the permutation matrix in dimension
$p$.  This symmetry gives us one of the two conditions
\begin{alignat}{5}
&&  \left(U_{\mathcal{Z}}\otimes U_F\right) \ket{\Psi_0}&{}=\ket{\Psi_0}\label{eq:Zauner1}\\
\text{or}\quad \nonumber \\ 
\quad&&  \left(U_{\mathcal{Z}}\otimes U_F^{-1}\right) \ket{\Psi_0}&{}=\ket{\Psi_0},\label{eq:Zauner2}
\end{alignat}
and we discuss the first possibility in what follows.  The two
possibilities are related by a Clifford transformation, but given the precise way in 
which we decided to enphase the basis vectors only one of them will afford a fiducial 
vector in a number field with the smallest possible degree, see Appendix \ref{sec:phase_ambiguity}.  
In expanded form, the symmetry \eqref{eq:Zauner1} reads
\begin{equation}
  \left( \begin{array}{c|c|c|c}
    0 & U_F & 0 & 0 \\ \hline
    0 & 0 & U_F & 0 \\ \hline
    U_F & 0 & 0 & 0 \\ \hline
    0 & 0 & 0 & U_F
  \end{array} \right)
  \left( \begin{array}{c}
    {\bf v}_1 \\ \hline
    {\bf v}_2 \\ \hline
    {\bf v}_3 \\ \hline
    {\bf v}_4
  \end{array} \right)
  =\left( \begin{array}{c}
    U_F{\bf v}_2 \\ \hline
    U_F{\bf v}_3 \\ \hline
    U_F{\bf v}_1 \\ \hline
    U_F{\bf v}_4
  \end{array} \right)
 =\left( \begin{array}{c}
    {\bf v}_1 \\ \hline
    {\bf v}_2 \\ \hline
    {\bf v}_3 \\ \hline
    {\bf v}_4 
  \end{array} \right).
\end{equation}
This requires
\begin{alignat}{5}
  {\bf v}_2 = U_F^{-1}{\bf v}_1
  \qquad\text{and}\qquad
  {\bf v}_3 = U_F{\bf v}_1,\nonumber\\
  \noalign{as well as}
  U_F^3{\bf v}_1 = {\bf v}_1
  \qquad\text{and}\qquad
  U_F{\bf v}_4 = {\bf v}_4. \label{eq:symmetri}
\end{alignat}
Thus the vector ${\bf v}_4$ has a permutation symmetry of order
$3\ell$, while for the remaining vectors the symmetry is of order
$\ell$. For the second option \eqref{eq:Zauner2}, the position of the
vectors ${\bf v}_2$ and ${\bf v}_3$ in the Ansatz
\eqref{eq:fid_ansatz} is interchanged.  The independent vectors ${\bf
  v}_1$ and ${\bf v}_4$ can be written as
\begin{equation}
  {\bf v}_1 = (1, y_1, \dots , y_{p-1})^{\rm T},
\qquad
  {\bf v}_4 = (\sqrt{x_0}, x_1, \dots , x_{p-1})^{\rm T}. \label{eq:v0_v1}
\end{equation} 
The normalisation factor $N$ was chosen to make one component equal to
$1$.  The numbers $\sqrt{x_0}$, $x_i$, $y_i$ are to be determined.

Conditions \eqref{eq:symmetri}, together with the form
\eqref{eq:diagonal} of $U_F$, imply that there are $(p-1)/3\ell$
independent numbers $x_i$ and $(p-1)/\ell$ independent numbers $y_i$.
To make this explicit, we introduce a multiplicative ordering of the
$p-1$ vector-indexing integers $j= 1,\ldots,p-1$ by re-indexing them
according to a new variable~$r$ with the range~$0\leq r\leq p-2$,
viz.
\begin{equation}
  j = \theta^r \bmod p, \label{theta}
\end{equation}
the relevance of which becomes clear in eqs.~\eqref{eq:components} below. 
We then introduce two cyclically indexed sets of complex numbers, 
\begin{equation}
  \{  \alpha_0,\alpha_1,\ldots,\alpha_{(p-1)/3\ell-1}\}
  \qquad\text{and}\qquad
  \{  \beta_0,\beta_1,\ldots,\beta_{(p-1)/\ell-1}\}.
\end{equation}
(Looking ahead: These are the numbers that will eventually turn out to be 
square roots of Stark phase units). We extend these cycles to cycles of
length $p-1$ by defining  
\begin{alignat}{5}
  \alpha_{r + j (p-1)/3\ell}
    &{}= \alpha_{ r}, &\qquad {j} &{}= 0, 1, \dots , 3\ell - 1\\
  \beta_{ r + k(p-1)/\ell}
    &{}= \beta_{r}, &\qquad {k} &{}=0, 1, \dots, \ell - 1.
\end{alignat}
The relation between the components $x_j$ and $y_j$ of the vectors
${\bf v}_4$ and ${\bf v}_1$, resp., and the numbers $\alpha_r$ and
$\beta_r$, resp., is given by
\begin{equation}
  x_{\theta ^r} =\alpha_r
  \qquad\text{and}\qquad
  y_{\theta ^r} =\beta_r.\label{eq:components}
\end{equation}
Using equation \eqref{eq:diagonal} with $\delta =
\theta^{(p-1)/3\ell}$ we can now check that conditions
\eqref{eq:symmetri} hold.

We also build an anti-unitary symmetry in. Again guided by Scott's conjectures \cite{Andrew} 
we take it to be 
\begin{equation}
( U_P^{(4)}\otimes U_P^{(p)}) \ket{\Psi_{0}}^* =  \left( \begin{array}{c|c|c|c}
    U_P^{(p)} & 0 & 0 & 0 \\ \hline
    0 & U_P^{(p)} & 0& 0 \\ \hline
    0 & 0 & U_P^{(p)} & 0 \\ \hline
    0 & 0 & 0 & -U_P^{(p)}
  \end{array} \right)\ket{\Psi_{0}}^* =\ket{\Psi_{0}}, \label{UPd}
\end{equation}
where $U_P^{(4)}$ is the parity operator in dimension four,
$\ket{\Psi_0}^*$ denotes component-wise complex conjugation on the
entries of the target vector~$\ket{\Psi_{0}}$, as in eq.~\eqref{UP4},
and $U_P^{(p)}$ is the parity operator in dimension $p$, a permutation
matrix obtained by setting $\delta = -1$ in equation
\eqref{eq:diagonal}.  This requires
\begin{alignat}{5}
  y_{-j}&{} = y_j^* , 
   &&\text{which is equivalent to $\beta_{r + (p-1)/ 2} = \beta_{ r}^*$,} \label{eq:cc1} \\
  x_{-j}&{} = -x_j^*,
    &&\text{which is equivalent to $\alpha_{r + (p-1)/2} = - \alpha_{ r}^*$,} \label{eq:cc2} \\ 
  \text{and}\quad
  \sqrt{x_0}&{} = - (\sqrt{x_0})^* .   \label{eq:cc3} 
\end{alignat}
Hence $x_0 < 0$. 

Finally we make the more dramatic assumption that 
\begin{alignat}{10}
  1&{}= |x_1|^2&&{}= \dots&&{}= |x_{p-1}|^2&&{}= |y_1|^2&&{}= \dots&&{}= |y_{p-1}|^2\\ 
   &{}= |\alpha_0|^2&&{}= \dots&&{}= |\alpha_{(p-1)/3\ell-1}|^2&&{}= |\beta_0|^2&&{}= \dots&&{}= |\beta_{(p-1)/\ell-1}|^2. 
\end{alignat}
Thus all these numbers lie on the unit circle, and the vector
$\ket{\Psi_0}$ is consequently referred to as being \emph{almost
flat}. Another way to state this assumption is to require that a real
SIC fiducial vector exists, from which the almost flat vector is
reached by means of a suitable Clifford transformation. Because we use
a non-standard representation of the Clifford group the standard
argument to this effect \cite{Roy, Mahdad, ABGHM} has to be modified,
and we give the details in Appendix \ref{sec:phase_ambiguity}.

For an almost flat vector some of the overlaps are real. The SIC
condition requires them to be phase factors divided by $\qd1$. Taken
together this means that we must (tentatively) impose
\begin{equation}
  \langle \Psi_0|\openone_4\otimes Z^j|\Psi_0\rangle
    = \begin{cases}
      \hfil 1, & \text{if $j = 0$;} \\
    \pm  \frac{1}{\qd1}, & \text{if $j \neq  0$.}
  \end{cases}\label{baby1pm} \end{equation}
We assume that the left-hand side forms a Galois orbit when $j \neq
0$, so the sign must be independent of $j$.  Written out, this leads
to the conditions
\begin{equation} N^2(|x_0|+d-1) = 1 
  \qquad\text{and}\qquad N^2(|x_0|-1) = \pm \frac{1}{\qd1}. 
\end{equation}
If we choose the positive sign, and recall that we have already
established that $x_0$ is negative, we find that
  \begin{equation} x_0 = - 2 - \qd1
  \qquad\text{and}\qquad
  N^2 = \frac{1}{d+1+\qd1}. \label{eq:x0}
\end{equation}
For the negative sign there is no solution (with $|x_0| \geq 0$ and $d > 3$). Hence we add equations 
\eqref{eq:x0} to the Ansatz. This means that the SIC conditions 
\begin{equation}
  \langle \Psi_0|\openone_4\otimes Z^j|\Psi_0\rangle
  = \begin{cases}
    \hfil1, & \text{if $j = 0$;} \\
  \frac{1}{\qd1}, & \text{if $j \neq  0$.}
  \end{cases}\label{baby1}
\end{equation}
  are built into the Ansatz. It also follows that 
\begin{equation}
  \langle \Psi_0|D_{0,2}\otimes Z^j|\Psi_0 \rangle
  = \langle \Psi_0|D_{2,0}\otimes Z^j|\Psi_0\rangle
  = \langle \Psi_0|D_{2,2}\otimes Z^j|\Psi_0 \rangle
  = - \frac{1}{\qd1}, \label{baby2}
\end{equation}
where the displacement operators $D_{i,j}$ are with respect to our
chosen representation in dimension four, and the operator $Z$ is again in the
standard representation in dimension $p$.

To go further we need to know more about the phase factors $\alpha_r$
and $\beta_r$.  This is where we turn to number theory.

\section{The number fields used for the SIC construction}\label{sec:numbertheory}
The position now is that we need two sets of calculable and
cyclically-indexed numbers on the unit circle, one set with
$(p-1)/3\ell$ members and one with $(p-1)/\ell$ members, to place in
the Ansatz for the SIC fiducial vector. We will discuss those numbers
now, subject to some hypotheses from the Stark conjectures.  We will
partly rely on the background provided in~\cite[\S III (A)]{ABGHM},
and then focus on special features arising when $d = 4p$. In
particular, we are interested in the lattice of ray class fields that
arises when we successively add factors of $2$ to their modulus
(Propositions~\ref{hydra1}, \ref{hydra2}, \ref{hydra4}, \ref{hydra5}
and~\ref{hydra7}).  For any number field~$F$, the notation~$\Z_F$ will
denote its ring of integers.  If~$\pp$ is a prime ideal of~$F$ lying
above the rational prime~$p$, with ramification index~$e\geq1$,
then~$F_\pp$ will denote the local field obtained by completion at the
place corresponding to~$\pp$,~$\Z_\pp$ its ring of integers,
and~$v_\pp(x)$ the valuation at~$\pp$ normalised so that~$v_\pp(p) =
e$.

Given any dimension $d\ge 4$, we start with the quadratic field
$K=\QD$, where $D$ is the square-free part of
$(d+1)(d-3)=(d-1)^2-4$ as specified in eq.~\eqref{eq:D}. The field has a
single non-trivial automorphism $\tau$ mapping $\sqrt{D}$ to
$-\sqrt{D}$. By $\jj$ we denote the embedding of $K$ into the complex
numbers with $\jj(\sqrt{D})>0$; hence $\jj^\tau$ is the embedding
mapping $\sqrt{D}$ to a negative number.  

For the construction of the fiducial vector, we need the fields shown
in Figure~\ref{fig:subfield_lattice}.  The lines and the numbers next to
the lines indicate a field extension and its degree.  The field
$K^{(1)}$, denoted below by~$H_K$, is the \emph{(wide) Hilbert class
field} of $K$, the maximal everywhere unramified abelian extension
of~$K$.  The notation~$(1) = 1 \zk$ signifies the ideal $\zk$ itself, 
so this is consistent with the notation for ray class fields which is introduced immediately below. 
The Galois group~$\Gal{\HK}{K}$ is isomorphic to~$\cK$, the ideal
class group of~$\zk$.  The degree of the extension, which is the order
of~$\cK$, is the class number of~$K$ and is denoted by~$\h$.

As noted in the introduction, every real quadratic field $\QD$ is
connected to a dimension tower $\{d_\ell(D) \}_{\ell = 1}^\infty$.
In~\cite{AFMY} and~\cite{ABGHM} it is explained that the
dimensions $d_\ell(D)$ above a given fixed value of $D$ take values
given by adding $1$ to the traces of the integer powers $u_D^r$ of the
first totally positive power $u_D$ of a fundamental unit $u_K$
for $\zk$.  That is, denoting the $\ell$th dimension above $\QD$
by $d_\ell(D)$:
\begin{equation}\label{deeell}
d_\ell(D) = u_D^\ell + u_D^{-\ell} + 1.
\end{equation}
An example with $D=5$, so $d_1 = 4$, was given in
eq.~\eqref{eq:tower}.  It is expected (and confirmed in every case we
have studied) that the unitary symmetry of a ray class SIC in
dimension $d_\ell$ is of order $3\ell$; hence SICs that occur at
position $\ell > 1$ in a dimension tower are easier to construct than
the mere size of the dimension would indicate.  However, for all but
one choice of the quadratic field $K=\QD$, dimensions of the form $d =
4p$ can occur only when $\ell = 1$. The unique exception is the
dimension tower for $D = 5$, given in eq.~\eqref{eq:tower}; see
Appendix~\ref{subsec:towers}.  And indeed, the highest dimension that
will be reached in this paper ($d =39604$) is for $D = 5$ and $\ell =
11$.

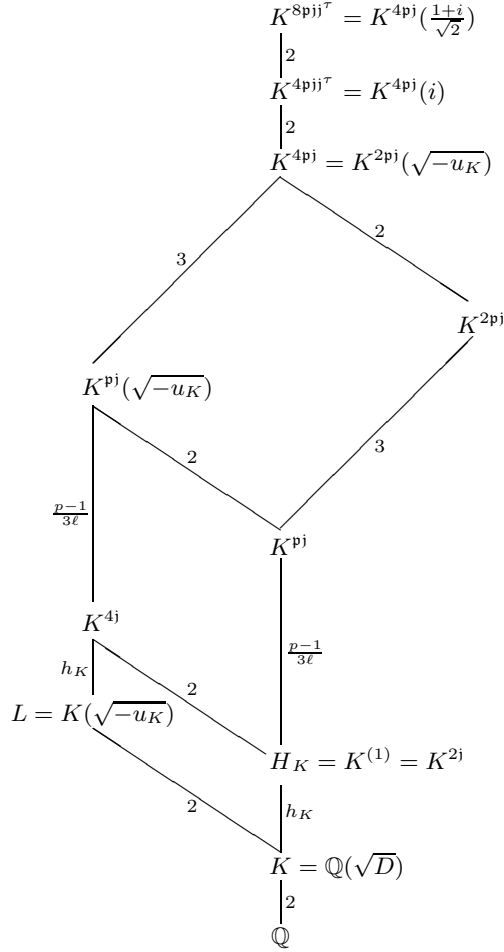
\begin{figure}[hbt]
  \centerline{{\small\unitlength0.9\unitlength
  \begin{picture}(250,410)(-20,-105)
    \put(100,290){\makebox(0,0){$K$\rlap{${}^{8\pp\jj\jj^\tau}=K^{4\pp\jj}(\frac{1+i}{\sqrt{2}})$}}}
    \put(102,275){\makebox(0,0)[l]{\scriptsize$2$}}
    \put(100,266){\line(0,1){18}}
    \put(100,260){\makebox(0,0){$K$\rlap{${}^{4\pp\jj\jj^\tau}=K^{4\pp\jj}(i)$}}}
    \put(102,245){\makebox(0,0)[l]{\scriptsize$2$}}
    \put(100,236){\line(0,1){18}}
    \put(100,230){\makebox(0,0){$K$\rlap{${}^{4\pp\jj}=K^{2\pp\jj}(\sqrt{-u_K})$}}}
    \put(100,224){\line(3,-2){78}}
    \put(139,199){\makebox(0,0)[bl]{\scriptsize$2$}}
    \put(100,224){\line(-1,-1){78}}
    \put(61,187){\makebox(0,0)[br]{\scriptsize$3$}}
    \put(22,136){\makebox(0,0){$K$\rlap{${}^{\pp\jj}(\sqrt{-u_K})$}}}
    \put(178,163){\makebox(0,0){$K$\rlap{${}^{2\pp\jj}$}}}
    \put(100,76){\line(-3,2){78}}
    \put(61,104){\makebox(0,0)[bl]{\scriptsize$2$}}
    \put(139,114){\makebox(0,0)[lt]{\scriptsize$3$}}
    \put(100,77){\line(1,1){80}}
    \put(100,70){\makebox(0,0){$K$\rlap{${}^{\pp\jj}$}}}
    \put(100,-14){\line(0,1){78}}
    \put(102,25){\makebox(0,0)[l]{\scriptsize$\frac{p-1}{3\ell}$}}
    \put(100,-20){\makebox(0,0){$H$\rlap{${}_K=K^{(1)}=K^{ 2\jj}$}}}
    \put(21,83){\makebox(0,0)[r]{\scriptsize$\frac{p-1}{3\ell}$}}
    \put(22,38){\makebox(0,0){$K$\rlap{${}^{4\jj}$}}}
    \put(94,-18){\line(-3,2){72}}
    \put(61,7){\makebox(0,0)[bl]{\scriptsize$2$}}
    \put(100,-59){\line(-3,2){78}}
    \put(61,-42){\makebox(0,0)[bl]{\scriptsize$2$}}
    \put(22,46){\line(0,1){82}}
    \put(21,18){\makebox(0,0)[r]{\scriptsize$\h$}}
    \put(22,7){\line(0,1){23}}
    \put(22,-1){\makebox(0,0){$L = K{}(\sqrt{-u_K})$}}
    \put(100,-59){\line(0,1){28}}
    \put(102,-42){\makebox(0,0)[l]{\scriptsize$\h$}}
    \put(100,-65){\makebox(0,0){$K$\rlap{${}=\QD$}}}
    \put(100,-89){\line(0,1){18}}
    \put(102,-80){\makebox(0,0)[l]{\scriptsize$2$}}
    \put(100,-95){\makebox(0,0){$\Q$}}
  \end{picture}
  }}
  \caption{Lattice of ray class fields for dimension $d=4p$.  Here $i$
    and $\frac{1+i}{\sqrt{2}}=e^{i\pi/4}$ denote primitive fourth and
    eighth roots of unity, respectively. The relations between the
    fields are discussed in Section \ref{sec:fields}.
    We compute Stark units for fields isomorphic to
    $K^{\mathfrak{p}\mathfrak{j}}$ and $K^{2\mathfrak{p}\mathfrak{j}}$.
    \label{fig:subfield_lattice}}
\end{figure}

\subsection{The exact sequence of global class field theory}\label{glob}
We need to introduce a small amount of notation and tools from global class field theory. 
Let $\mm_0$ be any {integral ideal} of the ring $\zk$, and let
$\mm_\infty$ denote some---possibly empty---subset of $\{ \jj,
\jj^{\tau} \}$.  The formal product $\mm = \mm_0 \mm_\infty$ is a
\emph{modulus}.  
In view of the absence of any standardised notation in the literature, 
we shall just write $K^{\mm}=K^{\mm_0\mm_\infty}$ for the ray
class field of $K$ for the modulus $\mm$, and $F(\alpha)$ for the
extension of any field $F$ by an algebraic number or indeterminate
$\alpha$.

We state the exact sequence of global class field theory (see
eq.~(2.7) in Ref.~\cite{cohenstev}, Theorem 1.7 in Ref.~\cite{milne},
or \S0 in Ref.~\cite{tate}).  With \emph{any} number field $K$ as base
field, and \emph{any} modulus $\mm = \mm_0 \mm_\infty$, the following
sequence (defining the map $\psi$) is exact:
\begin{equation}\label{globcft}
1 \rightarrow U_1^{\mm} \rightarrow {\zkx} \xrightarrow{\psi } \mg{\mm_0} \times{\{\pm1\}}^{\# \mm_\infty } 
	\rightarrow \Gal{K^\mm}{K} \rightarrow \cK   \rightarrow 1,
\end{equation}
where ${\# \mm_\infty }$ is the number of real infinite places in
$\mm$.  Note that ${ \mm_\infty }$ is a strict set (\ie, not a multiset) 
so that any infinite place may only appear once or not at all. 
The term $\ker\psi = U_1^\mm$ is the subgroup of the global
units ${\zkx}$ which are simultaneously congruent to $1$ modulo
$\mm_0$ and positive at the real places in $\mm_\infty$. 

In the case
of real quadratic fields, by Dirichlet's unit theorem~\cite[Theorem 5.1]{milneANT} the $\Z$-rank of
the unit group is $1$ 
and so this kernel has $\Z$-rank one. 
Moreover, it is
torsion-free; except possibly when the residue class ring $\ag{ \mm_0
}$ has characteristic $2$.  
As explained in~\cite{ABGHM}, the unit group $\zkx$ of $K$ is
generated by $-1$ and a fundamental unit $u_K=(n_1+\sqrt{n_1^2+4})/2$, 
where $n_1$ once more is the minimal positive 
integer such that the expression $n_1^2+4$ ($ = d_1+1$) 
under the square root sign has square-free part $D$. 
Taking the positive square root, we have
$\jj(u_K)>1$ (and therefore $\jj^\tau(u_K)<1$ in our negative norm 
cases where our dimensions $d$ are always of the form $n^2+3$).  

Specialising to the cases considered in this paper and its predecessor~\cite{ABGHM}, 
where $d$ is of the form $n^2+3$, 
the first totally positive power $u_D$ of the
fundamental unit $u_K = (n_1+\sqrt{n_1^2+4})/2$ is always $u_D = u_K^2$. 
When moreover $d=n^2+3=4p$, the rational prime $p$ splits over
$K$, \ie, the ideal $(p)$ factors into prime ideals $\pp$ and
$\overline{\pp}=\pp^\tau$, as already noticed in eq.~\eqref{eq:split}
above.  

For ease of notation later on we extract a pair of short exact sequences from the middle
of~\eqref{globcft} as follows: 
\begin{equation}\label{snort}
1 \rightarrow {\zkx} / U_1^{\mm}  \xrightarrow{\psi}  \mg{\mm_0} \times  {\{\pm1\}}^{\# \mm_\infty}
	\rightarrow  \coker\psi    \rightarrow 1 , 
\end{equation}
where the tail fits back into~\eqref{globcft} via
\begin{equation}\label{snortlet}
1 \rightarrow \coker\psi \rightarrow \Gal{K^\mm}{K} \rightarrow \cK   \rightarrow 1 . 
\end{equation}
So $\Gal{K^\mm}{K}$ is an abelian group extension of the ideal class group $\cK$ by $\coker\psi$, 
the image of the multiplicative group of the ray residue ring modulo the global units.

\subsection{Required hypotheses from the Stark conjectures}

We need to clarify the aspects of Stark's programme of
conjectures~\cite{StarkIII} which we shall exploit in our
construction.  For more detailed explanations surrounding their
application to the $n^2+3$ subset of the SIC problem, as well as
linking our notation to the literature, we point to Section IV (C)
of~\cite{ABGHM}.  For this paper, we need to assume the following
three hypotheses, which are just Hypotheses 2, 3, 4 from \cite{ABGHM}.
As above, the notation $\mm_0$ refers to a generic finite modulus,
and $\jj$ is the specific embedding defined at the beginning of
Section~\ref{sec:numbertheory}. 
We need to designate an involution $\sigt \in \Gal{K^{\mm_0\jj}}{H_K}$, 
following the notation in~\cite{ABGHM}, 
which acts as complex conjugation for every 
complex embedding of $K^{\mm_0\jj}$. \label{gammaray}

\begin{hypothesis}[\S4, p.~74 in Ref.~\cite{StarkIII}]\label{S1}
   The Galois element $\sigt$ induces complex conjugation in the
   complex embeddings of $K^{\mm_0\jj}$.  Since it is also algebraic
   inversion, it forces the Stark units in $K^{\mm_0\jj}$ to lie on
   the unit circle in their complex embeddings.
\end{hypothesis}
\noindent We will refer to these complex numbers as \emph{Stark phase units}.

\begin{hypothesis}[Theorem 1 (i) in Ref.~\cite{StarkIV}]\label{tood}
  In their real embeddings, the Stark units $\es$ are all positive.
\end{hypothesis}

\begin{hypothesis}[Stark/Tate~\rm{`over-$\Z$': Conjecture in Ref.~\cite{roblot2}}]
\label{sqaw}
  The extension $K^{\mm_0\jj}(\qea)$ of $K^{\mm_0\jj}$ obtained by
  adjoining the square root of any one of the $\es$ is
  itself an abelian extension of $K$.
\end{hypothesis}

\subsection{Structure of the tower of fields $K^{4\pp\jj}/K^{2\pp\jj}/K^{\pp\jj}/\HK/K$}\label{sec:fields}
The following ancillary results provide the class field-theoretic
backbone of the algorithms to be described in the succeeding sections
of the paper.  Unless stated otherwise, we assume everywhere in this
section that we are working with a dimension $d$ of the form $n^2+3$
which is $4$ times a prime number $p$.  As noted in~\cite{Anti1} and
equation~(3) of~\cite{ABGHM}, \emph{any} prime $p>3$ which divides
into such a dimension must satisfy $p\equiv 1\bmod 3$.

When writing conductors or moduli as principal ideals of $\zk$, the
ideal generated by an element $e \in\zk$ will be denoted by any one
of $e\zk = (e) = e$, depending on the context.

\subsubsection{The extensions $K^{2\jj}/\HK$, $L/K$ and $L\HK/\HK$}\label{grass}
\begin{proposition}\label{hydra1}
All four ray class fields $K^{2}$, $K^{2\jj}$, $K^{2\jj^\tau}$,
$K^{2\jj\jj^\tau}$ of $K$ with finite part of the modulus equal to
$(2) = 2\zk$ are isomorphic to the Hilbert class field $\HK$.
\end{proposition}

Note the difference between our $n^2+3$
cases, where the fundamental units have norm $-1$, and all other
cases, where the fundamental units have norm $+1$.  In the latter
situation the order of exactly two of the four ray class groups in the
proposition will always be a factor of two bigger than the class number (the
remaining two will be isomorphic to the class group, as in the
proposition). 

\begin{proof}
Note first of all that this result makes no use of the value of $d = d_k(D)$; merely of $D$. 
A glance at the exact sequence \eqref{globcft} shows that the statement of the proposition is
equivalent---when $\mm_\infty \in \bigl\{ \phi$, $\{\jj\}$, $\{\jj^\tau\}$,
$\{\jj, \jj^\tau\}\bigr \}$ and $\mm_0 = (2)$---to $\psi$ being surjective.
Now since $D\equiv5\bmod8$ we know by standard arguments (see for
example Lemma~3 of~\cite{ABGHM}) that the prime $2$ is inert in the
extension $K/\Q$, and so $\left( \zk / {(2)} \right)^\times \cong
C_3$; $C_N$ being the symbol for a cyclic group of order $N$.  

Moreover, we now show that this group is generated 
by the image of $u_K = \frac{n_1+\sqrt{n_1^2+4}}{2}$. 
The minimal polynomial for $u_K$ 
over $\Q$ is $X^2 - n_1X - 1$, where $n_1$
is as above.  
In our case $n_1$ is odd, so modulo $2$ this becomes
just $X^2+X+1$.  If $u_K$ itself were $\equiv1\bmod 2\zk$ then
$u_K^2+u_K+1$ would not be zero modulo $2\zk$; so the image of $u_K$
must generate the cubic cyclic component, as asserted.

It remains to show that the $2$-primary part of the image of $\psi$
maps onto the component ${\{\pm1\}}^{\# \mm_\infty }$ in
eq. \eqref{globcft}.  But we have assumed that the norm of the
fundamental unit is $-1$; hence the fundamental unit is of mixed
signature.  Hence all four possible combinations of signs
$(\pm1,\pm1)$ occur (depending on the real places in the conductor)
for the odd powers of $\pm u_K$.
\end{proof}

As in~\cite{ABGHM}, we write $L = K(\sqrt{-u_K})$. 
The extension $L/K$ is properly quadratic because it is \emph{inert} over $\jj$, in 
Gras' terminology~\cite{roblot2,gras}: meaning that it becomes complex. 
For reference, we explain briefly this terminology and its 
more customary alternative. 
In \cite{gras} a 
real place which remains real in the infinite places above it is 
said to `split completely'. 
On the other hand, in Hasse's more standard language wherein the latter `split' 
extension would be said to be `unramified', our extension is~`ramified'. 

Since the Hilbert class field $\HK$ of $K$ is totally real, it 
follows that the compositum $L  \HK$ is equal to $\HK(\sqrt{-u_K})$ 
and is also a proper quadratic extension of $\HK$. 
Also following~\cite{ABGHM} we define, for each $k \geq 1$, 
\begin{equation}\label{norxi}
\xi_k = \sqrt{x_0} = \sqrt{-2-\sqrt{d_k+1}} ,
\end{equation}
reserving again the notation $n_k = \tr u_K^k$ for the positive integer 
such that $u_D^k+u_D^{-k}+1 = d_k = n_k^2+3$. 

It was observed in \cite{ABGHM} that the extension of $K$ generated by the
$\xi_k$ is independent of $k$, so that for a fixed $D$, we may simply focus on $\xi_1$. 
For completeness, we provide an argument here. 
For every odd $\ell \in \N$ we recall that 
\begin{equation}
  u_K^\ell  =  \frac{ n_\ell + \sqrt{n_\ell^2+4} }{2}
  \qquad\text{and}\qquad
   u_K^{-\ell}  =  -\frac{ n_\ell - \sqrt{n_\ell^2+4} }{2},
\end{equation}
and so in particular, 
\begin{equation}
\sqrt{d_\ell+1}  =  \sqrt{n_\ell^2+4} =  u_K^\ell  + u_K^{-\ell} .
\end{equation}
By Kummer theory we need only show that the ratio $\xi_\ell^2 / \xi_1^2$
is a square in $K$: 
\begin{alignat*}{5}
  \frac{\xi_\ell^2}{\xi_1^2}
  & {}=  \frac{-2-\sqrt{d_\ell+1}}{-2-\sqrt{d_1+1}} \\
  & {}=  \frac{ 2 + u_K^\ell + u_K^{-\ell} }{ 2 + u_K + u_K^{-1} },&\quad&\text{from above,} \\
  & {}=  \frac{1}{u_K^{\ell-1}} \times \frac{ u_K^{2\ell} + 2 u_K^\ell + 1 }{ u_K^2 + 2u_K + 1 },&&\text{by multiplying top and bottom by $u_K^\ell$,} \\ 
  & {}=  \left(  \frac{1}{u_K^{\frac{\ell-1}{2}}} \times   \frac{ u_K^{\ell} +1  }{ u_K + 1  } \right)^2 , &&\text{since $\ell$ is odd,}\\
  & {}\in  (K^\times)^2, && \text{as required.}
\end{alignat*}

Moreover, $L = K(\xi_1) = K(\qu)$, the last equality holding because
$\xi = \qu - \frac{1}{\qu};$ and conversely: $\qu =
\frac{u_K-1}{n}\xi.$ In particular, the discriminant of the
extension $L/K$ must divide into $4\zk$.

\begin{proposition}\label{hydra2}
The quadratic extension $L  \HK = \HK(\sqrt{-u_K})$ of the Hilbert class 
field $\HK$ equals the ray class field $K^{4\jj}$ of $K$ of modulus $4\jj$.  
\end{proposition}

\begin{proof}
This is case \rm (A) (III) \rm of Proposition 8 of~\cite{ABGHM} (see also Remark 9 (i) there). 
We could also prove it directly by writing the respective versions of
equation \eqref{snort} for the conductors $ 4\jj$ and $ 2\jj$, 
linking them by the natural homomorphisms induced by the 
divisibility relation between the conductors, and using the 
snake lemma~\cite[(II.28), p.120]{gelman} together with 
Proposition~\ref{hydra1};
see also Proposition~\ref{ram2} in the Appendix.
\end{proof}

\subsubsection{The extensions $K^{\pp\jj}/\HK$ and $K^{2\pp\jj}/\HK$}
We note that $\pp=\bigl((\qd1 + 1)/2\bigr)$, by direct calculation: 
its $\Gal{K}{\Q}$-conjugate $\ol\pp$ is a distinct ideal ($p = d/4$ 
cannot be ramified in $K/\Q$, since the discriminant $D$ of $K/\Q$ is coprime to $p$) 
which multiplies with it to give the ideal $\pp\ol\pp = p\zk$.

\begin{proposition}\label{hydra3}
The extension $K^{\pp\jj}/\HK$ is cyclic of degree $(p-1)/3\ell$. 
\end{proposition}

\begin{proof}
Both assertions will follow from Proposition~\ref{hydra4}, upon 
comparison of the exact sequences \eqref{snort} for the two 
respective conductors, since the only term introduced by the 
extra factor of $2$ in the conductor is a copy 
of $\left( \zk/(2) \right)^\times \cong C_3$, and the 
Galois group here is a quotient of the (cyclic) one below 
and consequently must itself be cyclic. 
\end{proof}

\begin{proposition}\label{hydra4}
The extension $K^{2\pp\jj}/\HK$ is cyclic of degree $(p-1)/\ell$.  
\end{proposition}

\begin{proof}
The degree is given by Proposition 10 {\rm(II)} of~\cite{ABGHM}.
Suppose first of all that $D\neq5$, which by Proposition~\ref{ellone} means
that we may assume $\ell=1$.  To prove that the Galois group in
question is cyclic, we must produce an element of
order $\frac{p-1}{\ell} = p-1$ inside the cokernel of $\psi$
in \eqref{globcft} (the exact sequence being expressed 
with conductor $\mm = 2\pp\jj$).  The codomain
of $\psi$ is of the form $\{\pm1\} \times \left( \zk/(2)
\right)^\times \times \left( \zk/\pp \right)^\times \cong C_2 \times
C_3 \times C_{p-1}$;
see also the discussion after the statement of Proposition~\ref{ants}. 

Let $\lambda = \lambda + \pp$ be a generator of $\left( \zk/\pp
\right)^\times \cong C_{p-1}$. 
There are of course $\phi(p-1)$ choices for a generator, 
where $\phi$ is the ordinary Euler totient function; but 
without loss of generality we may
choose it to satisfy $\lambda^{\frac{p-1}{3}} \equiv u_K \bmod \pp$,
since $\ord{\pp}{u_K} = 3\ell = 3$ by applying the argument in the
proof of the same Proposition~10 of~\cite{ABGHM} cited above but using
(in that paper's notation) $\partial_\ell/2$ in place
of $\partial_\ell$.  
Notice further that independently of the choice
of the generator $\lambda$ we will have $\lambda^{\frac{p-1}{2}} = -1
+ \pp$.

Similarly, in the proof of Proposition~\ref{hydra1} above 
(or see Proposition 8~(A)~(III) of~\cite{ABGHM}) it
is shown that the restriction of the image of $u_K$ generates the
group $\left( \zk/(2) \right)^\times$; so from now on we shall use the
image $u_K + 2\zk$ of $u_K$ as a choice of generator for $\left(
\zk/(2) \right)^\times \cong C_3$.  Where it causes no confusion we
shall simply write the elements of $\zk$ to represent their images
under the respective reduction maps.

Putting all of this together, bearing in mind the choice of the real
infinite place $\jj$ which sends $u_K$ to a positive number as well as
the fact that the ring $\zk/(2)$ has characteristic $2$, the image
of $\psi$ therefore contains the diagonal
embedding $\langle\psi(-1)\rangle = \langle(-1,1,-1)\rangle$ of the
second roots of unity $\{\pm1\}$ as well as a $C_{3\ell} = C_3$ term
generated by $\psi(u_K) = (1,u_K,u_K)$, since $U_1^{2\pp\jj} = \langle
u_K^{3\ell}\rangle = \langle u_K^{3}\rangle$.  So the elements of the
group $\im\psi$ inside the codomain as expressed above may be
denumerated as
\[
\im\psi = 
\{
(1,1,1), (-1,1,\lambda^{\frac{p-1}{2}}), (1,u_K,\lambda^{\frac{p-1}{3}}), 
(-1,u_K,\lambda^{\frac{5(p-1)}{6}}), (1,u_K^2,\lambda^{\frac{2(p-1)}{3}}), 
(-1,u_K^2,\lambda^{\frac{(p-1)}{6}})
\} ;
\]
that is, a cyclic group of order $6$ generated by either 
of $(-1,u_K,\lambda^{\frac{5(p-1)}{6}})$ or $(-1,u_K^2,\lambda^{\frac{(p-1)}{6}})$. 
Hence for example raising the element $(1,1,\lambda)$ to a power $r$ 
will land in the image of $\psi$ if and only if $r$ is divisible by $p-1$. 
This provides us with a cyclic subgroup of $\coker\psi$ of order $p-1$, as required. 

For the remaining case where $D=5$ we first observe that $\ell$ must
be odd and not divisible by $3$, by~\eqref{zeroed} in
Appendix~\ref{tours}.  So the argument above goes through virtually
unchanged, since raising to the power $\ell$ is an automorphism of
$C_6$.
\end{proof}

\subsubsection{The extensions $K^{\pp\jj}(\alpha_r)/K$ and $K^{2\pp\jj}(\beta_r)/K$}
Recall that we still work under the hypothesis that $d=n^2+3$ 
is of the form $4p$ for some (odd) prime $p$. 
We are now ready to prove the main result of this section.  Once again
the results of~\S 3A of~\cite{ABGHM} pretty much suffice to prove it;
however because this will be important in other contexts and it is a
relatively self-contained sub-case, we shall prove it more directly.
We recall in passing that $p\zk$ splits as $\pp \ol\pp$, and in the
extension $L/K$, $\pp$ splits and $\ol\pp$ remains inert
if $p\equiv1\bmod4$, and vice-versa if $p\equiv3\bmod4$ (see Lemma~13
of~\cite{ABGHM}).

\begin{proposition}\label{hydra5}
With notation as above, the square roots $\alpha_r$ of the Stark phase
units $\es$ of $K^{\pp\jj}$ are contained in the field $L K^{\pp\jj} = K^{4\pp\jj}$. 
\end{proposition}

\begin{proof}
The fact that $L K^{\pp\jj} = K^{4\pp\jj}$ follows 
from Proposition~\ref{hydra2} and the fact that the 
conductor of a compositum of fields is the lowest common multiple of the 
conductors of the fields; see~\cite[Proposition II.4.1.1]{gras}. 
So for the containment of the $\alpha_r$, it will be enough to show that 
the conductor of $K^{\pp\jj}(\alpha_r) / K$ divides into ${4\pp\jj}$. 
We know by Hypothesis \ref{sqaw} and class field theory, that this conductor 
must divide into ${2^t\pp\jj}$ for some minimal $t\geq0$.

Let $\alpha$ be any of the square roots of the Stark phase units
mentioned in the statement of the proposition.  By Kummer theory
with $K^{\pp\jj}$ as base field, the argument below will be independent of this
initial choice.  If $\alpha\in K^{\pp\jj}$ then there is nothing
further to prove; so we suppose that $K^{\pp\jj}(\alpha)/K^{\pp\jj}$
is a proper quadratic extension.  

Here we need to invoke 
Hypothesis~\ref{tood}, since we want all of the real places 
of $K^{\pp\jj}$ (that is, those above $\jj^\tau$) to remain real in 
the field $K^{\pp\jj}(\alpha)$, so that in 
particular the conductor of $K^{\pp\jj}(\alpha) / K$ 
still only contains the one real place $\jj$ of $K$. 
(As remarked in Section~\ref{grass}, in Gras' terminology~\cite{roblot2,gras} 
these real places `split completely' in the 
extension $K^{\pp\jj}(\alpha) / K^{\pp\jj}$, and 
therefore in $K^{\pp\jj}(\alpha) / K$; 
in most older references they would
be said to be `unramified'). 

Since $p=\pp\ol\pp$ is odd there is no ramification 
above the prime ideal $2\zk$ in the extension $K^{\pp\jj} / K$. 
Hence, using Lemma~7 of~\cite{ABGHM} with $F=K^{\pp\jj}$ and $u=\alpha$, 
in our case the absolute ramification index $e$ equals $1$ 
and we deduce by putting the local results together 
(as per the proof of that lemma) that the power of $2$ appearing in the  
conductor of $K^{\pp\jj}(\alpha) / K$, referred to above as $t$, must be $1$ or $2$. 
This proves that the conductor of $K^{\pp\jj}(\alpha_r) / K$ divides into ${4\pp\jj}$, as asserted. 
\end{proof}

\begin{proposition}\label{hydra6}
The numbers $\alpha_r\sqrt{x_0}$ lie in the field $K^{\pp\jj}$.
\end{proposition}

\begin{proof}
This Kummer theory argument was observed in Theorem~14
of~\cite{ABGHM}, under Hypothesis~\ref{sqaw} of this paper.
\end{proof}

\begin{proposition}\label{hydra7}
The square roots $\beta_r$ of the Stark phase units of $K^{2\pp\jj}$
are contained in the field $LK^{2\pp\jj} = K^{4\pp\jj} $.
\end{proposition}

\begin{proof}
By class field theory and Proposition \ref{hydra2}, 
$$
K^{\pp\jj}  \leq  K^{2\pp\jj}  \leq  K^{4\pp\jj} = LK^{2\pp\jj} \ .
$$
Forming exact commutative diagrams from the exact
sequences~\eqref{snort} and \eqref{snortlet} for the
conductors $2\pp\jj$ and $4\pp\jj$ and using the functoriality of
class field theory for the respective connecting maps, we obtain by
two applications of the snake lemma~\cite[(II.28), p.120]{gelman} a
short exact sequence
\begin{align*}
1 \lra U_1^{2\pp\jj} / U_1^{4\pp\jj} \lra  \ker{\left( \mg{4} \lra \mg{2} \right)} \lra \Gal{K^{4\pp\jj}}{K^{2\pp\jj}} \lra 1,
\end{align*}
where we know (proof of Proposition~8 (A) of~\cite{ABGHM}) the central
term is a Klein $4$-group and (by the argument concerning the orders
of $u_K$ in the proof of Proposition~\ref{hydra4}) the left term is a $C_2$
group, so the right hand side must also be isomorphic to $C_2$.
So $K^{4\pp\jj}$ is a proper quadratic extension of $K^{2\pp\jj}$,
ramified above $2$ by the characterisation of $K^{4\pp\jj}$ as containing $L 
K^{\pp\jj}$ in Proposition~\ref{hydra5}.

On the other hand, 
by Hypothesis~\ref{sqaw} (with $\mm_0 = 2\pp$), 
the $2$-part of the conductor of the ray class field of $K$ containing
the square root $\beta_r$ of any Stark unit from $K^{2\pp\jj}$ 
must divide into $2^t\zk$ for some $t\geq0$. 
(It is independent of the choice of $\beta_r$, by Hypothesis~\ref{tood}). 
But then a fortiori by Lemma~7 of~\cite{ABGHM} applied as in Proposition~\ref{hydra5}, 
it follows that the field formed by extending $K^{2\pp\jj}$ by this
square root, is indeed contained in $K^{4\pp\jj}$.
\end{proof}

\section{Main results}\label{sec:results}
It is time to pull the threads together.  We start this section with a
compact description of the algorithm that we have successfully used to
compute exact fiducial vectors in the sixteen dimensions of the
form $d=n^2+3=4p$ appearing in Table~\ref{tab:solutions} below.
 After that we will provide explanations and
discuss several aspects, including some possible variations of some of
the steps. The reader may find it helpful to refer to Figure
\ref{fig:subfield_lattice} for some of the notation, and to look
at Section \ref{sec:example} where an example is worked out in
complete detail. A list of dimensions where we have successfully
computed an exact SIC is given in Table \ref{tab:solutions}.  Data for
the solutions can be found online \cite{online}.

\subsection{Our algorithm}\label{sec:recipe}
As introduced above, the field $K$ is the real quadratic field
$K=\QD$ where $D$ is the square-free part of $(d+1)(d-3)$,
which equals the square-free part of $d+1$ in the case considered
here. 
Choosing the factor
$\pp=\bigl(\qd1 + 1)/2\bigr)$ of $(p)$ over the integers
of $K$ yields the ray class field $K^{\pp\jj}$, which is of degree
$(p-1)/3\ell$ over the Hilbert class field, by Proposition~\ref{hydra3}.

The main steps are as follows:
\begin{enumerate}
  \item Compute numerical real Stark units for the fields
    $K^{\pp^\tau\jj^\tau}$ and $K^{2\pp^\tau\jj^\tau}$ to sufficient
    precision and determine their exact---\ie, algebraic---minimal
    polynomials $r_1(t)$ and $r_2(t)$ over $K$.
  \item Apply the automorphism $\tau$ to obtain the polynomials
    $p_1(t)=r_1^\tau(t)$ and $p_2(t)=r_2^\tau(t)$. 
    Compute the
    factorisation of $p_1(t^2/x_0)$ and $p_2(t^2)$ in $K[t]$
    and pick factors $\tilde{p}_1(t)$ and $\tilde{p}_2(t)$.
  \item Compute exact defining polynomials for the ray class field
    $K^{2\pp\jj}$, yielding a tower of number fields that
    includes the Hilbert class field $H_K$ as well as the field
    $K^{\pp\jj}$. The field $K^{4\pp\jj}$ is obtained as a quadratic
    extension (at most) $K^{4\pp\jj}=K^{2\pp\jj}(\sqrt{-u_K})=K^{2\pp\jj}(\sqrt{x_0})$ of $K^{2\pp\jj}$.
  \item Compute the Galois group of the extension $K^{2\pp\jj}/K$,
    which by class field theory will be {abelian}. Then pick an
    automorphism $\sigma \in \Gal{K^{2\pp\jj}}{K}$ that fixes $H_K$
    and has order $(p-1)/\ell$.
  \item Compute the exact roots $\tilde{\alpha}_r$ and $\beta_r$ in
    $K^{2\pp\jj}$ of the polynomials $\tilde{p}_1(t)$ and $\tilde{p}_2(t)$.
    The square roots $\alpha_r$ of the Stark phase units of the field
    $K^{\pp\jj}$ are obtained as
    $\alpha_r=\tilde{\alpha}_r/\sqrt{x_0}$.
  \item In order to construct a fiducial vector using the Ansatz of
    Section \ref{sec:ansatz} from this data, we have to make suitable
    choices for the remaining free parameters.

    We have already fixed an automorphism $\sigma$ and an element
    $\beta_0$, defining the cyclic orbit $\beta_r=\sigma^r(\beta_0)$.
    First, we have to find a suitable generator $\theta$ of the
    multiplicative group of the integers modulo $p$.  We do not need
    to test all generators $\theta$ of the cyclic group of order
    $p-1$.  It is sufficient to consider them modulo the permutation
    symmetry of order $\ell$.  Moreover, we have to determine which of
    the two factors of each of the polynomials $p_1(t^2/x_0)$ and
    $p_2(t^2)$ gives the correct sign for the square roots of the
    Stark phase units. We can choose an arbitrary fixed first element
    $\beta_0$ in the large cycle of conjugates, but we have to find a
    correlated first element $\alpha_0$ of the small cycle. Finally,
    we have to test which of the two possibilities \eqref{eq:Zauner1}
    and \eqref{eq:Zauner2} for the symmetry leads to a fiducial
    vector. The total number of choices are
    $\varphi\bigl((p-1)/\ell\bigr)$, $2^2=4$, $\deg K^{\pp\jj}/K=\h(p-1)/3\ell$,
    and two, respectively. In brief: 
    \begin{enumerate}
      \item pick a $\theta\in(\Z/p\Z)^\times$ modulo the $\ell$-fold symmetry group
      \item pick signs $s_0,s_1\in\{1,-1\}$
      \item pick one of the square roots of the Stark phase units in
        $K^{\pp\jj}$ as $\alpha_0$
      \item for each choice of $\theta$, $\alpha_0$, and $s_0,s_1$,
        using
        \begin{equation}
          x_{\theta ^r} =s_0\alpha_r=s_0\sigma^r(\alpha_0)
          \qquad\text{and}\qquad
          y_{\theta ^r} =s_1\beta_r=s_1\sigma^r(\beta_0).\label{eq_sigma_theta}
        \end{equation}
        (cf. eq.~\eqref{eq:components}) we obtain the vectors ${\bf v}_1$
        and ${\bf v}_4$ of \eqref{eq:v0_v1}, which in combination with
        \eqref{eq:symmetri} yields a candidate for the fiducial
        vector $\ket{\Psi}$ for which we test the overlap $\bra{\Psi}I_4\otimes X^{(p)}\ket{\Psi}$
     \item test whether 
        $({\bf v}_1^{\rm T},{\bf v}_2^{\rm T},{\bf v}_3^{\rm T},{\bf v}_4^{\rm T})^{\rm T}$ or
        $({\bf v}_1^{\rm T},{\bf v}_3^{\rm T},{\bf v}_2^{\rm T},{\bf v}_4^{\rm T})^{\rm T}$
        corresponding to \eqref{eq:Zauner1} and \eqref{eq:Zauner2}
        yields a fiducial vector
    \end{enumerate}
\end{enumerate}

In Table \ref{tab:solutions}, we provide the main data for the
dimensions $d=n^2+3=4p$ for which we have computed an exact fiducial
vector. Table \ref{tab:timings} provides information on the time for
some of the computational steps.

\begin{table}[hbt]
  \caption{Dimensions of the form $d=n^2+3=4p$ for which we have
    computed a fiducial vector. With the exception of $d=12$ treated
    in Section \ref{sec:twelve}, the list is complete up to $n =
    67$. The factorisation of the degree of the ray class field
    $K^{2\pp\jj}$ over $K$ corresponds to the cyclic factors of prime-power order of the
    Galois group $\Gal{K^{2\pp\jj}}{K}$.  We need a precision to approximately twice as many
    digits as the value given in the column `size' defined in eq.~\eqref{eq:size}.  The number of combinations we have to test is
    $4\times (\deg K^{2\pp\jj}/K)/3\times
    \varphi\bigl((p-1)/\ell\bigr)$.\label{tab:solutions}}
  \let\t\times
  \def\arraystretch{1.2}
  \arraycolsep0.75\arraycolsep
  \begin{small}
  \centerline{$
    \begin{array}{|r|r@{{}=4\times{}}l|c|c|r@{{}={}}l|c|c|c|c|c|c|c|c|c|c|c|c|c|c|c|c|}
      \hline
      n & d & p & D & \ell & \multicolumn{2}{c|}{\deg K^{2\pp\jj}/K} & \h&\text{size}& \varphi\bigl((p-1)/\ell\bigr)\\
      \hline
      &\multicolumn{2}{c|}{}&&&\multicolumn{2}{c|}{}&&&\\[-3ex]
      \hline
      5 &   28 &  7  &   29 &  1 &    6  & 2\t 3                     & 1 &    3 &    2\\\hline
      7 &   52 &  13 &   53 &  1 &   12  & 2^2\t 3                   & 1 &    7 &    4\\\hline
     11 &  124 &  31 &    5 &  5 &    6  & 2\t 3                     & 1 &    3 &    2\\\hline
     13 &  172 &  43 &  173 &  1 &   42  & 2\t3\t 7                  & 1 &   42 &   12\\\hline
     17 &  292 &  73 &  293 &  1 &   72  & 2^3\t 3^2                 & 1 &   94 &   24\\\hline
     25 &  628 & 157 &  629 &  1 &  312  & 2\t 2^2\t 3\t 13          & 2 &  402 &   48\\\hline
     29 &  844 & 211 &    5 &  7 &   30  & 2\t 3\t 5                 & 1 &   16 &    8\\\hline
     31 &  964 & 241 &  965 &  1 &  480  & 2\t 2^4\t 3\t 5           & 2 &  734 &   64\\\hline
     35 & 1228 & 307 & 1229 &  1 &  918  & 2\t 3^3\t 17              & 3 & 1288 &   96\\\hline
     41 & 1684 & 421 & 1685 &  1 &  840  & 2\t 2^2\t 3\t 5\t 7       & 2 & 1667 &   96\\\hline
     43 & 1852 & 463 & 1853 &  1 &  924  & 2\t 2\t 3\t 7\t 11        & 2 & 1829 &  120\\\hline
     49 & 2404 & 601 & 2405 &  1 & 2400  & 2\t 2\t 2^3\t 3\t 5^2     & 4 & 3686 &  160\\\hline
     55 & 3028 & 757 & 3029 &  1 & 3024  & 2^2\t 2^2\t 3^3\t 7       & 4 & 5217 &  216\\\hline
     67 & 4492 &1123 & 4493 &  1 & 3366  & 2\t 3^2\t 11\t 17         & 3 & 7465 &  320\\\hline
    139 &19324 &4831 &  773 &  1 & 4830 & 2\t 3\t 5\t 7\t 23        & 1 & 10815 & 1056 \\\hline
    199 &39604 &9901 &    5 & 11 &  900 & 2^2\t 3^2\t 5^2            & 1 &   577 & 240 \\\hline
  \end{array}
  $}
  \end{small}
\end{table}

\begin{table}[hbt]
  \caption{Run-time for the calculation of the defining polynomial for
    the number field $K^{2\pp\jj}$ and for the numerical Stark units.  Most of
    the timings are for Magma.  Timings for the number field using
    Hecke and timings for numerical Stark units using Pari/GP are
    marked with ${}^*$. Note that we have not always used the very
    same algorithm nor the same computer.\label{tab:timings}}
  \let\t\times
  \def\s{\text{ s}\ }
  \def\ss{\text{ s}\rlap{${}^*$}\ }
  \def\arraystretch{1.2}
  \arraycolsep0.75\arraycolsep
  \begin{small}
  \centerline{$
    \begin{array}{|r|r|r@{{}=}l|r|r|r|r|c|c|c|c|c|c|c|c|c|c|c|c|c|c|c|}
      \hline
      n & d & \multicolumn{2}{c|}{\deg K^{2\pp\jj}/K} &\text{number field}&\text{precision} & \text{Stark units} \\
      \hline
       &&\multicolumn{2}{c|}{}&&&\\[-3ex]
      \hline
      5 &   28 &    6 & 2\t 3                 &   0.050\s &   100 \ &   0.920\s \\\hline
      7 &   52 &   12 & 2^2\t 3               &   0.110\s &   100 \ &   2.110\s \\\hline
     11 &  124 &    6 & 2\t 3                 &   0.050\s &   100 \ &   0.860\s \\\hline
     13 &  172 &   42 & 2\t 3\t 7             &   1.250\s &   150 \ &  27.660\s \\\hline
     17 &  292 &   72 & 2^3\t 3^2           &   1.770\s &   200 \ & 118.580\s \\\hline
     25 &  628 &  312 & 2\t 2^2\t 3\t 13      & 110.660\s &   900 \ &\text{$72$ min}^* \\\hline
     29 &  844 &   30 & 2\t 3\t 5             &   0.080\s &   100 \ & 7.390\s \\\hline
     31 &  964 &  480 & 2\t 2^4\t 3\t 5       &  10.700\s &  1500 \ & \text{$26$ days}\\\hline
     35 & 1228 &  918 & 2\t 3^3\t 17          & 654.100\ss&  2600 \ & \text{$26$ days}\\\hline
     41 & 1684 &  840 & 2\t 2^2\t 3\t 5\t 7   &   7.250\s &  3500 \ & \text{$194$ hours}^*\\\hline
     43 & 1852 &  924 & 2\t 2\t 3\t 7\t 11    &  49.700\s &  4000 \ & \text{$406$ hours}^*\\\hline
     49 & 2404 & 2400 & 2\t 2\t 2^3\t 3\t 5^2 & 440.670\ss &  7500 \ &\text{$881$ days}^*\\\hline
     55 & 3028 & 3024 & 2^2\t 2^2\t 3^3\t 7   &  298.680\ss & 10700 \ &\text{$19.1$ years}^* \\\hline
     67 & 4492 & 3366 & 2\t 3^2\t 11\t 17    &  274.378\ss  & 15200 \ &\text{$70.5$ years}^*\\\hline
    139 &19324 & 4830 & 2\t 3\t 5\t 7\t 23   &  762.586\ss &22000\  &\text{$328$ years}^* \\\hline
    199 &39604 &  900 & 2^2\t 3^2\t 5^2       &  11.560\s  & 2000\ &\text{$57$ days}\\\hline
  \end{array}
  $}
  \end{small}
\end{table}

\subsection{Comments on individual steps}
\subsubsection{Computing numerical real Stark units}
Stark conjectured \cite{StarkIII,StarkIV} that for certain ray class
fields $K^{\mm\jj}$ with totally real base field $K$, there are
algebraic units $\epsilon$ which can be obtained from the derivative
of zeta functions at zero via
\begin{alignat}{5}
  \epsilon_{\mathfrak{c}} = \exp(2\zeta'(0,\mathfrak{c}))\label{eq:exp_Stark}
\end{alignat}
(we include
a factor $2$ to match our conventions below).
Here $\mathfrak{c}$ is a representative ideal of an ideal class in the ray class group 
for $K$ with modulus $\mm\jj$.
Since by class field theory this group is
isomorphic to the Galois group $G=\Gal{K^{\mm\jj}}{K}$, we may equally well use the
elements $\sigma\in G$ of the Galois group to label the Stark
units. Computing numerical approximations $\tilde{\epsilon}_\sigma$ of
those Stark units for all elements of the ray class group, we obtain a
complete set of conjugates. Then the coefficients $\tilde{c}_i$ of the
polynomial
\begin{alignat}{5}
  \tilde{f}(t)=\prod_{\sigma\in G}(t-\tilde{\epsilon}_\sigma)=\sum_{i}\tilde{c}_it^i
\end{alignat}
are numerical approximations of integers $c_i$ in $K$. Defining the
size of the polynomial $\tilde{f}(t)$ as
\begin{alignat}{5}
  \text{size}\bigl(\tilde{f}(t)\bigr)=\max_i \log_{10} |\tilde{c}_i|,\label{eq:size}
\end{alignat}
we observe that a numerical precision of about twice the size (in number of 
digits) is sufficient to obtain the exact values of the coefficients
$c_i$ using an integer relation algorithm.

Stark relates the zeta function in eq.~\eqref{eq:exp_Stark} to Hecke
$L$-series attached to Hecke characters $\psi$ via
\begin{alignat}{5}
  L(s,\psi)=\sum_{\mathfrak{c}}
  \psi(\mathfrak{c})\zeta(s,\mathfrak{c}),\label{eq:L_zeta}
\end{alignat}
where we have dropped a factor $1/2$ in relation to the expression
given on p.~66 of \cite{StarkIII}. We can obtain the values of the
zeta function from the values of the $L$-series inverting
eq.~\eqref{eq:L_zeta} using the orthogonality relations for
characters.  For any Hecke character $\psi$, there is an associated
primitive character \cite[ch. 16, \S4,
  Definition~4.4]{Husemoeller}. Following \cite{StarkIII}, denoting
the primitive character by $\psi^*$, we have
\begin{alignat}{5}
   L'(0,\psi)
     = L'(0,\psi^*)\prod_{\mathfrak{p}|\mm}\bigl(1-\psi^*(\mathfrak{p})\bigr),\label{eq:Euler_factors}
\end{alignat}
where the inverse Euler factor product is over all finite primes
$\mathfrak{p}$ dividing into the modulus $\mm$.  If $K^{\mm_1\jj}\le
K^{\mm\jj}$ is a subfield with the finite part $\mm_1$ of the modulus
dividing $\mm$, then the primitive characters corresponding to the
Hecke characters of $K^{\mm_1\jj}$ are a subset of those of
$K^{\mm\jj}$. Therefore, it suffices to compute the derivative of the
$L$-series at zero for all primitive characters of the largest
field. Then, for each field we can use eq.~\eqref{eq:Euler_factors} to
compute the corresponding values $L'(0,\psi)$. From those, we
eventually obtain the numerical approximations of the Stark units for
both $K^{2\pp^\tau\jj^\tau}$ and $K^{\pp^\tau\jj^\tau}$ using
eq.~\eqref{eq:Euler_factors}, inverting eq.~\eqref{eq:L_zeta}, and
using eq.~\eqref{eq:exp_Stark}.

In our situation, we compute the values $L'(0,\psi^*)$ for all
associated primitive Hecke characters of the field
$K^{2\pp^\tau\jj^\tau}$.  Algorithms for this are available,
\eg, in the computer algebra systems Magma and Pari/GP \cite{PARI2},
with the latter providing a more advanced parallel
implementation. Timings for these calculations are shown in the last
column of Table \ref{tab:timings}.

\subsubsection{Minimal polynomials of the square roots}
For the square roots of Stark phase units, we consider the following
factorisation of the minimal polynomials $p_i(t)$ of the Stark phase
units after some quadratic substitution, \ie,
\begin{alignat}{5}
  p_1(t^2/x_0) &{}= \tilde{p}_1(t)\tilde{p}_1(-t)\label{eq:factor_p0}\\
  \noalign{and}
  p_2(t^2) &{}= \tilde{p}_2(t)\tilde{p}_2(-t).\label{eq:factor_p1}
\end{alignat}
Note that \textit{a priori}, we cannot distinguish between
$\tilde{p}_i(t)$  and $\tilde{p}_i(-t)$.  This requires that we determine
the signs $s_i$ in our final step. 

The first factorisation \eqref{eq:factor_p0} for the square roots
$\alpha_r$ is related to the result of Proposition \ref{hydra6} that all
products $\alpha_r\sqrt{x_0}$ lie in the ray class field $K^{\pp\jj}$,
which implies that one of the factors $\tilde{p}_i(t)$ is their
minimal polynomial over $K$.  For the second factorisation
\eqref{eq:factor_p1} we have observed in all cases thus far  
that such a factorisation always exists over the field $K$, which
implies that the square roots $\beta_r$ actually lie in the same field
$K^{2\pp\jj}$ as the Stark phase units.

There are some possible variations here.  As stated, the algorithm
uses a rescaling of the Stark phase units in $K^{\pp\jj}$ by a factor
$\sqrt{x_0}$ before the minimal polynomial of their square roots is
calculated.  For the Stark phase units in $K^{2\pp\jj}$ no such
rescaling is necessary since we find their square roots to lie in
$K^{2\pp\jj}$. We can do without the rescaling also for the smaller
field if we factor the polynomials $p_i(t^2)$ over $K(\sqrt{x_0})$, at
the price of having to perform the calculations below in the larger
field $K^{4\pp\jj}$. Using rescaling, there are a few natural scale
factors to choose from, such as $\sqrt{x_0}$ and $\sqrt{-u_K}$.  The
latter choice leads to somewhat smaller coefficients in the minimal
polynomials for the square roots of the rescaled Stark phase units in
the subfield $K^{\pp\jj}$.  However, we decided to use $\sqrt{x_0}$ in
the discussion here, as this is directly related to the zeroth
component of the vector ${\bf v}_4$.

\subsubsection{Computing defining polynomials for $K^{2\pp\jj}$}
The computer algebra system Magma supports the calculation of defining
polynomials for ray class fields.  It computes a defining polynomial
corresponding to each cyclic factor of the Galois group (cf. the
factorisation of the degree of $K^{2\pp\jj}/K$ in Table
\ref{tab:timings}).  The run-time depends mainly on the largest cyclic
factor.  For the cases where the largest factor was at least $17$, we
used the computer algebra system Hecke \cite{hecke} instead, as it
provides a more advanced algorithm for this task.  As the degree of
the fields in our examples is relatively smooth---with the largest
cyclic factor being $27$---the run-time for this step is less than
$15$ minutes.

Where the degree of the ray class field contains prime-power cyclic factors, 
we computed successive extensions wherein each was 
of prime degree over the previous level, yielding an overall tower of
number fields with successive extensions having prime degree.  
In some cases,
we used various additional methods to find defining polynomials
with smaller coefficients.  Heuristically, this results in faster
arithmetic in the number fields.

\subsubsection{Computing the Galois group}
Computer algebra systems like Magma have built-in algorithms to
compute the group of automorphisms of number fields.  In our case, we
can make use of the tower of number fields with successive relative
extensions of prime degree referred to above.  Assume we have an
extension $F_{i+1}=F_i(q_{i})$, with say $f_i(x)$ being the minimal
polynomial of $q_i$ over $K_i$.  An automorphism of $F_{i+1}$ that
fixes $F_i$ maps $q_i$ to some other root of $f_i(x)$: \ie, we have to
find all roots of the polynomial $f_i(x)$, which is of small prime
degree in our case, and all roots lie in $F_{i+1}$.  If $\pi$ is a
non-trivial automorphism of $F_i$, then $q_i$ can be mapped to any
root of the polynomial $f_i^\pi(x)$, which is obtained by applying
$\pi$ to the components of $f_i(x)$.  We also note that the defining
polynomials for each cyclic factor have coefficients in the quadratic
field $K=\QD$, by the fundamental theorem of abelian groups.  Hence we
are able to restrict ourselves to finding roots in fields of
relatively small degree.

When we use this approach to compute the Galois group, we get all
automorphisms, but their group structure is only implicit.  In order
to find an automorphism $\sigma$ of order $(p-1)/\ell$, we make a
random choice and compute the order.  Note that such an automorphism
$\sigma$ is guaranteed to exist in $\Gal{K^{2\pp\jj}}{H_K}$ by
Proposition \ref{hydra4}. We can lift $\sigma$ to an automorphism of
$K^{4\pp\jj}/H_K$ which satisfies $\sigma(\sqrt{x_0})=\sqrt{x_0}$,
because as illustrated in Figure~\ref{fig:subfield_lattice},
$LK^{\pp\jj}$ and $K^{2\pp\jj}$ are linearly disjoint over
$K^{\pp\jj}$, and their compositum is $K^{4\pp\jj}$,
so $\Gal{K^{4\pp\jj}}{K}$ is a direct product of the $C_2$-group
$\Gal{K^{4\pp\jj}}{K^{2\pp\jj}}$ with $\Gal{K^{2\pp\jj}}{K}$.

We fix an embedding of our tower of number fields into the complex
numbers that maps $\sqrt{D}$ to a positive number.  Then we can
identify an element of the Galois group of order two that acts like
complex conjugation; see the definition of $\sigt$ in
Section~\ref{gammaray}.  For the Stark phase units, we know that this
automorphism maps them to their inverse. So we do not need that
automorphism when working with Stark phase units. But when we
transform the fiducial vector to the standard basis, it is more
convenient to know how complex conjugation acts on general elements of
our number fields.

\subsubsection{Computing exact roots $\alpha_0$ and $\beta_0$}\label{sec:exact_roots}
A more complex step is the determination of the roots of the minimal
polynomials $\tilde{p}_i(t)$.  The degree of those polynomials is
$\h(p-1)/3\ell$ and $\h(p-1)/\ell$, respectively.  By Propositions
\ref{hydra5} and \ref{hydra7}, those roots lie in the number field
$K^{4\pp\jj}$. While for small dimension we can make use of standard
algorithms, for larger examples we use an approach that is adjusted to
our situation.  First, since we have computed the Galois group and we
expect all roots to lie in the corresponding ray class field, it
suffices to compute a single root.  We compute the factorisation in
steps aligned with the tower of number fields. In each step, the
degree of the factors is reduced by a prime number.  It is sufficient
to continue with only one of the factors in each step.  For those who
are interested in more details, we use a Trager-like algorithm and use
modular techniques to compute the required greatest common divisors of
polynomials.

\subsubsection{Search for the remaining parameters}\label{sec:final_steps}
As already stated in Step 6 of our algorithm in Section \ref{sec:recipe},
the Ansatz of Section \ref{sec:ansatz} for the fiducial vector leaves
a few choices.

Eq. \eqref{eq:components} that relates the
sequences of square roots $\alpha_r$ and $\beta_r$ to the components
$x_{\theta^r}$ and $y_{\theta^r}$ of the vectors ${\bf v}_4$ and ${\bf
  v}_1$, respectively, requires that we choose an element $\theta$ of order
$p-1$ in the integers modulo $p$. For $\ell>1$, there will be an
additional permutation symmetry of order $\ell$ among those elements,
and it suffices to test the candidates for $\theta$ modulo that symmetry.

For each of the sequences $\alpha_r$ and $\beta_r$, we have to pick a
first element $\alpha_0$ and $\beta_0$ in the orbit with respect to
the automorphism $\sigma$.  We can fix an arbitrary element $\beta_0$
for the larger of the two orbits, since any other choice will result
in a fiducial vector as well, provided $\alpha_0$ was chosen
accordingly.  For $\alpha_0$, we test all possible choices.  The same
is true for the choice of signs $s_0$ and $s_1$ (but see the next
subsection).

For each choice of $\theta$, we get the permutation matrix $U_F$
acting on $\C^p$ using $\delta=\theta^{(p-1)/3}$ in
eq.~\eqref{eq:diagonal}.  We compute the vectors ${\bf v}_4$ and ${\bf
  v}_1$ given in eq.~\eqref{eq:v0_v1}, as well as ${\bf
  v}_2=U_F^{-1}{\bf v}_1$ and ${\bf v}_3=U_F {\bf v}_1$.  Those
vectors are combined to yield a non-normalised candidate vector $({\bf
  v}_1^{\rm T},{\bf v}_2^{\rm T},{\bf v}_3^{\rm T},{\bf v}_4^{\rm
  T})^{\rm T}$.

We compute the overlap with respect to $\openone_4\otimes X^{(p)}$ which
(ignoring normalisation) equals
\begin{alignat}{5}
  \sum_{i=0}^3 \bra{{\bf v}_i} X^{(p)}\ket{{\bf v}_i}.\label{eq:baby_test}
\end{alignat}
(Note that in order to compute this inner product, we need the
automorphism that acts as complex conjugation.)  This is nothing but
the sum of the inner product of the vectors ${\bf v}_i$ with their
shifted versions which can be easily computed in the corresponding
number field. The drawback is that eq.~\eqref{eq:baby_test} does not
allow us to discriminate between the two possibilities $U_F$ and
$U_F^{-1}$ for the symmetry in eqs.~\eqref{eq:Zauner1} and
\eqref{eq:Zauner2}.  Recall that replacing $U_F$ by $U_F^{-1}$ results
in swapping ${\bf v}_2$ and ${\bf v}_3$. In order to discriminate
between these two options, we have to test additional overlaps.  As
indicated in Section \ref{sec:example}, using $D_{i,j}^{(4)}\otimes
X^{(p)}$ for any displacement operator $D_{i,j}^{(4)}$ in dimension
$4$ is inconclusive.  We would have to consider more general
displacement operators in dimension $p$, which would then in turn
require us to use $p$th roots of unity and hence perform calculations
in a field of much larger degree.  Instead, we transform the two
candidates for the fiducial vector to the basis of the standard
representation of the Weyl--Heisenberg group (inverting the
transformation \eqref{eq:basis_change}) and test the conditions there.
The transformation requires calculations in the field
$K^{8\pp\jj\jj^\tau}$, adjoining an eighth root of unity.  Then we
test the quantities $G(i,k)$ defined in eq.~\eqref{Gik1} in Section
\ref{sec:verification} below.  We observe that testing just one
$G(i,k)$ with $i,k$ co-prime to $d$, say $G(1,1)$, is usually
sufficient for us to determine which of $U_F$ or $U_F^{-1}$ we need.

\subsection{Some improvements}\label{sec:signs} 
The fact that the final step of our algorithm involves a search over a
finite number of undetermined choices is clearly a weakness.  However,
if one is willing to make further assumptions, it is possible to
reduce the search.  In particular we can determine the generator
$\theta\in(\Z/p\Z)^\times$ of the non-zero integers modulo $p$, as
well as the signs $s_0$ and $s_1$, without checking the SIC property.
Notice that when the dimension is prime these are the only choices one
has to make \cite{ABGHM}.  The algorithm as described in Section
\ref{sec:recipe} does not depend on these additional assumptions; but
we have observed them to hold in all of the cases which we have examined.

To determine $\theta$, first recall that we can use elements of the
Galois group to label the numerical real Stark units.  What is more,
we have identified a mapping from the non-zero integers modulo $d$ to
numerical Stark units that is compatible with the action of the Galois
group as well as with multiplication in $\Z/d\Z$, based upon the
natural identification of ideals modulo $d$ in $\Z$ with those modulo
$2\pp$ in $\zk$.  This allows us to associate each non-zero index in
the fiducial vector with a numerical real Stark unit.  Given the exact
expression for the square roots of the Stark phase units from
Section \ref{sec:exact_roots}, we can choose a real embedding of
the ray class field and match the exact square roots with the
numerical Stark units. From the action of the Galois transformation
$\sigma$ on the exact square roots, we can derive the corresponding
permutation of the indices in the fiducial vector.  Matching the exact
and the numerical values allows us not only to determine the generator
$\theta$, but also the first elements $\alpha_0$ and $\beta_0$ in the
orbits.  We hope to return to this issue and discuss it in detail
elsewhere.

The signs can be fixed very simply, if one relies on an observation
made in Section \ref{sec:overlaps}.  We give the details there. Using
that method, we can reduce the overall run-time of our search by what amounts 
essentially to a factor of four.

\subsection{Complete verification}\label{sec:verification}
So far, we have not been able to prove that our algorithm always
works. Hence we have to rely on the complete verification of the SIC
conditions for the exact SIC fiducial vectors that we have calculated.
This is clearly a major task since the number of overlaps grows
quadratically with the dimension. Fortunately simplifications are
possible.  Let us first recall how the use of the standard
representation of the Weyl--Heisenberg group allows us to do the
entire calculation in the number field generated by the components of
the fiducial vector \cite{Roy, Mahdad, ADF}.  The idea is to take a
discrete Fourier transform of the sequences of SIC conditions, and
calculate 
\begin{equation}
  G(i,k) 
  :=
  \frac{1}{d}
    \sum_j\omega^{kj}|\bra{\Psi}X^iZ^j\ket{\Psi}|^2
  =
  \sum_{r=0}^{d-1}\bar{a}_{r+i}\bar{a}_{r+k}a_{r}a_{r+i+k}, \label{Gik1}
\end{equation}
where $a_r$ denotes a component of $\ket{\Psi}$ in the chosen basis,
and $\omega=e^{2\pi i/d}$ is a primitive complex $d$th root of unity.
This works because the matrix representation of the operator that
occurs here is
\begin{equation}
  (X^iZ^j)_{r,s} = \omega^{js}\delta_{r,s+i},. 
\end{equation}
where $\delta$ is the Kronecker delta. 
The invertibility of the discrete Fourier transform then implies that
we can rephrase the SIC conditions \eqref{eq:SIC_conditions} in terms
of $G(i,k)$:
\begin{equation}
  |\bra{\Psi}X^iZ^j\ket{\Psi} |^2
  = \frac{d\delta_{i,0}\delta_{j,0}+1}{d+1}
  \quad \Longleftrightarrow \quad 
  G(i,k) = \frac{\delta_{i,0} + \delta_{k,0}}{d+1}. \label{Gik2}
\end{equation}
Since $G(i,k)$ can be expressed in terms of the components of the
fiducial vector in the standard basis, we only need eighth roots of
unity, but not $p$th roots of unity when checking that this condition
holds.  This is the approach which we applied in \cite{ABGHM}, where we
do not need complex roots of unity at all. In Appendix \ref{sec:Gik} we 
describe a significant improvement to this approach. 

A critical question is: what is the minimum number of sufficient conditions? 
Once this question is answered we find, somewhat counterintuitively, that 
we can reduce the overall complexity of
the verification when we consider the overlap phases in
eq.~\eqref{eq:SIC_conditions}, which requires calculations in a number
field that contains $d$th roots of unity (with $d=4p$) and which has
an even larger degree. We have to adjoin both a fourth and a $p$th root
of unity to the field $K^{4\pp\jj}$ in which the fiducial vector in
the adapted basis can be expressed.  We consider the overlap phases of
the fiducial vector in our adapted basis
\begin{alignat}{5}
  \qd1\bra{\Psi_0}D_{i,j}^{(4)}\otimes X^a Z^b\ket{\Psi_0},\label{eq:overlaps}
\end{alignat}
where $D_{i,j}^{(4)}$ is an element of the Weyl--Heisenberg group in
our representation in dimension four, and $X^aZ^b$ is in the
Weyl--Heisenberg group in the standard representation in dimension
$p$.  For simplicity we have not included the cyclotomic phase factors
in the displacement operators acting on the dimension $p$ factor in
eq.~\eqref{eq:overlaps}.  They do become relevant in Section
\ref{sec:overlaps}, where we discuss the actual numbers that
constitute the overlaps.

We can use the action of the Galois group on eq.~\eqref{eq:overlaps} to
reduce the number of overlap phases that we have to check.
Since the cyclotomic field generated by the $p$th root of unity is
disjoint to the field $K^{4\pp\jj\jj^\tau}$ containing the fiducial
vector and the displacement operators $D_{i,j}^{(4)}$, for every $b\ne
0$ there is a Galois automorphism that maps $X^a Z^b$ to $X^a Z$ and 
changes neither the fiducial vector nor $D_{i,j}^{(4)}$.  The
modulus of the overlap in eq.~\eqref{eq:overlaps} will be the
same. Hence it is sufficient to consider the exponents $b=0$ and $b=1$
in \eqref{eq:overlaps}.  Moreover, applying the automorphism $\sigma$
to the fiducial vector multiplies the indices in the dimension-$p$
component by $\theta$ (see eq.~\eqref{eq_sigma_theta}). This implies
\begin{alignat}{5}
  \sigma\left(\qd1\bra{\Psi_0}D_{i,j}^{(4)}\otimes X^a Z^b\ket{\Psi_0}\right)
  =\qd1\bra{\Psi_0}D_{i,j}^{(4)}\otimes X^{a\theta^{-1}} Z^{b\theta^{-1}}\ket{\Psi_0}.
\end{alignat}
As $\theta$ is a primitive element of $\Z/p\Z$, it suffices to
consider the exponents $a=0$ and $a=1$ in eq.~\eqref{eq:overlaps}.

To reduce the number of displacement operators that we have to
consider in the dimension four component, we use the action of the
pre-ascribed Zauner symmetry $U_{\mathcal{Z}}$ (see
eq. \eqref{eq:UZauner}) of our fiducial vector on the operators
$D_{i,j}^{(4)}$ given in eqs.~\eqref{eq:X_adapted} and
\eqref{eq:Z_adapted}.  By definition, the additional symmetry does not
change the fiducial vector, but we can use it to transform the
displacement operator.  Additionally, we note that replacing an
operator $D_{i,j}^{(4)}\otimes X^a Z^b$ in eq.~\eqref{eq:overlaps} by its
adjoint-up-to-phase $D_{-i,-j}^{(4)}\otimes X^{-a} Z^{-b}$ does not change the
modulus.  The same is true for complex conjugation.  It turns out that
it is sufficient to consider the operators $D_{i,j}^{(4)}$ for
$(i,j)=(0,0),(0,1),(0,2)$.  Note that $D^{(4)}_{0,0}=\openone_4$,
$D^{(4)}_{0,2}=\openone_2\otimes \sigma_z$ (see
eq. \eqref{eq:diag_ops}), and $D^{(4)}_{0,1}$ is given in
eq. \eqref{eq:Z_adapted}.

In total, we have to check only only $12$ out of the $d^2$ conditions
for our candidate fiducial vectors.  Four of the overlaps, namely
those for the diagonal operators $D^{(4)}_{0,0}\otimes X^0Z^0$,
$D^{(4)}_{0,0}\otimes X^0Z^1$, $D^{(4)}_{0,2}\otimes X^0Z^0$, and
$D^{(4)}_{0,2}\otimes X^0Z^1$ are already implied by our Ansatz (see
eqs.~\eqref{baby1} and \eqref{baby2}).
Hence we only have to check eight overlap phases, but in number fields
of very large degree.  More details on these overlap phases can be
found in Section \ref{sec:overlaps} and Table \ref{tab:baby4p}.

We have verified the solutions for dimensions up to $d=1852$ using
exact arithmetic.  For larger dimensions, for which the degree of the
required number field is also much larger, we have checked those eight
conditions numerically.  Timings are given in
Table~\ref{tab:verification}.  For the exact verification, we state
the time it took to compute the squared modulus of the eight overlaps.
We note that the most time-consuming operation is not to compute the
overlap phase, but to multiply it with its conjugate value.  Recall
that the degree of the number field including the $d$th root of unity
scales as $O(d^2)$.  For the numerical verification, we first have to
compute a numerical fiducial vector from the exact one. Mapping the
elements of the high-degree number field to high-precision floating
point values takes considerable time. Hence, for these cases, we state
the total time to compute a numerical fiducial vector from the exact
one and to compute the numerical overlaps.  The time for computing the
overlaps and their absolute value is given in parenthesis.  For
comparison, we also give the run-time to compute a single term
$G(i,k)$ for fixed $i$, $k$ using exact arithmetic in the last column.
For the last three dimensions we only provide estimates.

\begin{table}[hbt]
  \caption{Run-time for the verification of the solutions. The column
    $G(i,k)$ provides information on the time required to compute one
    value $G(i,k)$ using exact arithmetic. For $d=4492$, $d=19324$,
    and $d=39604$, the times for $G(i,k)$ are estimates.  Note that we
    did not always use the same computer.\label{tab:verification}}
  \let\t\times \def\s{\text{ s}\ } \def\ss{\text{ s}\rlap{${}^*$}\ }
  \def\arraystretch{1.2} \arraycolsep0.75\arraycolsep
  \begin{small}
  \centerline{
    \begin{tabular}{|r|r|r@{${}={}$}l|r|c@{\;}c|c|c|r|c|c|c|c|c|c|c|c|c|c|c|}
      \hline
      $n$ & $d$ & \multicolumn{2}{c|}{$\deg K^{4\pp\jj}/K$} &\multicolumn{1}{c|}{precision}&\multicolumn{2}{c|}{CPU time} & \multicolumn{1}{c|}{$G(i,k)$} \\
      \hline
       &&\multicolumn{2}{c|}{}&&&&\\[-3ex]
      \hline
      $5$ &    $28$ &   $12$ & $2\t 2\t 3$                 & \multicolumn{1}{c|}{exact} & \multicolumn{2}{c|}{$<1$ s} & $<1$ s  \\\hline
      $7$ &    $52$ &   $24$ & $2\t 2^2\t 3$               & \multicolumn{1}{c|}{exact} & \multicolumn{2}{c|}{$<1$ s} & $<1$ s  \\\hline
     $11$ &   $124$ &   $12$ & $2\t 2\t 3$                 & \multicolumn{1}{c|}{exact} & \multicolumn{2}{c|}{$<1$ s} & $<1$ s  \\\hline
     $13$ &   $172$ &   $84$ & $2\t 2\t 3\t 7$             & \multicolumn{1}{c|}{exact} & \multicolumn{2}{c|}{$31.5$ s} & $5.4$ s \\\hline
     $17$ &   $292$ &  $144$ & $2\t 2^3\t 3^2$             & \multicolumn{1}{c|}{exact} & \multicolumn{2}{c|}{$264.3$ s} & $37.8$ s  \\\hline
     $25$ &   $628$ &  $624$ & $2\t 2\t 2^2\t 3\t 13$      & \multicolumn{1}{c|}{exact} & \multicolumn{2}{c|}{$207.8$ min} & $19$ min \\\hline
     $29$ &   $844$ &   $60$ & $2\t 2\t 3\t 5$             & \multicolumn{1}{c|}{exact} & \multicolumn{2}{c|}{$205.0$ s} & $12.7$ s\\\hline
     $31$ &   $964$ &  $960$ & $2\t 2\t 2^4\t 3\t 5$       & \multicolumn{1}{c|}{exact} & \multicolumn{2}{c|}{$24.8$ h} & $101$ min\\\hline
     $35$ &  $1228$ &  $1836$ & $2\t 2\t 3^3\t 17$         & \multicolumn{1}{c|}{exact} & \multicolumn{2}{c|}{$4.3$ days} & $9.0$ h \\\hline
     $41$ &  $1684$ & $1680$ & $2\t 2\t 2^2\t 3\t 5\t 7$   & \multicolumn{1}{c|}{exact} & \multicolumn{2}{c|}{$5.5$ days} & $6.8$ h\\\hline
     $43$ &  $1852$ & $1848$ & $2\t 2\t 2\t 3\t 7\t 11$    & \multicolumn{1}{c|}{exact} &  \multicolumn{2}{c|}{$7.7$ days} & $8.7$ h\\\hline
     $49$ &  $2404$ & $4800$ & $2\t 2\t 2\t 2^3\t 3\t 5^2$ &  $100\,000$ digits & $61.2$ min &($154.8$ s) &$89.5$ h \\\hline
     $55$ &  $3028$ & $6048$ & $2\t 2^2\t 2^2\t 3^3\t 7$   &  $100\,000$ digits & $97.7$ min &($206.6$ s) & $204.4$ h \\\hline
     $67$ &  $4492$ & $6732$ & $2\t 2\t 3^2\t 11\t 17$    &  $100\,000$ digits & $4.9$ h &($403.1$ s) & $\approx 15.8$ days\\\hline
    $139$ & $19324$ & $9660$ & $2\t 2\t 3\t 5\t 7\t 23$   &  $100\,000$ digits & $30.4$ h &($20.2$ min)& $\approx 570$ days\\\hline
    $199$ & $39604$ & $1800$ & $2\t 2^2\t 3^2\t 5^2$       & $100\,000$ digits & $109.3$ min &($39.8$ min) & $\approx 30$ days\\\hline
    \end{tabular}}
  \end{small}
\end{table}

\section{A detailed example in dimension $52$}\label{sec:example}
We illustrate our algorithm by giving one example in full detail, and we choose $d = 52 = 4\times 13$
for this purpose. The quadratic base field is $K = \Q(a)$, where
$a = \sqrt{53}$. Other than the somewhat simplifying fact that the class number $\h = 1$, 
this dimension provides a good illustration of the
algorithm outlined in Section \ref{sec:recipe}.  The element $x_0$ is
$x_0=-2-\qd1=-2-a$. The ray class fields $K^{\pp\jj}$, $K^{2\pp\jj}$ and
$K^{4\pp\jj}$ that we need in order to construct the fiducial vector have degrees 
$4$, $12$ and $24$, respectively. This should be compared with the degree of the
full SIC field including all the relevant roots of unity, which has
degree $1152$. So the saving is considerable.
\medskip

Step 1: We compute real Stark units for the fields
$K^{\pp^\tau\jj^\tau}$ and $K^{2\pp^\tau\jj^\tau}$. The precision
we need in order to determine their minimal polynomials, $r_1(t)$
respectively $r_2(t)$, is very modest in this case (cf. Table
\ref{tab:solutions}). We find
\begin{equation}
   r_1(t) = t^4 - \frac{1}{2}(5a+39)t^3 + (21a+154)t^2 - \frac{1}{2}(5a+39)t + 1.
 \end{equation}
The polynomial $r_2(t)$ is of degree $12$, and we wait until the next
step---where the coefficients shrink somewhat---before we express it 
explicitly.
\medskip

Step 2: We apply the automorphism $\tau$, that is $a \mapsto -a$, to
the coefficients of the minimal polynomials.  The roots of the
transformed polynomials $p_1(t)=r_1^\tau(t)$ and $p_2(t)=r_2^\tau(t)$
are phase factors. We want their square roots.  For $p_2$ we obtain
them by factoring the polynomial $p_2(t^2)$ over the quadratic
field. It will factor because the square roots of the Stark phase
units in $K^{2\pp\jj}$ lie in $K^{2\pp\jj}$. One of the factors is
\begin{alignat}{5}
 \tilde{p}_2(t)= t^{12} &{}- \frac{a-1}{2}t^{11} + \frac{7a-43}{2}t^{10} + \frac{15a-113}{2}t^9 - \frac{79a-573}{2}t^8 \nonumber \\ 
&{}- \frac{53a - 391}{2}t^7 + (122a-891)t^6 - \frac{53a-391}{2}t^5 \\ 
&{}- \frac{79a-573}{2}t^4 +  \frac{15a-113}{2}t^3 + \frac{7a-43}{2}t^2 - \frac{a-1}{2}t + 1 . \nonumber
\end{alignat}
So this is the minimal polynomial whose roots $\beta_r$ are the square
roots of the Stark phase units in $K^{2\pp\jj}$.  Because we picked
one out of the two factors of the degree $24$ polynomial $p_2(t^2)$,
we have introduced one global sign ambiguity.

For $p_1(t)$ we have to deal with the fact that the square roots of
the Stark phase units in $K^{\pp\jj}$ do not lie in $K^{\pp\jj}$.  If
we rescale them with the factor $\sqrt{x_0}$ they do (see Proposition
\ref{hydra6}), so we can factor the polynomial $p_1(t^2/x_0)$ over
$K$. One of the factors is
\begin{equation}
  \tilde{p}_1(t) = t^4 + \frac{3a - 15}{2}t^3 - (4a - 41)t^2 - \frac{9a - 129}{2}t + 4a + 57.
 \end{equation} 
Its roots $\tilde{\alpha}_r$ are square roots of the rescaled Stark
phase units in $K^{\pp\jj}$. Again there is a global sign ambiguity depending
on which factor we pick.

As mentioned in Section \ref{sec:signs}, and discussed in detail in
Section \ref{sec:overlaps} below, we can determine the correct signs
already in this step, if we are willing to take equations
\eqref{eq:baby3} and \eqref{eq:baby4} from Section~\ref{sec:overlaps}
on trust. In fact the above choices of $\tilde{p}_1(t)$ and
$\tilde{p}_2(t)$ were determined in this way, and we will take this
into account when we come to Step 6 below.
\medskip

Step 3: Here we build the number fields $K^{\pp\jj}$, $K^{2\pp\jj}$, and $K^{4\pp\jj}$, by
means of a tower of field extensions of degrees equal to the prime
factors of the degree of $K^{4\pp\jj}$. Magma provides good guidance for this,
but many slight variations and improvements are possible in this
step. One possibility is
\begin{alignat}{5}
               & K(t_1),    &\quad & \text{$t_1$ being a root of $x^2 + \frac{3a-23}{2}$;}\nonumber\\
  K^{\pp\jj}  ={}& K(t_1,t_2),  && \text{$t_2$ being a root of $x^2+a+1 - \frac{10a+68}{13} t_1$;}\nonumber \\ 
  K^{2\pp\jj} ={}& K^{\pp\jj}(t_3), &&\text{$t_3$ being a root of $x^3 - 78x + 25a - 27$;}\nonumber \\ 
  K^{4\pp\jj} ={}& K^{2\pp\jj}(\xi), &&\text{$\xi$ being a root of $x^2 - x_0$.}
\end{alignat}
The field $K^{4\pp\jj}$ will not be needed until we come to Step 6 of our algorithm. 
\medskip

Step 4: The Galois group of the extension $K^{2\pp\jj}/K$ permutes the roots
of the minimal polynomials that appear in the tower leading up to
$K^{2\pp\jj}$. Notice that we want the minimal polynomials over the fixed
field $K$. For the first extension, clearly the roots are $t_1$ and
$-t_1$. For $t_2$, the minimal polynomial is a degree four polynomial
with coefficients in $K$. However, it comes to us in an already
factored form, namely
\begin{equation} p_{t_2}(x) = \left( x^2+a+1 - \frac{10a+68}{13} t_1\right) \left( x^2+a+1 + \frac{10a+68}{13} t_1 \right) . 
\end{equation} 
The first factor is the polynomial that actually appeared in the
tower, the second factor arises when we take the Galois conjugates of
its coefficients, in this case letting $t_1 \mapsto -t_1$. Hence
there are four roots, and they are the Galois conjugates of
$t_2$. Clearly there are three conjugates of $t_3$. This means that
the Galois group is a cyclic group of order $12$, containing four
generators of order $12$. We pick one and call it $\sigma$. It effects
\begin{alignat}{5}
  \sigma(t_1) &{}= - t_1,\nonumber \\
  \sigma (t_2) &{}= - \frac{1}{4}\bigl((4a+30)t_1 + 3a + 23\bigr)t_2 , \\ 
  \sigma (t_3) &{}= - \frac{1}{30}\bigl((a-1)t_3^2 - (a-41)t_3 - 52a + 52\bigr).\nonumber
\end{alignat}
It can be worked out by hand that if $t_2$ (say) is a root of
$p_{t_2}(x)$ then so is $\sigma (t_2)$, and that $\sigma^4(t_2) = t_2$
while $\sigma^3(t_3) = t_3$. In higher dimensions the computer algebra
package has to do the work.

For complex conjugation, we pick the automorphism that changes the
signs of $t_2$ and $\xi$, and fixes the other generators.
\medskip

Step 5: It is unclear whether we could calculate the roots $\tilde{\alpha}_r$ and $\beta_r$ of the polynomials 
$\tilde{p}_1(t)$ and $\tilde{p}_2(t)$ by hand, but the computer algebra package does it without apparent delay. We need only 
one root from each, because the rest can be generated using the Galois transformation $\sigma$ from the 
previous step. We pick one root from each polynomial, say  
\begin{equation}
  \tilde{\alpha}_0 = -\frac{1}{16}\Bigl(\bigl((8a+58)t_1 + 3a + 55\bigr)t_2 + (2a+18)t_1 + 6a - 30\Bigr) \label{eq:alfaval}
\end{equation}
\begin{alignat}{5}
  \beta_0 = \frac{1}{720}\Bigl(&\bigl((-2t_1+a-11)t_2 - \frac{1}{13}(10a+42)t_1 + 2a - 10\bigl)t_3^2 \nonumber \\
  &{}- \bigl(((5a+33)t_1 + 11a-11)t_2 - 4t_1 + 22a - 130\bigl)t_3  \label{eq:betaval} \\ 
  &{} - \bigl((120a + 826)t_1 + 127a + 283\bigr)t_2\nonumber\\
  &{}+ (70a + 318)t_1 - 74a + 490 \Bigr). \nonumber
\end{alignat}
We now rely on the Galois group to provide us with two ordered
sequences of roots. Since $K^{\pp\jj}$ is a subfield of $K^{2\pp\jj}$ a 
single Galois group element suffices for both sets of roots. We obtain  
\begin{equation}
 \tilde{\alpha}_r = \sigma^r (\tilde{\alpha}_0), \qquad \beta_r = \sigma^r (\beta ). 
\end{equation}
It is convenient to let $r$ range from 0 to $p-1$ in both cases, even
though this means that the sequence $\{ \tilde{\alpha}_r\}_{r=0}^{11}$
repeats itself three times.

Notice that the choice of the automorphism $\sigma$, out of the four generators of the 
group, was arbitrary.  So were the choices of starting points for the two sequences. 
\medskip

Step 6: We now have two cycles of numbers $\tilde{\alpha}_r$ and
$\beta_r$ to place in the Ansatz for a SIC fiducial vector that we
described in Section \ref{sec:ansatz}. The $\tilde{\alpha}$--cycle is
to give the components $x_i$ made from $12/3 = 4$ distinct numbers,
while the $\beta$--cycle is to give the twelve distinct components
$y_i$.  The precise relation is given by
\begin{equation} 
  j = \theta^r \bmod p \qquad \Longrightarrow \qquad x_j = \tilde{\alpha}_r,\ y_j = \beta_r.
\end{equation}
Here $\theta$ is a generator of the multiplicative group of non-zero
integers modulo $p$. When $p = 13$ this is again a cyclic group of
order $12$, with four distinct possible choices $2$, $6$, $7$, $11$ for its 
generator.

We have arrived at the search part of our algorithm. We used a short-cut
in Step 2, hence all signs are determined.  But we do not know which
of the four possible choices of $\theta$ corresponds to the arbitrary
choice $\sigma$ that we made for the generator of the Galois
group. The arbitrary choice that we made for the starting point of the
$\beta$-cycle does not matter, since the $12$ different choices will
lead to $12$ Clifford equivalent SIC vectors.  But given a choice of
$\beta_0$ together with the choice of $\tilde{\alpha}_0$, the starting point of the
shorter $\tilde{\alpha}$--cycle, does matter. Finally, two possible
Zauner unitaries were displayed in equations \eqref{eq:Zauner1} and
\eqref{eq:Zauner2}. They correspond to two ways of ordering the
vectors ${\bf v}_2$ and ${\bf v}_3$ in the Ansatz, and the choice
matters. Hence we have $2\times 4\times 4$ candidate SIC fiducial
vectors to investigate.

The search can be done in two stages. First we test the overlap
$\langle \Psi_0|\openone_4\otimes X^{(p)}|\Psi_0 \rangle$.  This
calculation can be done within the number field $K^{4\pp\jj}$ holding the
vector $|\Psi_0 \rangle$, and it is insensitive to the ordering of the
vectors ${\bf v}_2$ and ${\bf v}_3$. Hence there are only $4\times 4$
candidates to look through, and for $d = 52$ the search can be made in
a fraction of a second. With the choice of $\sigma$ that we made in
Step 4 and the choice of $\beta_0$ that we made in equation
\eqref{eq:betaval}, it turns out that this overlap has absolute value squared 
$1/(d+1)$ if and only if $\theta = 7$ and $\tilde{\alpha}_0$ is the
root that was written down in eq.~\eqref{eq:alfaval}.  (Not by
accident since we adjusted the latter equation after the fact).

To find the correct ordering of ${\bf v}_2$ and ${\bf v}_3$ we can, for instance, test the overlap $\langle \Psi |
Z^{(4)}\otimes Z^{(p)}|\Psi \rangle$. It has the right absolute value only if we place the vectors in the order 
$({\bf v}_1^{\rm T}, {\bf v}_2^{\rm T}, {\bf v}_3^{\rm T}, {\bf v}_4^{\rm T})^{\rm T}$.  However, such a 
calculation requires us to bring in the $p$th roots of unity. A faster way is to test one of the $G(i,k)$ from 
Section \ref{sec:verification}. Such a calculation can be performed in the
field $K^{8\pp\jj\jj^\tau}$, obtained by extending the 
field $K^{4\pp\jj}$ by an eighth root of unity. 

In this case a complete verification that we really do have a SIC vector can be done in a few seconds.  

\section{Dimensions $4$ and $12$}\label{sec:twelve}
For completeness, let us give a fiducial SIC vector for dimension four. 
There the Stark units 
drop out of our Ansatz, so we obtain \cite{monomial}:
\begin{equation}
  |\Psi_0\rangle = N\left( \begin{array}{c} 1 \\ 1 \\ 1 \\ \sqrt{x_0} \end{array} \right) ,
\end{equation}
where the normalising factor $N = N(d)$ was defined in eq.\eqref{eq:x0}.

The other dimension that we have so far avoided is dimension $12 = 4\times 3$. Let us first recall that 
\begin{equation}
  d = n^2 + 3 \qquad \Longrightarrow \qquad
  d = 2^{e_2}\times 3^{e_3}\times p_1^{e_{p_1}}\times \ldots \times p_s^{e_{p_s}},
\end{equation}
where $e_2 \in \{0,2\}$, $e_3 \in \{ 0,1\}$, and all the primes
$p_j\equiv 1 \bmod 3$ \cite{ABGHM}. In this paper we have been
concerned with the special features that arise if there is a factor of
$4$ in the dimension. It turns out that a factor of $3$ also
introduces some special features. For one thing the symmetry of the
ray class SIC is then of type $F_a$ in the terminology of
reference \cite{Scott}. This means that the order $3$ symmetry group acts
trivially in the dimension $3$ factor, so that our Ansatz has to be
slightly modified. To see what the next special feature is we simply
write down a solution for dimension $12$ in almost flat form, using
the enphased monomial representation of the Weyl--Heisenberg group in
the dimension-four factor. The prime $p=3$ does indeed still split over the
quadratic field, so that the ideal $(3)$ splits into prime ideals that
we again denote by $\mathfrak{p}$ and $\mathfrak{p}^\tau$. Using
\begin{equation}
  a = \sqrt{13} \quad\text{and}\quad t_1 = \sqrt{-\frac{a+1}{2}},
\end{equation}
the Stark phase units in the field $K^{\mathfrak{p}\mathfrak{j}} = K(t_1)$ are 
\begin{equation}
  \epsilon_0 = \frac{1}{4}(1-a - 2t_1) \qquad\text{and}\quad \epsilon_1 = \frac{1}{4}(1-a + 2t_1). 
\end{equation}
Then we find the SIC fiducial vector 
\begin{equation} |\Psi_0\rangle = N \left( \begin{array}{c}
    {\bf v}_1 \\ \hline
    {\bf v}_1 \\ \hline
    {\bf v}_1 \\ \hline
    {\bf v}_4 
  \end{array} \right),\qquad\text{where}\quad
    {\bf v}_1 = \left( \begin{array}{c} 1 \\ -\epsilon_0 \\ -\epsilon_1 \end{array} \right)
\quad\text{and}\quad
    {\bf v}_4 = \left( \begin{array}{c} \sqrt{x_0} \\ -(\epsilon_1)^\frac{3}{2} \\ -(\epsilon_0)^\frac{3}{2} \end{array} \right).
\end{equation}
In this case it turns out that the ray class fields
$K^{\mathfrak{p}\mathfrak{j}}$ and $K^{2\mathfrak{p}\mathfrak{j}}$ are
identical as number fields, and the Stark phase units in the former
are in fact square roots of the Stark phase units in the latter. Hence
the form of the vector ${\bf v}_1$ is the expected one. The new
feature, which as far as we know always occurs when the dimension is
divisible by three, is that square roots of odd powers of Stark phase
units (in this case, cubes) also appear in the fiducial vector. As a
matter of fact this phenomenon is not exclusive to dimensions
divisible by three, it occurs also (this time in fifth powers) 
when $d = 82^2+3 = 6727 =  7\times 31^2$ and $\ell = 5$. 
Since this dimension is not of the form $4p$ we do not discuss it further here. 

\section{Overlap phases and determination of signs}\label{sec:overlaps} 
The SIC condition requires the overlaps $\qd1\langle
\Psi_0|D_{i,j}|\Psi_0 \rangle$, where $D_{i,j}$ is any displacement
operator in dimension $d$, to sit on the unit circle in the complex
plane. Hence they are referred to as overlap phases. They are not
obviously algebraic units, but in every case investigated (dimension
$d = 3$ excepted) they are indeed such.  They are important because
any SIC can be constructed from its overlap phases.

It is perhaps worth observing that given our Ansatz, where the complex conjugate of 
a fiducial vector is given by eq. \eqref{UPd}, the SIC overlap phases can be rewritten as a 
bilinear quadratic form 
\begin{equation} \qd1\langle \Psi_0|D_{i,j}|\Psi_0 \rangle = \qd1(\Psi_0^* , D_{i,j}\Psi_0 ) = 
\qd1 (\Psi_0, (U_P^{(4)}\otimes U_P^{(p)})D_{i,j}\Psi_0 ),  \end{equation}
\noindent where the parity operator $U_P^{(4)}\otimes U_P^{(p)}$ is an involution of order 2, and where we 
avoided the use of Dirac's notation since it does not work well for anti-unitary transformations. 

The number fields in which the various overlap phases (or their
squares) sit are listed in Table~\ref{tab:baby4p}. As an illustrative
example of how this comes about, consider the third row. Using the
explicit form \eqref{eq:Z_adapted} of the displacement operator
$D_{0,1}$ together with our Ansatz we find
\begin{equation}
  \langle \Psi_0|D_{0,1}\otimes D_{0,0}|\Psi_0\rangle = N^2\tau_4 
  \bigl(-\langle {\bf v}_1|{\bf v}_3\rangle - i\langle {\bf v}_2|{\bf v}_4\rangle + 
    i\langle {\bf v}_3|{\bf v}_1\rangle  - \langle {\bf v}_4|{\bf v}_2\rangle\bigr),
\end{equation}
where $\tau_4$ is a primitive eighth root of unity. But it is built into our Ansatz that 
\begin{equation}
  \langle {\bf v}_3|{\bf v}_1\rangle = \langle {\bf v}_1|{\bf v}_3\rangle,
  \qquad
  \langle {\bf v}_2|{\bf v}_4\rangle = -  \langle {\bf v}_4|{\bf v}_2\rangle.
\end{equation}
Moreover $\tau_4(1-i) = \sqrt{2}$, so 
\begin{equation}
  \langle \Psi_0|D_{0,1}\otimes D_{0,0}|\Psi_0\rangle = N^2\sqrt{2} 
  \bigl(-  \langle {\bf v}_1|{\bf v}_3\rangle -  \langle {\bf v}_4|{\bf v}_2\rangle\bigr) .
\end{equation}
The fourth root of unity has disappeared. The two terms are each
invariant under the simultaneous permutations $\beta_r \rightarrow
\beta_{r+1}$, $\tilde{\alpha}_r \rightarrow \tilde{\alpha}_{r+1}$,
which means that they sit in the base field except for the presence of
$\sqrt{x_0}$ in ${\bf v}_4$. Hence these overlap phases sit in
$K^{8\jj}$, and their squares in $K^{4\jj}$. The remaining entries in
the third column of Table \ref{tab:baby4p} can be dealt with
similarly.

According to a by-now-standard conjecture, for the cases we consider
the SIC overlap phases are square roots of Stark phase units, but this
time Stark units in the ray class field with modulus $d$ or $2d$ have
to be included.  This expectation was the starting point for Kopp's
construction of SICs in prime dimensions equal to $2$ modulo $3$, for
which the overlap phases form a single Galois orbit \cite{Kopp}. For
the case when $d = 4p$ where $p\equiv 1\bmod 3$ the situation is more
complicated, but for a few examples we have checked that the
connection between overlap phases and Stark phase units still holds,
although it is not just square roots that appear; see Table
\ref{tab:baby4p}, which has been fully verified for some cases only,
owing to the time it takes to compute Stark units in large ray class
fields.

\begin{table}[htb]
\caption{Non-trivial overlap phases for $d = 4p$. In the tensor product the first 
operator is a $d = 4$ displacement operator, the second ditto for dimension $p$. When 
$d = 4$ the operators $D_{0,0}$ and $D_{0,2}$ are diagonal, while $D_{0,1} = X$ is 
`generic'. When $d = p$ the operators $D_{0,1} = X$ and $D_{0,1} = Z$ are special, 
while $D_{1,1}$ is `generic'. The table gives one representative for each Galois orbit, 
up to possible signs. In the fourth column all entries were checked for $d = 28$, $52$, 
$124$. For $d =172$, $292$, $844$ the first five entries were checked.\label{tab:baby4p}}
  \centerline{
    \begin{small}
\renewcommand{\arraystretch}{1.5}
\begin{tabular}
{|c|c|c|c|c||c|c|c|c|c|}\hline
representative & degree & number field & remark \\ 
\hline
&&&\\[-4ex]
\hline
$\qd1\langle \Psi|D_{0,0}\otimes D_{1,0}|\Psi \rangle$ & $(p-1)/3\ell$ & 
$K^{\pp\jj}$ & Stark phase unit \\ 
\hline 
$\qd1\langle \Psi|D_{0,2}\otimes D_{1,0}|\Psi \rangle$ & $(p-1)/\ell$ & 
$K^{2\pp\jj}$ & $-$Stark phase unit\\ 
\hline 
$(d+1)\langle \Psi|D_{0,1}\otimes D_{0,0}|\Psi \rangle^2$ & 2 & 
$K^{4\jj}$ & $(\text{Stark phase unit})^\ell$ \\ 
\hline 
$(d+1)\langle \Psi|D_{0,1}\otimes D_{0,1}|\Psi \rangle^2$ & $2(p-1)/\ell$ & $K^{4\pp^\tau\jj}$ 
& Stark phase unit \\
\hline 
$(d+1)\langle \Psi|D_{0,1}\otimes D_{1,0}|\Psi \rangle^2$ & $2(p-1)/\ell$ & 
$K^{4\pp\jj}$ & Stark phase unit \\
\hline
&&&\\[-4ex]
\hline
$(d+1)\langle \Psi|D_{0,0}\otimes D_{1,1}|\Psi \rangle^2$ & $(p-1)^2/6\ell$ & 
$K^{p\jj}$ & Stark phase unit \\ 
\hline 
$(d+1)\langle \Psi|D_{0,2}\otimes D_{1,1}|\Psi \rangle^2$ & $(p-1)^2/2\ell$ & 
$K^{2p\jj}$ & Stark phase unit \\ 
\hline
$(d+1)\langle \Psi|D_{0,1}\otimes D_{1,1}|\Psi \rangle^2$ & $2(p-1)^2/\ell$ & 
$K^{d\jj}=K^{4p\jj}$ & \text{Stark phase unit} \\
\hline
\end{tabular}
    \end{small}}
\end{table}

We now focus on the first two rows of Table \ref{tab:baby4p}. Recall that the dimension $p$ displacement 
operator $D_{1,0}=X$ is a permutation matrix. We have observed the appealing 
formulas
\begin{alignat}{5}
 &&  \qd1\langle \Psi_0|D_{0,0}\otimes X^{-2j}|\Psi_0\rangle&{} =  x_j^2\label{eq:baby3}\\
\text{and}\quad\nonumber\\
 &&  \qd1\langle \Psi_0|D_{0,2}\otimes X^{-2j}|\Psi_0\rangle &{}= -y_j^2. \label{eq:baby4}
\end{alignat}
Here the numbers $x_i$ and $y_i$ are the components of the vector
$|\Psi_0\rangle$, as introduced in equations \eqref{eq:v0_v1}. The
components are square roots of Stark units in
$K^{\pp\jj}$ and in $K^{2\pp\jj}$
respectively, hence these particular overlap phases are Stark phase
units in themselves. In effect then equations \eqref{eq:baby3} and \eqref{eq:baby4} 
are identities for Stark units that, we conjecture, hold for all $d = 4p$. 
We have made a similar observation when the dimension is a
prime of the form $n^2+3$ \cite{ABGHM}, and indeed in all cases
investigated. As a consequence, the results reported in Table \ref{tab:baby4p} may also be interpreted as 
identities connecting Stark units in different subfields. 

If we assume that equations \eqref{eq:baby3} and \eqref{eq:baby4} hold
we can reduce the search in the final step of our algorithm, because they
will determine the signs that appear there.  Equivalently, they can be
used to distinguish between the factors $\tilde{p}_i(t)$ and
$\tilde{p}_i(-t)$ in eqs.~\eqref{eq:factor_p0} and
\eqref{eq:factor_p1} of the minimal polynomials of the Stark phase
units over the quadratic field $K$.

For both equations \eqref{eq:baby3}
and \eqref{eq:baby4}, we consider the sum from $j=1$ to $j =
p-1$. For the right-hand sides, we get
\begin{alignat}{9}
  &&\sum_{j=1}^{p-1} x_j^2 = 3\ell \sum_{j=0}^{\frac{p-1}{3\ell}-1} \alpha_j^2 =
     \frac{3\ell}{x_0} \sum_{j=0}^{\frac{p-1}{3\ell}-1} \tilde{\alpha}_j^2
     &{}= \frac{3\ell}{x_0}\Tr{K^{\pp\jj}}{\HK}(\tilde{\alpha}_0^2)=\frac{3\ell}{x_0}\tr(\tilde{\alpha}_0^2)\\
  \text{and}\quad&&
  -\sum_{j=1}^{p-1} y_j^2 = -\ell \sum_{j=0}^{\frac{p-1}{\ell}-1} \beta_j^2 &{}=-\ell\Tr{K^{2\pp\jj}}{\HK}(\beta_0^2)=-\ell\tr(\beta_0^2).
\end{alignat}
We are summing all units $\alpha_j^2$ and $\beta_j^2$ which form an
orbit with respect to the cyclic group generated by the automorphism
$\sigma$ which fixes the Hilbert class field $\HK$. Hence, the sum
equals the trace for the extension over the Hilbert class field, which
we abbreviate by $\tr(\_)$.
Those values are the negative of the coefficient of second-highest degree, of
the minimal polynomial of $\tilde{\alpha}_0$ and $\beta_0$ over the Hilbert
class field.  Hence, for class number $\h>1$, we have to compute the
factorisation of the minimal polynomials over the Hilbert class field.

For the left-hand sides, first observe that
\begin{alignat}{5}
  \sum_{j=1}^{p-1} X^{-2j}=J-\openone = \ket{\bf 1}\bra{\bf 1} - \openone,
\end{alignat}
where $J$ is the van der Waerden matrix, a $p\times p$ matrix with all entries being equal to
$1$, and $\ket{\bf 1}$ is the vector with all components equal to $1$. 
Hence, the left-hand sides of eqs.~\eqref{eq:baby3} and
\eqref{eq:baby4} are symmetric functions in
the components of the fiducial vector as well.  More precisely, we get
\begin{alignat}{5}
  \sum_{j=1}^{p-1}\bra{\Psi_0}&D_{0,0}\otimes X^{-2j}\ket{\Psi_0}
 =N^2\bigl(3|\braket{{\bf 1}}{{\bf v}_1}|^2+|\braket{{\bf 1}}{{\bf v}_4}|^2)-1 \\
\text{and}\qquad
  \sum_{j=1}^{p-1}\bra{\Psi_0}&D_{0,2}\otimes X^{-2j}\ket{\Psi_0} =N^2\bigl(|\braket{{\bf 1}}{{\bf v}_1}|^2-|\braket{{\bf 1}}{{\bf v}_4}|^2\bigr)+\frac{1}{\qd1},
\end{alignat}
where $N$ is the normalisation factor defined in eq.~\eqref{eq:x0}.
Additionally, we compute
\begin{alignat}{7}
  \braket{{\bf 1}}{{\bf v}_4}&{}=\sum_{j=0}^{p-1} x_j
   &&{}=\sqrt{x_0}+3\ell \sum_{j=0}^{\frac{p-1}{3\ell}-1}\alpha_j
   &&{}=\sqrt{x_0}+\frac{3\ell}{\sqrt{x_0}} \tr(\tilde{\alpha}_0)\\
\text{and}\qquad
  \braket{{\bf 1}}{{\bf v}_1}&{}=\sum_{j=0}^{p=1} y_j
    &&{}=1+\ell \sum_{j=0}^{\frac{p-1}{\ell}-1}\beta_j
    &&{}=1+\ell \tr(\beta_0),
\end{alignat}
where the trace is relative to the extension of the Hilbert class
field.

Combining all of this, we find that
\begin{alignat}{5}
  N^2\left(3|1+\ell\tr(\beta_0)|^2+\left|\sqrt{x_0}+\frac{3\ell}{\sqrt{x_0}}\tr(\tilde{\alpha}_0)\right|^2\right)-1
  &{}=\frac{3\ell\tr(\tilde{\alpha}_0^2)}{\sqrt{x_0}\qd1}\\
\text{and}\qquad
  N^2\left(|1+\ell\tr(\beta_0)|^2-\left|\sqrt{x_0}+\frac{3\ell}{\sqrt{x_0}}\tr(\tilde{\alpha}_0)\right|^2\right)+\frac{1}{\qd1}
     &{}=-\frac{\ell\tr(\beta_0^2)}{\qd1}.
\end{alignat}
From this, taking complex conjugates using equations \eqref{eq:cc1}--\eqref{eq:cc3}, we derive
\begin{alignat}{5}
  - 4\Bigl( x_0 + 6\ell \tr(\tilde{\alpha}_0) +{} &\frac{9\ell^2}{x_0}\tr(\tilde{\alpha}_0)^2\Bigr) \nonumber \\ 
  &{}=(1+\qd1)\left(3\ell\tr(\alpha_0^2)+3\ell\tr(\beta_0^2)+4\right)+d\label{eq:sign_alpha}
\end{alignat}
and
\begin{equation}
  4( 1+\ell\tr(\beta_0))^2=(1+\qd1)\left(3\ell\tr(\alpha_0^2)-\ell\tr(\beta_0^2)\right)+d.\label{eq:sign_beta} 
\end{equation} 
The right-hand side of these equations can be computed using the 
second-highest degree monomial in the minimal polynomials over the
Hilbert class field of the Stark phase units $\alpha_0^2$ and $\beta_0^2$,
which equals the negative of their trace.  The left-hand sides depend
on the trace of the square roots $\tilde{\alpha}_0$ and $\beta_0$, and
hence they are sensitive to the choice of the signs (respectively, the
choice of the factors) in eqs.~\eqref{eq:factor_p0} and 
\eqref{eq:factor_p1}.  Note that for class number $\h>1$, this yields
partial information about the choice of the element $\alpha_0$ as
well.

Hence, assuming that our observation in eqs.~\eqref{eq:baby3} and
\eqref{eq:baby4} was true in general, we can determine the correct
signs $s_0$ and $s_1$ independently of the choice of the generator
$\theta$ and the element $\alpha_0$.  Moreover, the number of possible
choices for $\alpha_0$ is reduced from $\h(p-1)/3\ell$ to
$(p-1)/3\ell$, \ie, the degree of $K^{\pp\jj}$ over the Hilbert class field.

\section{Conclusions and Outlook} \label{sec:conclusions}
We have collected a substantial amount of evidence for our claim that
we have formulated an algorithm that in principle allows us to convert
Stark units in appropriate ray class fields to SICs, in all dimensions
of the form $d = n^2+3$. Here it is important to guard against various
special features that occur in low (say, double digit) dimensions. In
so far partially unpublished work we have in fact computed fiducial
vectors in dimensions $d = n^2+3$ for all $n \leq 53$, as well as in
some higher dimensional cases (namely dimensions $3028$, $3252$,
$3484$, $3603$, $3724$, $3972$, $4099$, $4492$, $4627$, $5779$,
$6727$, $7399$, $7924$, $12324$, $19324$, $19603$, $39604$, and
$45372$). Hence we feel that we are safe against low dimensional
accidents, and with the various improvements reported here we believe
that we can talk about an ``algorithm'', rather than just a
``recipe'', as we did in \cite{ABGHM}

If $d = n^2+3$ the prime decomposition of $d$ is
\begin{equation}
  d = 4^{e_4}\times 3^{e_3} \times p_1^{r_1} \times \ldots \times
  p_s^{r_s},
\end{equation}
where $e_3, e_4 \in \{ 0, 1\}$, and all the primes $p_j \equiv 1$
modulo $3$. In \cite{ABGHM} we treated the conceptually simplest case
of $d=n^2+3=p$ being a prime. As we have seen a factor of four in the
dimension requires some special measures on the Hilbert space side in
order to `decouple' the fiducial vector from the cyclotomic field.  We
have therefore focused this paper on dimensions of the form $d= n^2+3
= 4p$, where $p$ is an odd prime. Moreover, the $d = 4p$ case has the
advantage that the degrees of the relevant fields are comparatively
smooth, and this facilitates the calculations. The highest dimension
reached is $d = 39604 = 199^2+3$. The degree of the number field
needed to write down all the $d^2$ SIC vectors in this case is
$71\,280\,000$.  We can handle this dimension because the Stark units
from which a suitably chosen fiducial vector is constructed belong to
a number field of a degree of just $1800$ over the rationals.  Along
the way we pointed out several improvements of our algorithm, as
compared to \cite{ABGHM}.

One reason why $d = 39604$ is manageable is that this dimension
corresponds to $\ell = 11$ in the dimension tower given in
eq.~\eqref{eq:tower}, starting at $d = 4$.  It is possible that this
sequence of dimensions contains an infinite sequence of dimensions of
the form $d = 4p$. However, the next candidate occurs for $\ell =
211$, and $d_{211}$ already has $89$ digits. Explicitly constructing a
fiducial vector in such a dimension is out of reach, as the number of
its components exceeds the estimated number of atoms in the known
universe. Hence we have reached a point where a proof that the
algorithm always works is urgently needed.

We believe that the restriction to dimensions $d$ where the quadratic
base field admits a fundamental unit of negative norm, and where SICs
with anti-unitary symmetry appear, namely $d = n^2+3$, may simplify an
existence proof.

\section*{Acknowledgements}
We thank Marcus Appleby for many valuable discussions.  We also thank
the Mathematics Department at Stockholm University as well as the Max
Planck Society via the Max Planck Institute for the Science of Light
in Erlangen for access to their computers.

I.\,B. acknowledges support by the Digital Horizon Europe project
FoQaCiA, Foundations of quantum computational advantage, GA
No. 101070558, funded by the European Union, NSERC (Canada), and UKRI
(UK).  

M.\,G. acknowledges support by the Foundation for Polish Science (IRAP
project, ICTQT, contracts no. 2018/MAB/5 and 2018/MAB/5/AS-1,
co-financed by EU within the Smart Growth Operational Programme).

G.\,M. thanks Myungshik Kim and the QOLS group at Imperial College
London for their generous ongoing hospitality and support.

\appendix

\section*{Appendices}
In the following four appendices we prove some results used or mentioned 
in the main text, as well as some general results which shed light on
the basic number-theoretic context of these constructions.  One of the
enduring mysteries of this subject is the ubiquitous appearance of the
so-called \emph{Zauner symmetry} of order three in every known SIC.
Many of the number-theoretic observations below are concerned with
internal $3$-symmetries; though it remains to be shown if there is a
link.

\section{General results about the towers $d_\ell (D)$}\label{tours}

This first appendix is concerned with the way intrinsic properties of
a quadratic field $\QD$ determine the primes dividing the values of
$d_\ell(D)$, defined in eq.~\eqref{deeell}, occurring in the tower of
dimensions $\{ d_\ell(D) \colon \ell \in \N \}$ that lie above it.

The appearance or otherwise of a particular prime in the tower of
dimensions above $D$---and more particularly the power to which it
first appears if it does---is a deep problem analogous to topics in
classical number theory like Wieferich primes \cite{Wieferich}.  In
particular it has a direct connection to so-called Wall-Sun-Sun primes
and late twentieth-century attempts to prove Fermat's Last
Theorem~\cite{sunsun,grasReg}; see the remark following
Proposition~\ref{ants}.

Unless explicitly stated, there is no stipulation here that dimensions
be of the form $d=n^2+3$, nor $d=4p$.  Note that where it is clear
from the context, we shall simply write $d_\ell$ for $d_\ell(D)$.

\subsection{The dimension $d_\ell(D)$ at position $\ell$ in the tower above $\QD$}\label{subsec:towers}

We set out the key structural results and then prove them below.  Let
$s,t$ be rational integers.  The notation $s \mid t$ will mean that
$t$ is divisible by $s$.  When $s$ is of the form $p^r$ for a prime
$p$ and a positive integer $r$, $p^r \mathrel{\Vert} t$ will mean that
$p^r$ is the highest power of $p$ dividing into $t$.  If $t$ is an
algebraic integer then $s \mid t$ will mean that the ideal $(t)$ is
contained within the ideal $(s)$.

\subsubsection*{The $4p$ case}
The main result of this appendix, from the perspective of applications to the rest of the paper, is the following. 
Note there is no assumption about $n^2+3$. 

\begin{proposition}\label{ellone}
Let $D$ be any square-free positive integer, and let $\ell \in \N$. 
Suppose that the dimension $d_\ell(D)$ is
equal to $4p$ for some odd prime $p$.  Then either $D = 5$
or $\ell = 1$.
\end{proposition}

\subsubsection*{New primes appearing at level $\ell$}

Going up the dimension towers introduces new prime divisors into 
the $d_\ell(D)$ at each successive value of $\ell$, except in some slightly pathological cases which we 
have excluded from the next result.  

\begin{proposition}\label{knewp}
  Fix $D$ square-free as above, $\ell \in \N$, and assume that the
  following holds:
   \begin{equation}
   \text{If $\ell=2$ then $d_1(D)$ is not of the form $2^N+2$ for some $N\geq2$.} \label{star}
   \end{equation}
  Then there exists a rational prime $\p$ which appears for the first time
  at level $\ell$ in the tower above $\QD$.  In other words, $\p\mid
  d_\ell(D)$ but $\p\mathrel{\nmid}d_{\ell'}(D)$ for any $1\leq
  \ell' < \ell$.
\end{proposition}
The technical restriction required on the pair $D$, $\ell$ in the
statement is easily seen to be irrelevant to cases of the
form $d_\ell(D) = n^2+3$, since by Lemma~1 of~\cite{ABGHM} this would
require solving an equation of the form $m^2 = (n+1)(n-3)$ in
integers, which is not possible.
Hence imposing this condition does not affect 
any application to the main part of the paper.

\subsubsection*{The growth of the $d_\ell(D)$ in a $\Z_q$-tower}\label{growth}
Now we get to the heart of the behaviour at a fixed prime $q \geq 5$, 
which we simply describe for reasons of brevity. 
Let $\ell_0$ denote the first value of $k$ for which $q \mid d_k(D)$. 
The ray class fields of $K$ with 
conductors ${{d}}_{q^r\ell_0}$, for $0 \leq r \leq \infty$, which 
are those naturally attached to the filament of the dimension 
towers with strictly increasing powers of $q$ in the conductor, 
contain the (totally real) cyclotomic $q^n$-power extensions of $K$, and
indeed the compositum of all of the fields $K^{( {d}_{q^r\ell_0} )}$ contains 
the cyclotomic $\Z_q$-extension of $K$.

The following result is true for all primes, but in a 
`shifted' form for $q=3$ which is not relevant
for this paper, so we do not state it. 

\begin{proposition}\label{ants}
Let $D$ be as above and let $q\neq3$ be a prime such that for some integers $\ell,r\geq1$, 
$q^{r} \mathrel{\Vert} d_{\ell}(D)$. 
Suppose further that $\ell$ is minimal for this property, and assume~\eqref{star}. 
Then:
\begin{center}
$q^{r+1} \mathrel{\Vert} d_{q\ell}(D)$ and $q\ell$ is the minimal such index.
\end{center}
\end{proposition}

We also mention that it is shown in~\cite[\S4]{Kopp} 
that for a fixed $D$, the asymptotic 
probability of a randomly-chosen prime $q\geq5$ 
dividing into $d_\ell(D)$ for some $\ell\geq5$ 
is $\frac{3}{8}$. 
The basic ingredient is the result that a prime $p \neq 3$ 
divides into some $d_\ell(D)$ if and only if the order 
of $u_K$ modulo $p\zk$ is divisible by $3$. 

\begin{remark} 
Proposition~\ref{ants} is a variant of a standard result in the theory of
$p$-adic regulators, see for example \cite[\S5.5]{wash}. 
As such, it is also the starting point for the connection with Wall-Sun-Sun primes 
and Fermat's Last Theorem. 
Define $\ell_0 = \ell_0(q,D)$ to be the minimum $\ell$ for which $q$ divides into $d_\ell(D)$. 
Essentially, the literature is replete with conjectures---usually in very different language---which assert 
both the truth, \eg, \cite{grasReg}, and the falsity, \eg, \cite{sunsun}, of the following statement:
\begin{center}
  Fix $D$. There exist only finitely many primes $p$ for which $v_p(d_{\ell_0(p,D)})>1$. 
\end{center}
\end{remark}

\subsubsection*{The growth of the $d_\ell(D)$ in $\R$}

The next statement is a corollary of a classical result
(Lemma~\ref{huaest}) and it is in turn needed for the proof of
Proposition~\ref{knewp}.  There is an easy minor modification for
$D=5$, which is only needed for the first few dimensions of that tower
anyway; but we do not go into this here.

\begin{proposition}\label{dmn}
Let $D\ne5$ be any square-free positive integer.  
For all $\ell\geq1$ and $n\geq2$:
\[
d_\ell^n > d_{\ell n} > \left(\frac{2}{3}\right)^n d_\ell^n .
\] 
\end{proposition}

\noindent
Finally, we mention another interesting structural result on the 
dimension towers, which 
can be proven by modifying results of~\cite{tbang} on 
congruences among coefficients of the Chebyshev 
polynomials of the first kind $\cheb{n}{x}$, to apply 
them to the shifted polynomials $\chsh{n}{x}$ 
defined in eq.~\eqref{shifty}. 
\begin{proposition}
Let $k,l \in \N$ satisfy $\gcd(k/l,3) = 1$. 
Then $\gcd(d_k,d_l) = d_{\gcd(k,l)}$.  \qed
\end{proposition}
\noindent
The hypothesis says that the power of $3$ dividing into $k$ is the
same as that dividing into $l$.  Whenever the $3$-adic valuation of
the rational number $k/l$ is different from zero, the corresponding
gcd of the dimensions is just $1$ or $3$, as may be seen from the
assertions \eqref{zeroed} and \eqref{onetwo} below.

\subsubsection*{Example: $D=5$}

When $D=5$, the asssertions \eqref{zeroed} and \eqref{onetwo} show
that $4$ must divide $d_\ell(5)$ for every odd $\ell$ coprime to $3$.
So all things being equal, one would expect an infinite sub-sequence
of the dimensions above $\Q(\sqrt{5})$ to be of the form $4p$ for $p$
prime.  We know that prime dimensions (when $D=5$) can occur only for
$\ell$ a power of $3$, and the entries are known to be primes $d_\ell
= p$ for $\ell = 3,9$.  We have used probabilistic primality tests to
look for more examples, and found none up to $\ell = 3^{12}$.  From
eq.~\eqref{onetwo} and Proposition~\ref{knewp} below one would expect,
given that $d_1(D)=4$, that dimensions of the form $d_\ell=4p$ would
be easier to find.  This is indeed the case: they occur for $\ell =
5$, $7$, $11$, $211$, $419$, $557$, $769$, $991$, $1259$, $1669$,
$2927$, $3607$, $4391$, $5857$, $7727$, $14591$, $16127$, $22453$,
with five more cases ($\ell=27827$, $51427$, $60103$, $61657$,
$93251$) for $\ell \le 94000$ being probably of this form.

We mention the first few of these: $d_1(5) = 4$, $d_5(5) = 4\times
31$, $d_7(5) = 4\times211$, $d_{11}(5)= 4\times
9901$,\\
{\footnotesize{$d_{211}(5) {=}
    4\times3896944262364468984431989653605889395802956455451592871589240445325974800179030855254151$}};
higher levels in the tower are too large to print here.  For the
growth rates, see Proposition~\ref{dmn}.

Using Pari/GP, it took about four weeks of CPU time to verify that
$d_\ell/4$ is prime for $d_{22453}$ with $9385$ digits.  The next
candidate $d_{27827}$ has $11631$ digits.  As a new result in this
paper, we have shown SIC existence for $\ell = 11$ in this series,
with $5$ digits.

\subsection{Proofs of the propositions}

We first set out the notation and some general straightforward 
deductions from the definitions, then we go on to 
prove each proposition in turn. 

\subsubsection*{Modified Chebyshev polynomials to navigate the tower}\label{modcheb}
Although they are clearly monotonically increasing with~$\ell$, the
dimensions in the same `tower complex' above a fixed~$\QD$
nevertheless only form a partial order in terms of divisibility, not a
total order.  Hence in order to navigate the towers arithmetically, we
introduced in~\cite{AFMY} a family of functions which enable us to
calculate any dimension $d_{n\ell}(D)$ given the value of $d_\ell(D)$,
akin to the role played by the usual Chebyshev polynomials $T_n(x)$,
$U_m(x)$ for the hyperbolic trigonometric functions.  These modified
Chebyshev polynomials of the first kind are defined by:
\begin{equation}\label{shifty}
\chsh{n}{x} = 1+2\cheb{n}{\tfrac{x-1}{2}},
\end{equation}
where $T_n(x)$ is the usual Chebyshev 
polynomial of the first kind of degree $n$. 
The $\chsh{n}{x}$ also satisfy the fundamental relation
\begin{equation}\label{nest}
  \text{$\chsh{m}{\chsh{n}{x}} = \chsh{mn}{x}$ for every $m,n\geq 0$.}
\end{equation}
The $\chshh{n}$ are independent of $D$. So once we 
know $d_0(D) = 3$ (which is true trivially for all $D$: see eq. \eqref{deezero}) 
and $d_1(D)$ (which requires that we know the 
fundamental unit of $\QD$), we know 
all $d_\ell(D)=\chsh{\ell}{d_1}$. 
The first few polynomials are $\chsh{0}{x}=3$, $\chsh{1}{x}=x$, 
$\chsh{2}{x}=x^2-2x$,  $\chsh{3}{x}=x^3-3x^2+3$.  
The defining recursion $\cheb{n}{x} = 2x\cheb{n-1}{x} - \cheb{n-2}{x}$ 
for the $\chebh{n}$ yields for the $\chshh{n}$:
\begin{equation}\label{shiftycur}
\chsh{n}{x} = x\chsh{n-1}{x} - x\chsh{n-2}{x} + \chsh{n-3}{x}.
\end{equation}

We mention some straightforward general consequences of these
definitions for later use; see also Proposition~6 of \cite{AFMY}.  Fix
$D$ and $d_\ell=d_\ell(D)$ for some $\ell\geq1$ and let $\lambda\geq1$
be an integer.  These next facts, immediate
from eqs.~\eqref{nest}, \eqref{shiftycur}, apply to all $D$ and all $\ell$.
Firstly, we may re-write eq.~\eqref{shiftycur} setting the variable $x$ to
be some given $d_\ell = d_\ell(D)$, viz.
\begin{equation}\label{shiftydee}
d_{n\ell} = d_\ell d_{(n-1)\ell}- d_\ell d_{(n-2)\ell} + d_{(n-3)\ell}.
\end{equation}
The following facts are immediate: 
\begin{alignat}{5}
\text{If $\lambda\equiv0 \bmod 3$:}&\quad\text{$d_{\lambda \ell}-3$ is a multiple of $d_\ell$.}\label{zeroed}\\
\text{If $\lambda\equiv1,2 \bmod 3$:}&\quad\text{$d_{\lambda \ell}$ is a multiple of $d_\ell$.}\label{onetwo}
\end{alignat}
  So in particular, 
      \begin{equation}\label{ddivd}
       \text{If $\gcd(\ell,3)=1$, this implies that $d_\ell$ is divisible by $d_1$.}
    \end{equation}
    Moreover, 
    \begin{equation}\label{twoodd}
    \text{$d_1(D)$ is odd if and only if $d_\ell(D)$ is odd for every $\ell \in \N$.}
    \end{equation}
   In other words, if $2$ divides into $d_\ell(D)$ for some $\ell$ then it already divides into $d_1(D)$. 
    \begin{equation}\label{threed}
    	\text{$d_{3\ell}(D)$ is odd for every $\ell \in \N$.}
    \end{equation}
    \begin{equation}\label{always3}
	\text{The prime $3$ always divides into one of $d_1(D)$, $d_2(D)$ or $d_4(D)$.}
    \end{equation}
This last assertion follows because the map $d_j \mapsto d_{2j} = d_j(d_j-2)$ considered
in a general variable $t$ modulo $3$ goes 
$1\mapsto2\mapsto0\mapsto0\mapsto\ldots$, and so $3$ 
must divide into at least one of $d_1$, $d_2$ or $d_4$ 
no matter what $D$ is.

Recall from Section~\ref{growth} that if a prime~$q\geq5$ divides into
some $d_\ell(D)$, we write $\ell_0=\ell_0(q,D)$ for the minimum such $\ell$. 
Further, let us write $\Omega=\Omega(q,D)$ for the order of the image of $u_D$ in 
  $\mg{q}$, the multiplicative group of the quotient ring $\ag{q}$. 
Finally, $\leg{q}{p} \in \{\pm1\}$ denotes the Legendre symbol. 
As alluded to in that section above as well, 
the next result may be deduced from the first section of the proof of
Proposition~4.2 in \cite{Kopp}; or alternatively by using the 
techniques in Lemma~12 of~\cite{AFMY}. 
\begin{equation}\label{knought}
 \text{For a given $D$ and $q$ as above: $\ell_0 = \Omega/3$.} 
\end{equation}
In particular, therefore, $1\leq \ell_0 \leq \frac{q-\leg{q}{3}}{3}$. 
Indeed it is further implicitly proven there that given some
fixed $D$, then:
\begin{equation}\label{essdee}
 \text{A prime $p\neq3$ occurs in $d_\ell(D)$ for some $\ell\in\N$ if and only if $3 \mid \Omega(p,D)$.} 
\end{equation}
Finally, let $\ell' \in \N$. 
Then 
\begin{equation}\label{kzero}
  \text{$q \mid  d_{\ell'}$ if and only if $\ell' = \lambda \ell_0$ for some $\lambda\in\N \setminus 3\N$.}
\end{equation}
The sufficiency
is just \eqref{onetwo}, again since $q\neq3$; the necessity follows by deriving a
contradiction to the minimality of $\ell_0$ using the same techniques as
in the previous paragraph, together with the observation that since $\lambda$ must
be co-prime to $3$, raising to the power $\lambda$ is just an
automorphism on the cube roots of unity.

The final simple observation we make---which again is one of the natural
symmetries around $3$ which occur throughout the SIC problem---is that
by \emph{defining} the $\chsh{\ell}{x}$ for negative integers $\ell$
via the relation~\eqref{shiftycur} we recover the obvious symmetry in
eq.~\eqref{deeell} itself whereby we could easily define $d_{-\ell} =
d_\ell$ in the same formula for all $\ell \in \Z$.  This gives as a
by-product, as observed before, that for \emph{every} $D$,
\begin{equation}\label{deezero}
d_0(D) = 3. 
\end{equation}
This is one way of viewing the ineluctable $3$-symmetry at the heart of every SIC.

\subsubsection*{Proof of Proposition~\ref{ellone} in the $4p$ case: $D=5$ or $\ell=1$}
\begin{proof}
  First of all, we may discard any values of $\ell$ which are
  divisible by $3$ thanks to \eqref{threed}.  Next, by~\eqref{ddivd},
  for $d_\ell$ to have the form $d_\ell=4p$ for some odd prime $p$
  (recalling that all dimensions arising via traces as in
  eq.~\eqref{deeell} must be greater than $3$) if $\ell > 1$ then
  $d_1$ must either be $4$ or $p$, forcing respectively either $D=5$,
  or else (by the lower bound in Proposition~\ref{dmn} which we shall
  prove independently below) $d_1$ is a prime less than $9$,
  meaning $d_1 = 2$, $3$, $5$ or $7$.  We exclude $d_1<4$ by construction
  (see eq.~\eqref{deeell}); $d_1 = 5$ or $7$ precludes
  any $d_\ell(D)$ being even by \eqref{twoodd}.  This completes the
  proof of the proposition.
\end{proof}

\subsubsection*{Proof of Proposition~\ref{knewp}: new primes appear in the tower at every level}

\begin{proof} 
Let $\ell\geq2$ be the minimal index for which no new prime appears in
the tower $d_\ell(D)$. 
Write $\rho\geq2$ for the largest rational 
prime divisor of the index $\ell$ itself.  
Suppose first that $\rho\neq3$, so that we fall within the ambit of
Proposition~\ref{ants}.  Now $d_j$ is always a 
strictly increasing sequence for $j\geq1$:
since no new primes appear at level $\ell$, 
Proposition \ref{ants} says that in fact
$d_\ell = \rho d_{{\ell/\rho}}$.  Proposition~\ref{dmn}, 
with $m=\ell/\rho$ and $n=\rho$, then implies that $d_\ell = \rho
d_{{\ell/\rho}}  > \left(\frac{2}{3}\right)^\rho{d_{{\ell/\rho}}}^\rho$, that is:
$d_{{\ell/\rho}}^{\rho-1} < \left(\frac{3}{2}\right)^\rho \rho $.  But $d_{{\ell/\rho}}
\geq d_1 \geq 2^2$ so $2^{2\rho-2} < \left(\frac{3}{2}\right)^\rho
\rho$ or in other words $\left(\frac{8}{3}\right)^\rho < 4\rho$, 
forcing $\rho<3$, \ie, $\rho=2$. 

So this only leaves the cases $\rho=2$ and $\rho=3$. 
If $\rho=3$ then the prime $3$ itself cannot be `new' by~\eqref{always3}. 
Moreover $2$ cannot be new either 
as $d_\ell = d_{3\frac{\ell}{3}}$ is odd by~\eqref{threed}. 
Recall $d_\ell = \chsh{3}{d_{ \ell  / 3 }} = {d_{ \ell  / 3 }}^3 - 3{d_{ \ell  / 3 }}^2 + 3$.  
We have the following two cases.  
If $3\mathrel{\nmid} d_{ \ell  / 3 }$ then $d_\ell$ is odd and coprime to $3d_{ \ell  / 3 }$, hence it contains one or more `new' odd primes different from $3$, contradicting our assumption.  
On the other hand if $3\mid d_{ \ell  / 3 }$ then $\gcd( d_\ell , d_{\ell /3} ) = 3$, but by Proposition~\ref{dmn}, $d_{\ell } > \frac{8}{27}d_{ \ell / 3 }^3 \geq 18$, as $d_{\ell /3} \geq 4$.   
Again since $d_{\ell }$ is odd, there must be one or more `new' odd primes different from $3$ which cannot divide into $d_{ \ell  / 3 }$.  

So $\rho=2$: that is to say, $\ell=2^s$ for some $s\geq1$.  The
explicit expression for $\chsh{2}{X}$ yields $d_{2^s} / d_{2^{s-1}} =
d_{2^{s-1}}-2$.  In particular therefore $\gcd(d_{2^{s-1}} , d_{2^s} /
d_{2^{s-1}}) = 1$ or $2$.  Suppose for a moment that $s\geq2$.  By
Proposition~\ref{ants} for the case $p=2$, recalling that we are under
assumption~\eqref{star}, $v_2({ {d_{2^s}} / {d_{2^{s-1}}} }) \leq 1$;
whereas direct calculations based upon relation \eqref{accbound} below
show that ${ {d_{2^s}} / {d_{2^{s-1}}} } > 2$.  So ${ {d_{2^s}} /
  {d_{2^{s-1}}} }$ must contain a divisor which is coprime to $2
d_{2^{s-1}}$, a contradiction.  Hence $s=1$, \ie, $\ell=2$, and we now
seek examples where any prime divisor $q$ of ${d_2}/{d_1}$ is already
a divisor of $d_1$.  Applying Proposition~\ref{ants} once more then
shows that such a situation can only occur for the prime $=2$, since
$\ell_0(2,D)=1$.  So $d_1 - 2 = d_2 / d_1 = 2^t$ for some $t\geq1$; or
in other words $d_1 = 2^t + 2$.

It only remains to show that $l=1$ whenever some $d_l$ has the form $2^N + 2$. 
But we have just shown that $\ell=2$ is a necessary condition for no new primes to appear in $d_\ell$, and from the form of $\chshh{2}$ we know that $d_{2l}=2^N d_l$ with $d_l$ even; so $l$ must be $1$, so we are in the situation excluded by~\eqref{star} and hence the proposition is proved.   
\end{proof}

\subsubsection*{Proof of Proposition~\ref{ants}: $q$-primary growth in the dimension towers}

We use the standard notation $\fl{t}$ for the greatest integer $\leq t$.  

\begin{proof}
The statement amounts to saying that 
$v_q(\frac{d_{q\ell}}{d_\ell}) = 1$, 
and that for any $\ell'$, $\ell<\ell'<q\ell$, 
$v_q(d_{\ell'}) \leq r$.
Notice that by induction the latter would then follow for all $1\leq \ell' \leq \ell$ as well. 
While this is true for the prime $p=2$ under the hypothesis~\eqref{star}, 
it requires a tedious modification of the argument and so for this proof we assume $p \ge 5$. 

The proof below is purely algebraic in nature and does not require any
arithmetic specifically related to $D$; hence for ease of exposition
we make the following notational simplification.  Since all of these
statements are relative to a `base dimension' $d_\ell$, without loss
of generality we may, thanks to relation \eqref{nest}, re-base
everything to $d_1$ and so $d_n$ will be understood as what was
denoted by $d_{n\ell}$ prior to the re-basing, et cetera.  Another way
to think of this is to `re-base' the definition of $u_D$ to
$u_D^\ell$.

Our starting point is to continue the decomposition 
of $d_q$ given in eq.~\eqref{shiftydee} 
into terms of lower degree by iteratively substituting for each 
term of the form $d_{q-3\lambda}$ the corresponding 
recursion relation from eq.~\eqref{shiftydee}, for $0 \leq \lambda \leq \flfrac{q}{3}-1$. 
The result is a sum of the form 
\begin{multline}
d_{q} = d_1 d_{q-1}  -  d_1 d_{q-2}  +  d_1 d_{q-4}  -  d_1 d_{q-5}  + \ldots \\
+  d_1 d_{q-1-3\lambda}  -  d_1 d_{q-2-3\lambda}  +  \ldots  
+  \begin{cases}
	d_1, & \text{if $q\equiv 1 \bmod 3$;} \\
	d_{2} = d_1(d_1-2), & \text{if $q\equiv 2 \bmod 3$.} 
\end{cases}
\end{multline}
This we in turn split into two sums:
\begin{equation}
d_{q}      =     d_1 \sum_{\lambda = 0}^{ \flfrac{q}{3}-1 } d_{q-1-3\lambda} 
			- d_1 \sum_{\lambda = 0}^{ \flfrac{q}{3}-2 } d_{q-2-3\lambda} 
+  \begin{cases}
	d_1, & \text{if $q\equiv 1 \bmod 3$;} \\
	d_1(d_1-2), & \text{if $q\equiv 2 \bmod 3$.} 
\end{cases}
\end{equation}
Now, when $q\equiv 1 \bmod 3$ we know by \eqref{zeroed} 
that the sequence of terms $d_{q-1}$, $d_{q-4}$, $\ldots$ , $d_{q-1-3\lambda}$, $\ldots$ , $d_{3}$ all lie 
in the coset $3+d_1\Z$; similarly when $q\equiv 2 \bmod 3$ the same
is true of the terms in the sequence $d_{q-2}$, $d_{q-5}$, $\ldots$ , $d_{q-2-3\lambda}$, $\ldots$ , $d_3$:  both sequences taken as $\lambda$ runs from $0$ to $\flfrac{q}{3}-1$. 
In either case the remaining half of the terms just lie in $d_1\Z$. 
In other words, we get something divisible by $d_1^2$ plus a term which is $d_1$ times the sum of $3\leg{q}{3}$, counted $\flfrac{q}{3}$ times. 
Dividing these expressions by $d_1$ on both sides, we are left modulo $d_1$ with 
$$
\frac{d_q}{d_1}  \equiv 
3 \leg{q}{3} \flfrac{q}{3}  +  
\left.
\begin{cases}
	1, & \text{if $q\equiv 1 \bmod 3$} \\
	-2, & \text{if $q\equiv 2 \bmod 3$} 
\end{cases} 
\right\}
\bmod d_1
$$
which in each case gives us $\frac{d_q}{d_1}  \equiv \leg{q}{3}q \bmod d_1\Z$. 
So indeed $q \mid \frac{d_q}{ d_1 }$. 
If $r \geq 2$, that is if $q^2 \mid d_1$, then this proves 
what we need, since then $q \mathrel{\Vert} \frac{d_q}{ d_1 }$. 

It remains to show this holds for the 
case $r=1$. 
That is, when $q \mathrel{\Vert} d_\ell$ we must show 
that $q^2 \mathrel{\nmid} \frac{d_{q\ell} }{ d_\ell }$. 
By \eqref{kzero} and the minimality of $\ell_0$ it 
suffices to show this for $\ell=\ell_0$, which we now re-base as 
before to $\ell_0=1$. 
So it suffices now to show that $q^2 \mathrel{\Vert} d_{q} $.

Let $\q$ be a prime of $K$ above $q$ and $\Z_{\q}$ the ring of integers of the
completion $K_\q$ of $K$ at the place $\q$. 
The discriminant of the extension $K/\Q$ is either $D$ or $4D$, and 
$\gcd(d_k(D),D) \in \{ 1 , 3 \}$ for every $k,D$ by the definition of $d_k(D)$: 
hence as a divisor of $d_k(D)$, $q$ 
is not ramified in $K/\Q$, by the assumption that $q\neq2,3$. 
Indeed $q$ may be split or inert in $K/\Q$; if
it is split then the following argument can be made at either of the
primes $\q$ or $\ol\q$ above $q$. 
In particular $v_\q(q) = v_q(q) = 1$ and so we 
may write $q = \pi\epsilon$ for some $\epsilon\in\Z_\q^\times$ 
and some choice $\pi$ of uniformiser for $\q$. 

By the Chinese remainder theorem and~\eqref{knought}, the images under
reduction modulo $d_1\zk$ of $u_D$ and $u_D^{-1}$ have to be primitive
cube roots of unity in every component of $\ag{d_1}$, and, in
particular, in $\zk/(q)$; otherwise the image of $u_D^2+u_D+1$ cannot
be zero, as $q\neq3$.  Moreover by our assumption that $q^2
\mathrel{\nmid} d_1$, we know that $u_D^{3} \not\equiv 1 \bmod
q^2\zk$: for assume the contrary.  Then $(u_D-1)u_Dd_1 =
(u_D-1)(u_D^2+u_D+1) \equiv 0 \bmod q^2$, meaning that $q \mid
(u_D-1)$, a contradiction to the minimality of $\ell_0$ (which we
recall was rebased to $1$, meaning that the lowest power of $u_D$
congruent to $1$ modulo $d_1$ is $u_D^3$).

Write $\omega$ for the Teichm\"uller representative of $u_D$, which is 
one of the two primitive cube roots of unity in $\Z_{\q}^\times$. 
From the foregoing, then, we know that the image of $u_D$ inside
$\Z_{\q}$ is of the form $\omega+\pi{\nu}$ for some ${\nu} \in
\Z_\q^\times$.  We develop the first few terms of the power series in
$\pi$ for $u_D^{-q}$ by inverting that for $u_D^q$ and obtain,
substituting $\pi\epsilon$ for $q$ where appropriate:
\begin{align*}
d_q &  = 1 + u_D^q + u_D^{-q} \\
& = 1 + (\omega+\pi{\nu})^q + (\omega+\pi{\nu})^{-q}  \\
\begin{split}
&  \equiv 1+\omega^q+\omega^{-q} + \\
& \qquad \frac{{\nu}}{\omega}(\omega^{q} - \omega^{-q}) \epsilon \pi^2  \ +  \\
&\qquad \qquad \frac{{\nu}^2}{2\omega^2}  ( (q-1)\omega^q + (q+1)\omega^{-q} ) \epsilon \pi^3 \ + \\
&\qquad\qquad \qquad \frac{{\nu}^3}{6\omega^3}  ( (q-1)(q-2)\omega^q - (q+2)(q+1)\omega^{-q} ) \epsilon \pi^4
+ O(\pi^4). 
\end{split}
\end{align*}
So we are reduced to showing that 
\[
1+\omega^q+\omega^{-q} + \frac{{\nu}\epsilon}{\omega}(\omega^{q} - \omega^{-q}) \pi^2 \not\equiv 0 \bmod \q^3. 
\]
But since $\gcd(q,3)=1$, $\omega^{q}$ and $\omega^{-q}$ are distinct
primitive cube roots of unity of order co-prime to $q$,
hence $1+\omega^q+\omega^{-q} = 0$.  Moreover by Theorem~4.3.2
of~\cite{cohen} their difference is a $\q$-adic unit, completing the
proof, since ${\nu}$, $\omega$, $\epsilon$ are also.
\end{proof}

\subsubsection*{Proof of Proposition \ref{dmn}: growth band for the dimensions}

\paragraph*{General bounds on the size of the fundamental unit}
The proof of Proposition~\ref{knewp} relies on Proposition~\ref{dmn}, which in its turn 
relies upon the following
`classical' result, for which a convenient reference is Theorem~13.4
on page~329 of the English edition of~\cite{huabook}. 
Also see~\cite{jake} for some more recent work.  This still essentially
sums up the current state of knowledge of the size of $u_K$ for
general $D$, even though Hua's original result for the upper
bound \cite{hua} was published in 1942.  We leave the case $D=5$ out
of this general limit since it makes things less sharp. 
Moreover we
have approximated the lower bound in order to simplify the expression;
but there is a slightly better, sharp lower bound which is realised
infinitely often and which in the notation of the statement of the lemma would
be $\frac{\D+\sqrt{\D^2-4}}{2}$. 

\begin{lemma}\label{huaest}
Let $D>1, D\neq5$ be a square-free integer and let $\Delta_D := t^2
D$, for $t \in \{1,2\}$, be the (`fundamental') discriminant of the
real quadratic field $\QD$.  That is, $t=1$ if $D\equiv1\bmod4$
and $t=2$ if $D\equiv2$ or $3\bmod4$.  Write $\D$ for the positive
real number $\sqrt{\Delta_D} = t\qD$.  Let $e$ be the base of the
natural logarithm.  Then with the usual notation for the integer part
of a real number,
\vskip\abovedisplayskip
\ \hfill$\displaystyle\lfloor \D \rfloor < u_K < ( e \D )^\D$.\hfill\qed
\vskip\belowdisplayskip
\end{lemma}

\paragraph*{Proof of proposition \ref{dmn}}
Recall from~\cite{AFMY} the definition of $u_D$ as the first totally
positive power of the fundamental unit $u_K$ (so $u_D$ is
either $u_K^2$ or $u_K$ according to whether $\zk$ respectively has a
unit of norm $-1$ or does not).  Equation (11) of \cite{AFMY} then
gives the definition of the $\ell$th dimension in the tower
above $\QD$ as $d_\ell(D) = 1 + u_D^\ell + u_D^{-\ell}$.

Now fix $\ell$: then the upper bound is immediate from
the definition (for $n\geq2$).  For the lower bound, let $\ve_\ell = 1 +
u_D^{-\ell} + u_D^{-2\ell}$, so that $ u_D^\ell = d_\ell / \ve_\ell$ and so
\begin{equation}\label{accbound}
  d_{\ell n} = u_D^{\ell n} + 1 + u_D^{-\ell n}\ >\ u_D^{\ell n} = d_\ell^n / \ve_\ell^n.
\end{equation}
By direct calculation using Lemma~\ref{huaest}, the maximal value of $\ve_\ell$ over all $D$ and
all $\ell$ occurs when $D=5$ and $\ell=1$, namely $6-2\sqrt{5} \approx
1.527864$. Everything else lies in the real interval $(1,\frac{3}{2})$. 
In particular therefore since we have omitted $D=5$, we recover the
stated lower bound. 
\qed

\section{Behaviour of the prime above $2$ in $L/\Q$}\label{sec:AppendixB}
For this appendix, unless stated otherwise, we relax the assumption
that $d=n^2+3$ be of the form $4p$ and merely require that $n$ be odd. 
Where we refer to $K=\QD$, $u_K$, $D$, $L = K(\qu)$ et cetera, we are
referring to some fixed value of $D$ under our usual hypotheses. 
Once more, we know that $D\equiv5\bmod8$ and so the prime $2$ is inert
in the extension $K/\Q$. 
However for convenience we do always assume that $\ell=1$, so that in
particular $\xi$, defined in eq.~\eqref{norxi}, means $\xi_1$ in the notation of \cite{ABGHM}. 
This assumption in no way restricts the applicability of the results stated. 

\subsubsection*{$\xi-1$ is a uniformiser for the quadratic ramified local field extension $L_\j2$ of $K_{(2)}$}
We first show directly that $2\zk$ is ramified in $\zl$ (hence the
discriminant is exactly $4\zk$, as per Remark 9 (i) of \cite{ABGHM}).
Let $\j2 = (2,\xi-1)$, where as usual we write $(a,b,c,\ldots)$ for
the ideal of $\zl$ generated by $a,b,c,\ldots$
\begin{proposition}\label{ram2}
With notation as above, $\j2^2 = 2\zl$. 
\end{proposition}

\noindent
In the event that $\j2$ is principal---including the class number one case---we just have $\j2 = (\xi-1)$. 

\begin{proof}
Let $\chi$ be the primitive Dirichlet character on $\Z$ of conductor $4$. 
That is, $\chi(n) = \pm1$ according as $n \equiv \pm1 \bmod 4$.  
Now, as ideals of $\zl$: 
\begin{equation}\label{twose}
\j2^2    =    ( 4 , 2\xi - 2 , \xi^2 - 2\xi + 1 )  =  (4 , 2\xi - 2 , \xi^2 - 1 )  =  ( 4 ,  2\xi - 2, 3+\qd1 ) . 
\end{equation}
But we know \cite[Corollary 2]{ABGHM} that $\pm2u_K^{\pm1} = n\pm\qd1$
and so since $n$ is odd, we may write $n=4m+\chi(n)$ which defines $m=\frac{n-\chi(n)}{4} \in \Z$. 
Hence since $u_K$ and its conjugate are \emph{a fortiori} units in $\zl$, we see that
\begin{equation}\notag
\begin{split}
2 \times \text{(a unit in $\zl$)} &= n \pm \qd1 \\
& =  4m + \chi(n) \pm \qd1 \\
 & = 4(m + \chi(n)) - 3\chi(n) \pm \qd1 \\
& =  \begin{cases}
4(m + \chi(n)) + 3 \pm \qd1,  & \text{if $n\equiv3\bmod4$}; \\
4(m + \chi(n)) - (3 \mp \qd1),   &  \text{if $n\equiv1\bmod4$,} 
  \end{cases}
 \end{split}
\end{equation}
showing that in either case we may express $2$ as a unit times a
$\zk$-linear combination of the generators in eq.~\eqref{twose}. 
So $2 \in \j2^2$, which is to say $\j2^2 \supseteq 2\zl$. 

On the other hand by our stipulation that $d$ in particular be even, it follows from
the fact that $\left( \zk/(2) \right)^\times \cong C_3$ that $\qd1
\equiv 1 \bmod 2\zk$, since that is true of its square, and squaring
is an automorphism of this group.  
Hence, finally, $3+\qd1 \in 2\zk \subseteq 2\zl$.  
Since the other two generators in eq.~\eqref{twose} are also in $2\zl$, we
see that $\j2^2 \subseteq 2\zl$ and the proposition is proven.  
\end{proof}

\begin{corollary}
$\xi-1$ is a uniformiser for the $\j2$-adic completion $L_{\j2}$ of $L$. \qed
\end{corollary}

Let us write $\K2$ for the unique quadratic unramified extension
of $\Q_2$, which---recalling once again that $D\equiv5\bmod8$ and so $2$
is inert in $K/\Q$---is isomorphic in all of our $n^2+3$, $n$ odd,
cases to the completion $K_{(2)}$ of the quadratic base
field $K = \QD$ at the unique prime above $2$.  Let $\Z_{\K2}$
denote its ring of integers.  We may regard $\K2$ as being generated
over $\Q_2$ by adjoining a primitive cube root $\omega$ of unity
with minimal polynomial $X^2+X+1$.  For convenience we may think
of $\K2$ as $\Q_2(\sqrt{5})$ (or indeed $\Q_2(\sqrt{-3})$).

We know already that, considering $\K2$ now as an abstract local field, 
$\K2 \subseteq L_{\j2}$ in every case, so it is
easiest to proceed by treating $L_{\j2}$ as a quadratic extension 
of $\K2$.  
Let $\LL$ be the collection of totally ramified extension 
fields of $\K2$ generated by the
elements $\qu$ as $K$ runs over 
the set of all square-free $D$ such that $d_1(D)$ is 
of the form $n^2+3$ for some odd $n$. 
$\LL$ is necessarily a finite
set by general results of Krasner \emph{et al}: 
see for example~\cite{serreLF},~\cite{jonrob} or~\cite{koblitz}.

Interestingly, as we now prove, the isomorphism class of the 2-adic
extension field $L_{\j2}$ of $\K2$ is always the same for this $n^2+3
\equiv 0 \bmod 4$ case, despite the existence of a total of $10$
possible such fields \cite[Table~1.1]{jonrob} (and indeed $59$
possible quartic extensions of $\Q_2$).

\begin{proposition}
The set $\LL$ contains exactly one field, up to isomorphism, which we
shall write as $\L2$. 
A
standardised \cite{jonrob} generating polynomial for $\L2$ over
$\Q_2$ is
\[
X^4+2X^2+4X+4 .
\]
\end{proposition}

In other words, taking any $d$ of the form $n^2+3$ for odd $n$ with
its associated $D$ and completing the extension $K(\qu)$ of $K =
\QD$ at the unique prime above $2$ always results in a field
isomorphic to exactly the same quartic extension $\L2$ of $\Q_2$.  
\begin{proof}
We briefly re-establish the somewhat more specific notation
of~\cite{ABGHM}: the invariant $\xi_\ell = \sqrt{-2-\sqrt{d_\ell+1}}$
is associated to a given $d_\ell(D)$, valid for all $d_\ell$ of the
form $n_\ell^2+3$.  We omit the $\ell$ subscripts from now on.  When
we write $\xi(d)$ below it shall refer to $\xi_\ell(D) =
\sqrt{-2-\qd1}$ where $d = d_\ell(D)$.

Let $d = m^2+3$ and $e = n^2+3$ be two
dimensions in our sequence, with $m,n$ odd.  Since $\K2$ is the unique
unramified quadratic extension of $\Q_2$, it follows by Kummer theory
that $(d+1)/(e+1)$ is a square in $\Q_2$.  Moreover it is a $2$-adic
unit as the numerator and denominator are both congruent to $5 \bmod
8\Z_2$ and so their ratio is congruent to $1\bmod8\Z_2$.  So up to a factor 
of $\pm1 \in \Z_2^\times$, we may
define a unit $\eta = \eta_{d,e} \in \Z_2^\times$ by setting
\[
\qe1 = \eta_{d,e} \qd1. 
\]
We shall henceforth choose $\eta_{d,e}$ such that $\eta_{d,e} \equiv 1
\bmod 4$; the other choice $\eta_{d,e} \equiv -1 \bmod 4$ would merely
mean interchanging the roles of $(\eta_{d,e} \pm 1)$ below.

Hence if we fix some appropriate $d$, then all of the quartic
extensions of $\Q_2$ in $\LL$ may be obtained by extending $\K2$ by
the (Eisenstein) minimal polynomial of the uniformiser $\xi(e)-1$ from
Proposition \ref{ram2}, which is
\begin{equation}\label{L2d}
m_e(X) = X^2 + 2X + 3 + \qe1 = X^2 + 2X + 3 + \eta_{d,e}\qd1,
\end{equation}
as $e$ ranges over the even dimensions of the form $n^2+3$.

To complete the picture somewhat, we mention that 
the minimal polynomial $ X^4 + 4X^2 - n^2 $ 
of $\xi(e)$ over $\Q_2$ for one $e=n^2+3$
generates a non-normal quartic extension 
whose normal closure is octic with Galois group $D_4$. 
A generating polynomial for the whole extension 
could be taken to be $X^8 - (n^2+2) X^4 + 1$: 
see equations~(14) and~(41) of~\cite{ABGHM}. 
In this form we could
solve this as an instance of Krasner's lemma for polynomials; but
finding the appropriate value of~{Krasner's constant} boils down
to the same exercise which we do explicitly below.

So fix $d$ and let $\L2$ be the extension of $\K2$ obtained from
eq.~\eqref{L2d}, with $e=d$.  Let $\j2$ be the maximal ideal of its
ring of integers $\Z_{\j2}$.  We write the corresponding discrete
valuation as $v_{\j2}$, normalised so that $v_{\j2}(2) = 2$, following
Proposition~\ref{ram2}.  If we have a distinct dimension $e \neq d$,
therefore, if we show that the quadratic $m_e(X) = X^2 + 2X + 3 +
\qe1$ has two distinct roots $1\pm\sqrt{-2-\qe1}$ in $\L2$ then we
shall have finished the proof, since $d$, $e$ were arbitrary.

By Hensel's lemma~\cite[I, \S5]{koblitz} if we can find some
$\alpha\in \L2$ such that
\begin{equation}\label{hens}
v_{\j2}(m_e(\alpha)) > 2v_{\j2}(m_e'(\alpha)), 
\end{equation}
where $m_e'(X) = 2X+2$ is the derivative with respect to $X$, 
then we shall be able to lift $\alpha$ to a solution in $\L2$. 
We show that this holds if we choose $\alpha = \xi(d)-1$. 
Firstly, 
\[
  m_e(\xi(d)-1) = \qe1 - \qd1  =  ( \eta_{d,e} - 1 ) \qd1.
\]
Now $\qd1$ is a $2$-adic unit, since its square is $1 \bmod 4$.
Hence 
\begin{equation}\label{v2al}
 v_{\j2}(m_e(\xi(d)-1))  =    v_{\j2}(\eta_{d,e} - 1). 
\end{equation}
On the other hand, once 
again since $\xi(d)$ is a $2$-adic unit: 
\begin{equation}\label{v2deriv}
v_{\j2}(m_e'(\xi(d)-1))  =   v_{\j2}( 2\xi(d))  =   v_{\j2}(2)  =  2.
\end{equation}
So we are reduced to considering the possible valuations of the numbers $\eta_{d,e}-1$. 
But noting by the choice we made above for $\eta_{d,e}$, $v_{\j2}(\eta_{d,e}+1)  =  2$: 
\begin{equation}\label{vsep}
v_{\j2}(\eta_{d,e}-1) = v_{\j2}(\eta_{d,e}^2-1) - v_{\j2}(\eta_{d,e}+1)   =  v_{\j2}(\eta_{d,e}^2-1) - 2 . 
\end{equation}
Writing $d+1 = 5+8r$, $e+1 = 5+8s$ for some $r,s\in\N$, (so $r = \frac{m^2-1}{8}$ and $s = \frac{n^2-1}{8}$) by definition: 
\[
\eta_{d,e}^2 - 1 = \frac{5+8s}{5+8r} - 1 =  1 + \frac{8(s-r)}{5+8r} - 1 = 2^3\frac{(s-r)}{5+8r}  = 2^3\frac{(n^2-m^2)}{8(5+8r)}   =   \frac{(n^2-m^2)}{(5+8r)} , 
\]
where the denominator is a $2$-adic unit, and so finally, by
eq.~\eqref{vsep}:
\begin{equation}\label{vfly}
v_{\j2}(\eta_{d,e}-1) =  v_{\j2}(\eta_{d,e}^2-1) - 2  =  v_{\j2}{(n^2-m^2)} - 2.
\end{equation}
Writing this out in the terms from inequality \eqref{hens} with
$\alpha = \xi(d)-1$ and using eqs.~\eqref{v2al}, \eqref{v2deriv} and
\eqref{vsep}:
\[
v_{\j2}(m_e(\alpha)) = v_{\j2}(\eta_{d,e}-1) =  v_{\j2}{(n^2-m^2)} - 6 + 2v_{\j2}(m_e'(\alpha))
\]
So by inequality \eqref{hens}, we just need to show that for any given
$e = n^2+4$, there exists a $d = m^2+4$ such that $v_{\j2}(n^2-m^2) >
6$, which is to say, back in ordinary rational integers,
that $v_2(n^2-m^2) > 3$.

There are two `orbits', depending upon the relative residue classes of
$m,n \bmod 4$.  First, set $d=4$, so $D=5$, $m=1$ and $r=0$.  Then we
see that \emph{whenever $s$ is even}, which translates to $n^2+4 = e+1 \equiv
5 \bmod 16$, our condition is satisfied.  On the other hand, set
$d=12$, so $D=13$, $m=3$ and $r=1$: \emph{whenever $s$ is odd}, we
have $ v_{\j2}(s-r) > 0 $ and so again if $n^2+4 = e+1 \equiv 13 \bmod 16$
then we have satisfied our condition.  So provided that we can show
that the two cases $D=5$ and $D=13$ themselves lead to the same
extension, then we shall have proven it for every $D \equiv 5 \bmod 8$
arising in this way.  But this is easily calculated by hand.
\end{proof}

\section{Two remarks on the monomial representation}\label{sec:phase_ambiguity}
In this Appendix we will show that the form of the Ansatz, as spelt
out in Section \ref{sec:ansatz}, can be derived by assuming that a
real SIC fiducial vector exists. However, first we will deal with an
issue that arose in Section \ref{sec:Clifford}. There we came
across a list of three two-fold ambiguities that arise when
representing the Weyl--Heisenberg group (the WH-group) and the
Clifford group in $d = 4$:

\begin{itemize} 
\item The Clifford group contains two distinct copies of the WH-group.
\item There are two ways to enphase the basis vectors of the monomial
  representation so that a given Zauner unitary is a permutation
  matrix.
\item Given Zauner unitaries in the factors we can form two Zauner
  unitaries $U_{\cal Z}^{(4)}\otimes U_{\cal Z}^{(p)}$ and
  $\bigl(U_{\cal Z}^{(4)}\bigr)^2\otimes U_{\cal Z}^{(p)}$ creating
  two distinct Zauner subspaces.
\end{itemize}
Purely for the purpose of this Appendix we define a {\it minimal
  fiducial} as a SIC vector of the form $({\bf v}^{\rm T}_1, {\bf
  v}^{\rm T}_2, {\bf v}^{\rm T}_3, {\bf v}^{\rm T}_4)^{\rm T}$, where
$\{ {\bf v}_j\}_{j=1}^4$ are vectors in $\C^p$ constructed using the
field $K^{4\pp\jj}$. We also define a {\it non-minimal fiducial} as a
SIC vector of the form $({\bf v}^{\rm T}_1, {\bf v}^{\rm T}_2, {\bf
  v}^{\rm T}_3, i{\bf v}^{\rm T}_4)^{\rm T}$, where $i$ is a fourth
root of unity. We claim:
\begin{quote}
  {\sl If, with a definite choice of enphasing for the basis vectors,
    there is a minimal fiducial in one of the Zauner subspaces for one
    of the WH-groups, then the other subspace contains a non-minimal
    fiducial for the same WH-group, and a minimal fiducial for the
    other WH-group. If we choose the other enphasing the roles of the
    two subspaces are switched.}
\end{quote}
Thus, making two definite choices in the above list will force the 
third if we insist on having a SIC fiducial vector that is 
completely free of roots of unity. 

We start with the second item on the list. Restricting ourselves to
using fourth roots of unity only we can change the representation used
in Section \ref{sec:Clifford} by applying the diagonal matrix
\begin{equation}
  D_{\rm H} = \mbox{diag}(1,1,1,i).
\end{equation}
Making this change of basis will not affect the Zauner permutation
matrix, but it will turn a minimal fiducial into a non-minimal
one. Replacing $i \rightarrow -i$ is irrelevant because it can be
undone by means of a Clifford transformation.

Consider the basis as fixed. There is a Clifford unitary transforming 
$U_{\cal Z}^{(4)}$ into $\bigl(U_{\cal Z}^{(4)}\bigr)^2$, namely 
\begin{equation}
  M = \left( \begin{array}{rr}
    2 & 1 \\
    3 & - 2
  \end{array} \right)_8 
   \quad\Longrightarrow\quad
 U_M = \begin{pmatrix}
   1 & 0 & 0 & 0 \\
   0 & 0 & 1 & 0 \\ 
   0 & 1 & 0 & 0 \\
   0 & 0 & 0 & -i
 \end{pmatrix}
   \quad\Longrightarrow\quad
   U_MU_{\cal Z}U_M^{-1} = U_{\cal Z}^2. \label{UM} 
\end{equation}
Notice the position of the $i$. If we apply the transformation
$U_M^{(4)}\otimes \openone_p$ in dimension $d = 4p$ it will transform
a minimal fiducial in one of the two Zauner subspaces into a
non-minimal fiducial in the other. If we change the enphasing using
$D_{\rm H}$ the roles of the two Zauner subspaces are switched because
minimal fiducials are changed into non-minimal fiducials, and
conversely.

Now we come to the first item on the list. We reinterpret the matrix $D_{\rm H}$ as providing 
a transformation in a fixed basis. It turns out to be an element of the third level 
of the Clifford hierarchy, defined as the collection of unitary operators that transforms 
the Weyl--Heisenberg group into the Clifford group. To be precise about it, it can be 
shown that 
\begin{equation}
  \tilde{X} \equiv D_{\rm H}XD_{\rm H}^{-1} = D_{1,0}U_{G_1},
  \qquad
  \tilde{Z} = D_{\rm H}ZD_{\rm H}^{-1} = D_{0,3}U_{G_2}, 
\end{equation}
 where 
 \begin{equation}
   G_1 = \left( \begin{array}{cc} 3 & 2 \\ 0 & 3 \end{array} \right)_8
   \qquad\text{and}\qquad 
   G_2 = \left( \begin{array}{cc} 3 & 0 \\ 2 & 3 \end{array} \right)_8. 
\end{equation}
Clearly 
\begin{equation}
  \tilde{Z}\tilde{X} = \omega \tilde{X}\tilde{Z},
\end{equation}
so this is another copy of the Weyl--Heisenberg group sitting inside
the Clifford group. Notice that $\tilde{X}^2 = X^2$ and $\tilde{Z}^2 =
Z^2$, so the two WH-groups intersect. Notice also that $\tilde{X}$ is
a Hadamard matrix in the standard representation of the original
WH-group, so at this point the monomial representation is helpful.

The two WH-groups are behind the regrouping phenomena observed for
SICs in $d = 4$ \cite{Huangjun}. For us it is relevant that if there
is a minimal fiducial with respect to one of the WH-groups in one of
the Zauner subspaces, then one can remove the $i$ from the non-minimal
fiducial in the other subspace so that a minimal fiducial with respect
to the other WH-group is obtained.

Finally, let us see why the form of our Ansatz in Section
\ref{sec:ansatz} is equivalent to the assumption that real SIC
fiducial vector exists. In the standard representation (which we do
use in the dimension $p$ factor of the Hilbert space) we have built in
the anti-unitary symmetry
\begin{equation}
  U_P^{(p)}{\bf v}_1 = {\bf v}_1^*, \quad
  U_P^{(p)}{\bf v}_2 = {\bf v}_2^* , \quad
  U_P^{(p)}{\bf v}_3 = {\bf v}_3^* , \quad
  U_P^{(p)}{\bf v}_4 = - {\bf v}_4^*,
\end{equation}
see eq.~\eqref{UPd}. In the standard representation of the Clifford group (in any dimension, but let us 
set the dimension equal to $p$) we find 
\begin{equation}
  {\cal F} = \left( \begin{array}{rr}
    0 & -1 \\
    1 & 0
  \end{array} \right)_{p}
  \qquad \leftrightarrow \quad
  \bigl(U_{\cal F}\bigr)_{r,s} = \frac{1}{\sqrt{p}}\omega^{rs}, \quad
  U_{\cal F}^2 = U_P, \quad
  U_{\cal F}^4 = \openone_p.
\end{equation}
The unitary matrix $U_{\cal F}$ is the Fourier matrix, effecting the
discrete Fourier transform (DFT) and obeying $U_{\cal F}^{*} = U_{\cal
  F}^3$. Suppose we have two vectors related by ${\bf v}_{\rm R} =
U_{\cal F}{\bf v}_{\rm C}$. Then
\begin{equation}
  {\bf v}_{\rm R}^* = {\bf v}_{\rm R}
    \quad\Longleftrightarrow\quad
  U_{\cal F}^*{\bf v}_{\rm C}^* = U_{\cal F}{\bf v}_{\rm C}  
    \quad\Longleftrightarrow\quad
    U_{\cal F}^3{\bf v}_{\rm C}^* = U_{\cal F}{\bf v}_{\rm C}
    \quad\Longleftrightarrow\quad
    U_P{\bf v}_{\rm C}^* = {\bf v}_{\rm C}.
\end{equation}
This implies that applying the DFT to ${\bf v}_1,{\bf v}_2,{\bf v}_3$
results in real vectors, while applying the DFT to ${\bf v}_4$ results
in an imaginary vector. This is where the dimension four Clifford
unitary $U_M$, defined in eq.~\eqref{UM}, comes in. We see that for a
vector $|\Psi_0\rangle$ having the symmetries postulated in our Ansatz
the vector
\begin{equation}
  \ket{\Psi_{\rm R}} = (U_M^{(4)}\otimes U_{\cal    F}^{(p)})\ket{\Psi_0} \label{CtoR}
\end{equation}
is indeed real. 

Now suppose that $\ket{\Psi_{\rm R}}$ is a SIC vector. For the sake of
transparency we first give the argument as it applies in prime
dimensions, with the added bonus that we can correct an oversight in
Ref.~\cite{ABGHM}.  Also there we have a real vector connected to a
complex vector by the DFT, so that the complex vector enjoys an
anti-unitary symmetry. For the real vector to be a SIC vector it must
hold, for $j \neq 0$, that
\begin{equation}
  \bra{\Psi_{\rm R}}X^j\ket{\Psi_{\rm R}} = \frac{1}{\qd1} .\label{eq:real_baby}
\end{equation}
In Ref.~\cite{ABGHM} it was stated that the sign on the right hand
side is a free choice, but that is not the case.  First note that in
our Ansatz, the Galois group and the pre-ascribed Zauner symmetry act
transitively on the overlaps on the left-hand side of
eq.~\eqref{eq:real_baby}, and hence the right-hand side is independent
of $j\ne 0$.  Second, the averaged auto-correlation $\bra{\Psi_{\rm
    R}}X^j\ket{\Psi_{\rm R}}$ of a real unit vector is bounded from
below by $-1/(d-1)$, since the sum over the left-hand side of
eq.~\eqref{eq:real_baby}, including $j=0$, equals the square of the
sum of the coefficients of the unit vector and is hence
non-negative. We then observe that
\begin{equation}
  \ket{\Psi_{\rm C}} = U_{\cal F}\ket{\Psi_{\rm R}}
    \quad \Longrightarrow \quad  
  \bra{\Psi_{\rm C}}Z^j\ket{\Psi_{\rm C}} = \bra{\Psi_{\rm R}}X^j\ket{\Psi_{\rm R}}.
\end{equation}
Denoting the components of the complex vector by $c_r$ this means that 
\begin{equation}
  \sum_{r=0}^{d-1}|c_r|^2\omega^{rj} =  \bra{\Psi_{\rm R}}X^j\ket{\Psi_{\rm R}},
    \quad\text{which is equivalent to\ }
  |c_k|^2 = \frac{1}{d}\sum_{j=0}^{d-1}\omega^{-kj} \bra{\Psi_{\rm R}}X^j\ket{\Psi_{\rm R}}. 
\end{equation}
The equivalence follows from the invertibility of the DFT, this time
applied to the absolute values squared of the components. We can
compute the latter because we know the values of $\bra{\Psi_{\rm
    R}}X^j\ket{\Psi_{\rm R}}$, and we find that $\ket{\Psi_{\rm C}}$
is almost flat.

The argument in the $d = 4p$ case is similar, but more involved. The
real SIC vector must obey
\begin{equation}
\bra{\Psi_{\rm R}}\openone_4\otimes X^j\ket{\Psi_{\rm R}}
 = \begin{cases}
  \hfil 1, & \text{if $j = 0$;} \\
  \frac{1}{\qd1} & \text{if $j \neq  0$,}
 \end{cases}\label{eq:real_baby1}
\end{equation}
\begin{equation}
    \bra{\Psi_{\rm R}}D_{0,2}\otimes X^j\ket{\Psi_{\rm R}}
  = \bra{\Psi_{\rm R}}D_{2,0}\otimes X^j\ket{\Psi_{\rm R}}
  = \bra{\Psi_{\rm R}}D_{2,2}\otimes X^j\ket{\Psi_{\rm R}}
  = - \frac{1}{\qd1}. \label{eq:real_baby2}
\end{equation}
Using the relation \eqref{CtoR} we deduce that $\ket{\Psi_{0}}$ obeys
the same equations, but with $X$ replaced by $Z$. Since all the
displacement operators in eqs.~\eqref{eq:real_baby1} and
\eqref{eq:real_baby2} are diagonal in the dimension four factor we can
write $\ket{\Psi_{0}}$ as a direct sum,
\begin{equation}
    \ket{\Psi_{0}} = \sum_{i=1}^4 \ket{i}\ket{{\bf u}_i},
\end{equation}
as in eq.~\eqref{eq:fid_ansatz}, and then solve the resulting system of equations for 
\begin{equation}
      \bra{{\bf u}_1}Z^j\ket{{\bf u}_1}
    = \bra{{\bf u}_2}Z^j\ket{{\bf u}_2}
    = \bra{{\bf u}_3}Z^j\ket{{\bf u}_3}
    = \begin{cases}
        \frac{\qd1 - 1}{4\qd1},& \text{if $j = 0$;} \\
        \hfil 0, & \text{if $j \neq  0$,}
    \end{cases}
\end{equation}
\begin{equation}
  \bra{{\bf u}_4}Z^j\ket{{\bf u}_4}
    = \begin{cases}
        \frac{\qd1 + 3}{4\qd1}, & \text{if $j = 0$;} \\
        \hfil\frac{1}{\qd1}, & \text{if $j \neq  0$}.
      \end{cases}
\end{equation}
These equations can again be solved for the absolute values squared of
the components by inverting a DFT, and we recover precisely the form
of our Ansatz from the assumption that a real SIC vector exists.

\section{Two remarks on verification}\label{sec:Gik}
In Ref.~\cite{ABGHM} we constructed SICs in dimensions $d = n^2+3 =
p$, and verified the SIC property by checking the equations for the
$G(i,k)$, as defined in eqs.~\eqref{Gik1}--\eqref{Gik2} above. For the
higher dimensions we had to resort to numerical checks due to the time
it takes to calculate a single $G(i,k)$ exactly.  However, it has been
conjectured that it is enough to calculate only $3d$ out of the $d^2$
conditions \cite{ADF, FHS}. Here we will prove a somewhat sharper
result for the special case that the fiducial vector obeys our Ansatz,
which in the prime dimensional case means that the vector is built
from a Galois orbit of $(p-1)/3\ell$ rescaled Stark units in a number
field of degree $h_K(d-1)/3\ell$ over $K$.

First, assume $d = p$. Since some of the conditions are built into the
Ansatz, and because the Ansatz implies that for the components $a_r$
of the the fiducial vector $a_r^* = a_{-r}$, it is enough to verify
that
\begin{equation}
  G(i,k) = \sum_{r=0}^{p-1}a_{-r-i}a_{-r-k}a_ra_{r+i+k} = 0, 
  \qquad
  1 \leq i,k \leq d-1. \label{Gcondition}
\end{equation}
By inspection we see that 
\begin{alignat}{5}
  G(i,k) = G(i,-k)^* &{}= G(-i,k)^* = G(-i,-k) \label{Geqs1} \\
  G(i,k) &{}= G(k,i). \label{Geqs2}
\end{alignat}
Moreover, due to the Ansatz, all the $G(i,k)$ are in fact real, 
$G(i,k)^* = G(i,k)$. The Ansatz further implies that there is a Galois transformation, 
cyclic of order $(p-1)/3\ell$, such that 
\begin{equation}
  \sigma (a_r) = a_{\theta r},
\end{equation}
where $\theta$ is a generator of the integers modulo $p$. Again by
inspection we see that this implies 
\begin{equation}
  G(i,k)^\sigma = G(\theta i, \theta k). 
\end{equation}
To check that eqs.~\eqref{Gcondition} hold it is enough to check one
representative from each Galois orbit. To count the number of
conditions that need checking we first note that the diagonal elements
$G(i,i)$ form a Galois orbit of their own. Because of
eqs.~\eqref{Geqs1} checking one representative will actually account
for $2(p-1)$ conditions. For the off-diagonal elements,
eq.~\eqref{Geqs2} also becomes operative, so checking one
representative will account for $4(p-1)$ conditions. The total number
of conditions that need checking is easily calculated from this
information. It is
\begin{equation}
  \frac{d+1}{4}.\label{eq:d_quarter}
\end{equation}
This is a significant improvement. Taking $d = 5779$ as an example,
the exact calculation of a single $G(i,k)$ took $88$ minutes.  With
the new result in hand the estimated total time needed for exact
verification drops from computing all $d^2$ values of $G(i,k)$ in
about $5500$ years to about $3$~months for the $(d+1)/2$ overlaps
stated in \eqref{eq:d_quarter}.  If instead we consider the overlaps
as they stand the number field is much larger, but so is the Galois
group. The number of conditions that need checking then drops to two,
for any prime $d$. For $d = 5779$ we have done the exact verification
in this way.  That calculation took about six days, most of which was
spent computing the absolute value squared of one of the two overlaps.

When the dimension is even, a non-standard basis in Hilbert space is
needed in order to ensure that the SIC vector sits in a small number
field. This complicates the derivation of the $G(i,k)$.  In the main
text we avoided this by transforming to the standard basis, thus
increasing the degree of the number field by a factor of four.
There is an alternative way to proceed, which we will now sketch.

Assume that $d = 4p$, and that the fiducial vector is constructed according 
to the Ansatz given in the main text. Let $D$ be any displacement operator 
in the dimension four factor, and define 
\begin{equation}
  G(D; i,k) = \frac{1}{p}\sum_{j=0}^{p-1}\omega^{jk} 
  |\bra{\psi} D\otimes X^iZ^j\ket{\psi} |^2,
\end{equation}
where $\omega$ is a primitive complex $p$th root of unity.  The $p$th
roots of unity drop out from the formula, and so do the eighth roots
of unity from $D$, but factors of $i$ are still present inside
$D$. The SIC conditions read
\begin{alignat}{5}
 &&  G(\openone; i,k) &{}=\frac{4\delta_{i,0}+\delta_{k,0}}{d+1}\\
\text{and}\quad&& G(D; i,k) &{}= \frac{\delta_{k,0}}{d+1},
   \qquad\text{for $D \neq\openone$}.
\end{alignat}
In this way a complete verification of the SIC property can be carried
out in a number field of degree twice that in which the fiducial
vector sits. Using the unitary symmetries in the dimension four
factor, and dividing the conditions into Galois orbits, one can show
that it is sufficient to check $(3p+5)/2 = (3d+20)/8$ conditions. But
the approach using overlap phases directly, taken in Section
\ref{sec:verification}, is much faster. Still, in Step 6 of the
algorithm in Section \ref{sec:results}, we can distinguish between the
two possibilities for the symmetry by calculating $G(Z;1,1)$. Because
the number field is slightly smaller this gives a slight speed-up
compared to the approach in Section \ref{sec:results}.

\bibliographystyle{unsrturl}
\bibliography{SIC-POVMs_from_Stark_Units}

\end{document}